\documentclass{article}

\usepackage[utf8]{inputenc}
\usepackage{amsmath}
\usepackage{amsthm}
\usepackage{amssymb}
\usepackage[backgroundcolor=gray!20, linecolor=gray, textsize=scriptsize]{todonotes}

\usepackage{hyperref}
\usepackage[capitalize, nameinlink]{cleveref}
\crefname{theorem}{Theorem}{Theorems}
\Crefname{lemma}{Lemma}{Lemmas}
\Crefname{conjecture}{Conjecture}{Conjectures}
\crefname{section}{Section}{Sections}

\usepackage{microtype}
\usepackage{graphicx}
\usepackage{booktabs} % for professional tables
\usepackage{algorithm}
\usepackage{algorithmic}
\usepackage{thmtools}
\usepackage{thm-restate}
\usepackage{caption}
\usepackage{subcaption}
\usepackage[toc, page]{appendix}

\newtheorem{theorem}{Theorem}
\newtheorem{lemma}{Lemma}
\newtheorem{corollary}{Corollary}
\newtheorem{definition}{Definition}

\newtheorem{remark}{Remark}
\newtheorem{fact}{Fact}

\crefname{claim}{Claim}{Claims}
\crefname{fact}{Fact}{Facts}
\crefname{example}{Example}{Examples}

\DeclareMathOperator*{\argmin}{arg\,min}
\newcommand{\E}{\textrm{\textbf{E}}}

\newcommand{\e}{\textrm{{e}}}
\newcommand{\eps}{\varepsilon}

\newcommand{\OPT}{\textrm{OPT}}
\newcommand{\opt}{\textrm{OPT}}

\renewcommand{\phi}{\varphi}

\newcommand{\fA}{\mathcal{A}}

\newcommand{\fC}{\mathcal{C}}
\newcommand{\fD}{\mathcal{D}}

\newcommand{\fM}{\mathcal{M}}
\newcommand{\poly}{\mathrm{poly}}
\renewcommand{\lg}{\log}
\newcommand{\km}{$k$-means\,}

\newcommand{\kmpp}{\texttt{$k$-means$++$\,}}
\newcommand{\kmpipe}{\texttt{$k$-means$||$\,}}
\renewcommand{\P}{\textrm{P}}
\newcommand{\ain}{A_{\textrm{in}}}
\newcommand{\aout}{A_{\textrm{out}}}
\newcommand{\lsp}{\texttt{Local-search++ } }
\newcommand{\nlsp}{\texttt{Local-search++ } with outliers }
\newcommand\cost[2]{\tau(#1, #2)}%our new cost
\newcommand{\cst}{\tau}

\newcommand{\ass}{\textrm{reassign}}

\newcommand{\out}{\textrm{outliers}}
\newcommand{\inn}{\textrm{in}}
\newcommand{\outt}{\textrm{out}}
\newcommand{\OPTG}{\textrm{OPT}^\textrm{guess}}
\newcommand{\Xin}{X_{in}}
\newcommand{\Xout}{X_{out}}
\newcommand{\kmc}{$\textrm{AFK-MC}^2$}
\emergencystretch=\maxdimen
 \author{Christoph Grunau\\ETH Zürich \and Václav Rozhoň\\ETH Zürich }
\title{Adapting $k$-means algorithms for outliers}
\date{March 2020}

\begin{document}

\maketitle

\begin{abstract}
 
 This paper shows how to adapt several simple and classical sampling-based algorithms for the $k$-means problem to the setting with outliers.

 Recently, Bhaskara et al. (NeurIPS 2019) showed how to adapt the classical $k$-means++ algorithm to the setting with outliers. However, their algorithm needs to output $O(\log (k) \cdot z)$ outliers, where $z$ is the number of true outliers, to match the $O(\log k)$-approximation guarantee of $k$-means++. In this paper, we build on their ideas and show how to adapt several sequential and distributed $k$-means algorithms to the setting with outliers, but with substantially stronger theoretical guarantees:
our algorithms output $(1+\eps)z$ outliers while achieving an $O(1 / \eps)$-approximation to the objective function.  
 In the sequential world, we achieve this by adapting a recent algorithm of Lattanzi and Sohler (ICML 2019). In the distributed setting, we adapt a simple algorithm of Guha et al. (IEEE Trans. Know. and Data Engineering 2003) and the popular $k$-means$\|$ of Bahmani et al. (PVLDB 2012).
 
 A theoretical application of our techniques is an algorithm with running time $\tilde{O}(nk^2/z)$ \footnote{$\tilde{O}(.)$ hides logarithmic factors in $n$. We assume that the minimum distance between any two points is $1$ and the maximum distance $\Delta$ between any two points is bounded by $\poly(n)$.} that achieves an $O(1)$-approximation to the objective function while outputting $O(z)$ outliers, assuming $k \ll z \ll n$. 
 This is complemented with a matching lower bound of $\Omega(nk^2/z)$ for this problem in the oracle model. 
 
\end{abstract}

\section{Introduction}
\label{sec:paper_intro}

Clustering is a fundamental tool in machine learning and data analysis. It aims to partition a given set of objects into \emph{clusters} in such a way that similar objects end up in the same cluster. 
The classical way of approaching the clustering problem is via the \km formulation. 
In this formulation, one works with a set $X$ consisting of $n$ points in the $d$-dimensional Euclidean space $\mathbb{R}^d$ and the objective is to output a set $C$ consisting of $k$ centers so as to minimize the sum of squared distances of points in $X$ to $C$, i.e.,
\[
\phi(X, C) = \sum_{x \in X} \min_{c \in C} \Vert x - c \Vert^2.
\]
Assigning each point to its closest cluster center naturally induces a partition of the points where nearby points tend to end up in the same partition.  

\paragraph{\km with Outliers} One major drawback of \km in practice is its sensitivity to outliers \cite{gupta2017local}. 
This motivates optimizing a more robust version of the \km objective.
Probably the simplest such formulation is the \emph{\km with outliers} formulation \cite{charikar2001algorithms}. In this formulation, one additionally receives a number $z \in \{0,1,\ldots,n\}$. The aim is to find a set $C$ consisting of $k$ cluster centers and additionally a set $\Xout \subseteq X$ consisting of $z$ \emph{outliers} so as to optimize the cost of \emph{inliers} $\Xin := X \setminus \Xout$. 
That is, we optimize 
$
\min_{\Xout, C} \phi(X \setminus \Xout, C). 
$ 

\paragraph{Known Results for \km with Outliers} It is known that the problem can be approximated up to a constant factor in polynomial time by rounding a certain linear program \cite{chen2008constant, Krishnaswamy2018}. However, the complexity of these known algorithms is a large polynomial. Moreover, as we observe in \cref{sec:paper_nk2z}, every constant factor approximation algorithm that outputs \emph{exactly} $z$ outliers needs to perform $\Omega(z^2)$ queries in the  query model.

This motivates to weaken our requirements and look for fast algorithms that are allowed to output slightly more than $z$ outliers. 
There is a line of work that makes progress toward such algorithms  \cite{charikar2001algorithms, meyerson2004k, gupta2017local, Guha_distributed, bhaskara2020, nkmeans}. However, to the best of our knowledge, known algorithms either need to output $A z$ outliers, for some large constant $A \gg 1$ to obtain a reasonable approximation guarantee, or they suffer from at least a $\Omega(z^2)$ running time and, hence, to be truly applicable even for large $z$ they need to be sped up by coreset constructions \cite{feldman2011unified, huang2018epsilon}. %  

\paragraph{Our Contribution} In this work we aim to further close the gap between theory and practice. 
We show how to adapt a number of classical sampling-based $k$-means algorithms to the setting with outliers. The main idea of the adaptations is a reduction to the variant of the problem with penalties described below. This is a known approach \cite{charikar2001algorithms}, but previous reductions of sampling based algorithms \cite{bhaskara2020} necessarily need to output $Az$ outliers for some large constant $A \gg 1$, while  we obtain algorithms that need to output only $(1+\eps)z$ outliers. 
More concretely, our contribution is threefold. 

First, building on ideas of \cite{lattanzi2019better}, we design a simple sampling-based algorithm with running time $\tilde{O}(nk/\eps)$. It outputs an $O(1/\eps)$-approximate solution while declaring at most $(1+\eps)z$ points as outliers. 
This shows that sampling based algorithms that are known to be fast and have good practical performance can also achieve strong theoretical guarantees. 

%In the natural setting where $k$ is small and $z = \Theta(n)$,  constant-factor approximation algorithms outputting exactly $z$ outliers has a running time of $\Omega(n^2)$, while our algorithm runs in  near linear time.\todo{rephrase}  %We observed on two datasets that our algorithm outperformed the algorithm of \cite{bhaskara2020} by around $10\%$. 

Second, we devise distributed algorithms for $k$-means with outliers where each machine needs to send only $\tilde{O}(k/\eps)$ bits. Our construction achieves an $O(1)$-approximation while outputting $(1+\eps)z$ outliers. Moreover, each machine only needs to perform polynomial-time computation. This improves on \cite{guo2018} who achieve an $(1+\eps)$-approximation while outputting the same number of outliers as our algorithm, but their computation time is exponential. 

Third, we show that one can achieve an $O(1)$-approximation guarantee while discarding $O(z)$ outliers in time $\tilde{O}(k^2 \cdot n/z )$. %The algorithm does not explicitly label the $O(z)$ outliers, as this is not possible for large $z$.\mtodo{Again, needs clarification. I am guessing that you are giving a disclaimer that this algorithm does not explicitly declare outliers (say what it does instead) and then your "not possible" needs a forward reference with a clear statement of what it means, e.g., quadratic time lower bound?} 
This is done by speeding-up sampling with the Metropolis-Hastings algorithm as done in \cite{bachem2016approximate, bachem2016fast} together with additional ideas. 
%The same bound is also shown\mtodo{"is shown", by whom? do you mean "we show"? avoid the phrasings where it is not clear where or by whom. This is an archaic style of English; was very common before 2000s, but is no longer recommended in the scientific literature due to the confusions that it creates} 
This result is complemented by a matching lower bound of $\Omega(k^2 \cdot (n/z) )$ for $k$-means/$k$-median/$k$-center algorithms that work for an arbitrary metric space accessed by distance queries. % (Almost all sampling-based algorithms also work in this setting). 
This improves on \cite{meyerson2004k,huang2018epsilon} who give algorithms in this setting with a running time of $\tilde{O}\left(  k^2 \cdot \left( n / z \right)^2\right)$. 
This is a significant improvement for $z \ll n$. 
%, hence follow-up $\tilde{O}(nk)$ algorithm such as our algorithm mentioned earlier would have complexity $\tilde{O}\left(\left( \frac{n}{z}\right)^2 k^2\right)$. 

\paragraph{Roadmap.} We overview the related work in \cref{sec:paper_previous_work} and our approach in \cref{sec:paper_warmup}. 
Then, we present our contributions in \Cref{sec:paper_localsearch,sec:paper_distributed,sec:paper_nk2z} and our experiments in \cref{sec:paper_experiments}. Technical details are mostly deferred to appendices. 

%Before we explain our results, we start to define commonly used notation.

\paragraph{Notation}
We define $\phi(x, C) = \min_{c \in C} \Vert x - c \Vert^2$
and set $\phi(X, C) = \sum_{x \in X} \phi(x, C)$. 
Similarly, we define $\tau_\Theta(x, C) = \min(\Theta, \phi(x, C))$ and $\tau_\Theta(X, C) = \sum_{x \in X} \tau_\Theta(x, C)$. 
We call an algorithm an $(\alpha, \beta)$-approximation if it outputs a set of $k$ centers $C$ and a set of $z$ outliers $X_{out}$ such that $\phi(X \setminus X_{out}, C) \le \alpha \OPT = \alpha \phi(X \setminus X_{out}^*, C^*)$, where $\OPT$ is the cost of a fixed optimal clustering $C^*$ with a set of outliers $X_{out}^*, |X_{out}^*| = z$. 
Moreover, we define $X_{in} := X \setminus X_{out}$ and $X^*_{in} := X \setminus X^*_{out}$.

\section{Previous Work}
\label{sec:paper_previous_work}

\paragraph{\km Problem}
%There is a huge body of work concerned with the \km problem. 
It is well-known that finding the optimal solution is NP-hard \cite{aloise2009np, mahajan2009planar}. Currently the best known approximation ratio is roughly 6.36 \cite{ahmadian2019better} and an approximation ratio of $(1+\eps)$ can be achieved for fixed dimension $d$ \cite{friggstad2019local} or fixed $k$ \cite{kumar2004simple}. %\mtodo{dimension $d$ was not defined before so now you say something like "if they points are in a d-dimensional Euclidean space"}
From the practical perspective, Lloyd's heuristic \cite{lloyd1982least, lloyd_survey} is the algorithm of choice. 
As Lloyd's algorithm converges only to a local optimum, it requires a careful seeding to achieve good performance. 
The most popular seeding choice is the \kmpp seeding \cite{arthur2007k, ostrovsky2013effectiveness}. 
In \kmpp seeding (\cref{alg:kmp_original}) one chooses the first center uniformly at random from the set of inputs points $X$. 
In each of the following $k-1$ steps, one samples a point in $x \in X$ as a new center with probability proportional to its current cost $\phi(x, C)$. The seeding works well in practice and a theoretical analysis shows that even without running Lloyd's algorithm subsequently it provides an expected $O(\log k)$-approximation to the \km objective \cite{arthur2007k}. 

\begin{algorithm}[]
	\caption{\kmpp seeding}
	\label{alg:kmp_original}
	Input: $X$, $k$
	\begin{algorithmic}[1]
		\STATE Uniformly sample $c \in X$ and set $C = \{ c \}$.
		\FOR{$i \leftarrow 2, 3, \dots, k$}
		% \STATE Sample $p \in P$ with probability $\frac{cost(p, C)}{\sum_{q \in P} cost(q, C)}$ and add it to $C$.
		\STATE Sample $c \in X$ with probability $\phi(c, C) / \phi(X, C)$ and add it to $C$.
		%\STATE Sample $c \in X$ w.p. $\frac{\phi(c, C)}{\sum_{x \in X} \phi(x, C)}$ and add it to $C$.
		\ENDFOR
		\STATE \textbf{return} $C$
	\end{algorithmic}
\end{algorithm}

%\paragraph{Distributed \km}

%\todo{why are we explaining here?}
%In a seminal work, Guha et al. \cite{guha2003streaming} proved that a ``clustering of a clustering is a clustering'': if one finds a set of centers $Y$ such that $\phi(X, Y)$ is small, then one can ``move'' each point $x \in X$ to its closest center $y \in Y$ to obtain a multiset $X'$. For each set of centers $C$, the cost $\phi(X',C)$ is comparable to $\phi(X,C)$. As $X'$ can be represented efficiently as a weighted instance, the approach enables one to get simple distributed $O(1)$-approximation algorithms. 
%For example, a distributed version of \kmpp called \kmpipe \cite{bahmani2012scalable} works similarly as \cref{alg:kmp_original}, but in each step one samples $O(k)$ points instead of just one. 
%In the end, a solution with $k$ centers is recovered by the aforementioned approach of \cite{guha2003streaming}. 

\paragraph{\km with Outliers}

There is a growing body of research related to the \km with outliers problem. 
%Given an $n$-element point set $X$ and two integers $1 \le k \le n$ and $1 \le z \le n$ the goal is to find a set $\Xout \subseteq X$ of $z$ outliers and a set $C$ consisting of $k$ cluster centers so as to minimize $\phi(\Xin, C)$ where $\Xin := X \setminus \Xout$. 
On the practical side, Lloyd's algorithm is readily adapted to the noisy setting \cite{kmeans-}, but the output quality still remains dependent on the initial seeding. 
%We now review work that was done in the direction of algorithms for the problem with provable guarantees. 
On the theoretical side, constant approximation algorithms based on the method of successive local search \cite{chen2008constant, Krishnaswamy2018} are known to provide a constant approximation guarantee. 
However, their running time is a large polynomial in $n$. % and thus they are mainly of theoretical interest. 
%Indeed, in \cref{sec:lower_bound} we observe that constant factor approximation algorithms that work in any metric space have time complexity $\Omega(z^2)$. 
We are interested in fast algorithms that are allowed to output slightly more than $z$ outliers. Several algorithms have been proposed \cite{charikar2001algorithms, meyerson2004k, gupta2017local, Guha_distributed, bhaskara2020, nkmeans}, but they either need at least $\Omega(z^2)$ time  or they need to output at least  $Cz$ outliers for some $C \gg 1$. 
The algorithms can in general be sped up by coreset constructions \cite{feldman2011unified, gupta2017local, huang2018epsilon, nkmeans}. 
However, there is still need for fast and simple algorithms with strong guarantees. 

\paragraph{\km with (Uniform) Penalties}

\km with penalties is a different way of handling outliers introduced in the seminal paper of Charikar et al.~\cite{charikar2001algorithms}. 
In the version of the problem with uniform penalties, we are given some positive number $\Theta$ and the goal is to output a set of centers $C = \{c_1, c_2, \dots, c_k\}$ so as to minimize the expression $\tau_\Theta(X, C) = \sum_{x\in X} \min\left( \Theta, \min_{c\in C} \Vert x - c \Vert^2 \right)$. 
That is, the cost of any point is bounded by a threshold $\Theta$. 

It turns out that it is usually much simpler to work with \km with penalties than \km with outliers. 
This is quite helpful because results for \km with penalties can be turned into results for \km with outliers \cite{charikar2001algorithms, guo2018, bhaskara2020, bhaskara2020online}. 
We also take this approach in this paper. We formalize the reduction in the next section and also present our improvement of it for sampling based algorithms. 
%This is an exciting message, as the authors of  \cite{bhaskara2020} found a connection between the penalty version and the \km with outlier formulation that we describe more closely later. 

%They showed that one can in general turn any $\alpha$ approximation algorithm for the penalties formulation to an $O(\alpha)$ approximation algorithm for the outliers formulation, with the downside of outputting $O(\alpha z)$ outliers. 
%They apply this reduction to \kmpp (\cref{alg:kmp})  that readily generalizes (together with its analysis \cite{li_penalties2020}) to the setting with penalties. This gives their result mentioned above. 

%The known results for \km with penalties include generalization of \kmpp [] and local search []. 

\section{Warmup: Reducing \km with Outliers to \km with Penalties}
\label{sec:paper_warmup}
%\mtodo{The text in this section is disconnected. Give some structing text so the reader knows what he/she is reading next and what are the connections. For instance, after the ``Our approach is based on .... " sentence, I would say something like ``we first review their reduction and its instantiation for k-means++, and we then discuss our improved/generalized reduction."}

Our approach is based on reducing the problem of $k$-means with outliers to the problem of $k$-means with penalties. 
In this section, we first review how one can reduce the problem of $k$-means with outliers to the problem of $k$-means with penalties. 
Then we review an instance of this reduction by \cite{bhaskara2020} and provide a refined result which is the starting point of our work. 
We note that this can be seen as an instantiation of a general principle in optimization where one replaces a constraint with a penalty function known as Lagrangian relaxation \cite{boyd2004convex}. 

\subsection{Review of the Previously Known Reduction} 
We start by giving the following lemma that formalizes how an $\alpha$-approximate solution to the \km with penalties objective can be used to obtain an algorithm providing an $(O(\alpha),O(\alpha))$-approximate solution to the \km with outliers objective. 

\begin{lemma}[\cite{charikar2001algorithms}]
	\label{lem:penalties_to_outliers}
	Let $C$ be an $\alpha$-approximate solution to the \km with penalties objective with penalty $\OPT/(2z) \le \Theta \le \OPT/z$. Let $X_{out}$ denote the set of points $x$ for which $\tau_\Theta(x, C) = \Theta$. 
	Then, it holds that $\phi(X \setminus X_{out},C) = O(\alpha \OPT)$ and $|X_{out}| = O(\alpha z)$. 
    %Then, declaring all points in $Z$ as outliers leads to an $(O(\alpha),O(\alpha))$-approximate solution to the \km with penalties objective. That is, $|Z| = O(\alpha z)$ and $\phi(X \setminus Z, C) = O(\alpha \OPT)$. %, where $\OPT := \min_{C',Z':|C'|=k,|Z'|=z} \phi(X \setminus Z', C')$.
\end{lemma}
\begin{proof} 
    Note that the optimal solution for \km with penalties has a cost upper bounded by $\phi(X \setminus X_{out}^*, C^*) + |X_{out}^*|\Theta = \OPT + z\Theta$. %Here, $\OPT, X_{out}^*$ and $C^*$ refer to the optimal solution in the outlier objective. 
	As $C$ is an $\alpha$-approximate solution to the \km with penalties objective, we have 
	$\phi(X \setminus X_{out}, C) \leq \tau_\Theta(X, C) \leq \alpha \left(\OPT + z \Theta \right) = O(\alpha \OPT).$
	Moreover, paying $\Theta$ for every $x\in X_{out}$ implies  $|X_{out}| \leq \frac{\tau_\Theta(X, C)}{\Theta} \le  \frac{\alpha(\OPT + z \Theta) }{\Theta} = O(\alpha z).$
\end{proof}

We remark that the penalty $\Theta$ depends on $\OPT$, so the algorithm for $k$-means with outliers needs to try this reduction for all $O(\log (n\Delta^2)) = O(\log n)$ powers of $2$ between $1$ and $n\Delta^2$. 
%$\log\OPT = O(\log n)$ of valu possible values of $\Theta$. \mtodo{This sentence needs to be expanded. Say "for the algorithm for the \km with outliers problem". Also, some non-TCS readers might not now what this ``try many guesses means" and why so many are enough"}.

This reduction is very helpful as it allows us to easily adapt  sampling-based algorithms like \kmpp to the setting with penalties since their analysis generalises to this setting. 
For the case of \kmpp, this was shown in \cite{bhaskara2020, li_penalties2020} in the following theorem. 

\begin{algorithm}[h]
	\caption{\kmpp (over)seeding with penalties}
	\label{alg:kmp}
	Input: A set of points $X$, $k$, $\ell$, threshold $\Theta$
	\begin{algorithmic}[1]
		\STATE Uniformly sample $c \in X$ and set $C = \{ c \}$.
		\FOR{$i \leftarrow 2, 3, \dots, \ell$}
		% \STATE Sample $p \in P$ with probability $\frac{cost(p, C)}{\sum_{q \in P} cost(q, C)}$ and add it to $C$.
		\STATE Sample $c \in X$ with probability $\tau_\Theta(c, C) / \tau_\Theta(X, C)$ and add it to $C$.
		\ENDFOR
		\STATE \textbf{return} $C$
	\end{algorithmic}
\end{algorithm}

\begin{restatable}{theorem}{kmeanspp}[\cite{arthur2007k, bhaskara2020, li_penalties2020}]
	\label{thm:kmpp}
	Suppose we run \cref{alg:kmp} for $\ell = k$ steps. 
	Then, the output set $C$ is an $O(\log k)$-approximation to the $k$-means with penalty $\Theta$ objective, in expectation.

\end{restatable}
\begin{proof}[Proof sketch]
The  analysis of \cite{arthur2007k} proves this guarantee for \cref{alg:kmp_original} for any $\phi(a,b) := d^2(a,b)$ such that $d$ is a metric. As the distance function $d'(a,b) = \min(d(a,b), \sqrt{\Theta})$ still defines a metric, one can thus directly use their analysis to prove \cref{thm:kmpp}.  
%The analysis of \cite{arthur2007k} relies on the approximate triangle inequality $\phi(a,c) \le 2(\phi(a,b) + \phi(b,c))$ which holds also for $\tau_\Theta$ (\cref{fact:generalised_triangle_inequality} in \cref{sec:appendix_lemmas}), so their proof also holds after replacing $\phi$ by $\tau_\Theta$. 
\end{proof}
% The only ingredient for making the original analysis of $\kmpp$ go through is to prove two sampling lemmas. The first sampling lemma states...

%The second sampling lemma states...
%You can find a proof for both of the sampling lemmas in the appendix. We also prove a version for them that works in arbitrary metric spaces that lose a $2$-factor compared to their euclidean counterpart.
More details are given in \cref{sec:appendix_lemmas}. 
By \cref{thm:kmpp}, plugging in \cref{alg:kmp} into \cref{lem:penalties_to_outliers} gives
an $(O(\log k),O(\log k))$-approximate algorithm for $\kmpp$ with outliers. Moreover, running \cref{alg:kmp} for $O(k)$ steps results in an $O(1)$-approximation to the $k$-means objective in metric spaces \cite{aggarwal2009adaptive,bhaskara2020}. 
This gives the following tri-criteria approximation via \cref{lem:penalties_to_outliers}: we get an $(O(1), O(1))$-approximation algorithm that needs to use  $O(k)$ centers.

\subsection{Our Improved Reduction}
The starting point of our work is the following improvement of the tri-criteria result from the previous subsection that enables us to get a constant factor approximation algorithm that outputs only $(1+\eps)z$ outliers, which is not possible by using \cref{lem:penalties_to_outliers} as a black box. The catch is that we need to use $O(k/\eps)$ centers. 

%A curious thing about the reduction is that a better approximation guarantee for the \km with penalties objective does not only translate to a better approximation guarantee for the \km with outliers objective, but it also leads to a fewer number of outliers. 
%One way to obtain a better approximation guarantee is by relaxing the number of centers the algorithm is allowed to output. More specifically, \cite{aggarwal2009adaptive} has shown that running \cref{alg:kmp_original} for $O(k)$ iterations results in an $O(1)$-approximation bicriteria algorithm, assuming $\phi = \tau_\infty$ as a cost function. 
%Their analysis directly carries over for arbitrary $\Theta$, which by \cref{lem:penalties_to_outliers} implies an algorithm that outputs $O(k)$ centers, $O(z)$ outliers and is an $O(1)$-approximation \cite{bhaskara2020}. 
%We strengthen the analysis to obtain the following result.

\begin{theorem}
\label{thm:aggarwal_informal}
 Running \cref{alg:kmp} for $\ell = O(k / \eps )$ iterations and $\OPT / (2\eps z) \le \Theta \le \OPT / (\eps z)$ results in a set $C$ with $\tau_\Theta(X^*_{in}, C) = 20 \OPT$, with positive constant probability. 
\end{theorem}
\begin{proof}[Proof sketch (full proof in \cref{sec:appendix_tricriteria})]
For an optimal set of centers $C^* = \{c_1^*, \dots, c_k^*\}$, we define $X_i^* \subseteq X_{in}^*$ as the subset of points $x \in X_{in}^*$ with $c_i^* = \argmin_{c^* \in C^*} \phi(x, c^*)$, where ties are broken arbitrarily. 
Fix one iteration of \cref{alg:kmp} and let $C$ be its current set of centers. 
We refer to a cluster $X_i^*$ as \emph{unsettled} if $\tau_\Theta(X_i^*, C) \ge 10 \tau_\Theta(X_i^*, C^*)$. 
%We will use the easy-to-prove fact that sampling a point from cluster $X_i^*$ such that $x \in X_i^*$ is sampled with probability $\tau_\Theta(x, C) / \tau_\Theta(X_i^*, C)$, makes $X_i^*$ settled with probability at least $\frac15$ (\cref{cor:10apx} in \cref{sec:appendix_lemmas}). 

Suppose that $\tau_\Theta(X^*_{in}, C) \ge 20\OPT$, since otherwise we are already done. We sample a new point $c$ from $X^*_{in}$ with probability
\[
\frac{\tau_\Theta(X^*_{in}, C)}{\tau_\Theta(X^*_{in}, C) + \tau_\Theta(X^*_{out}, C)}  \ge \frac{20\OPT}{20\OPT + \tau_\Theta(X^*_{out}, C)}  \ge \frac{20\OPT}{20\OPT + \OPT/\eps} \ge \frac{\eps}{2},
\]
where we used $\tau_\Theta(X^*_{in}, C) \ge 20\OPT$,  $\tau_\Theta(X^*_{out}, C) \le |X^*_{out}|\Theta \le \OPT/\eps$, and that $\eps$ is small enough.  
Moreover, the cost of all settled clusters is bounded by $10 \OPT$, hence given that $c$ is sampled from $X^*_{in}$, the probability that it is sampled from an unsettled cluster is at least
\[
\frac{\tau_\Theta(X^*_{in}, C) - 10\OPT}{\tau_\Theta(X^*_{in}, C)} 
\ge \frac{20 \OPT - 10 \OPT}{20\OPT}
= \frac12,
\]
where we again used $\tau_\Theta(X^*_{in}, C) \ge 20\OPT$. 
Given that $c$ is sampled from an unsettled cluster according to the $\tau_\Theta$-distribution, \cref{cor:10apx} in \cref{sec:appendix_lemmas} tells us that the cluster becomes settled with probability at least $\frac15$. 
Hence, we make a new cluster settled with probability at least $(\eps/2) \cdot \frac12 \cdot \frac15 = \eps/20$. 

By a standard concentration argument, this implies that after $O(k/\eps)$ steps of the algorithm, either all clusters are settled with positive constant probability, and hence we have $\tau_\Theta(X_{in}^*, C) \le 10\OPT$ and we are done, or during the course of the algorithm, the condition $\tau_\Theta(X_{in}^*, C) \ge 20 \OPT$ stopped being true, in which case we are again done. 
\end{proof}

Note that setting $\eps = 1$ results in the same tri-criteria result we discussed in the previous subsection: This holds as $\tau_
\Theta(X^*_{in}, C) = O(\OPT)$  implies that $\tau_\Theta(X, C) \le \tau_\Theta(X^*_{in}, C) + z\Theta = O(\OPT/\eps)$. 

However, we can get more out of \cref{thm:aggarwal_informal}: note that as $\tau_\Theta(X^*_{in}, C) = O(\OPT)$, for the set $X_{out}$ defined as those $x \in X$ with $\tau_\Theta(x, C) = \Theta$, we have that $|X_{out} \cap X_{in}^*| = \frac{O(\OPT)}{\Theta} = O(\eps z)$. 
Hence, setting $X_{out}$ as our output set of outliers gives on one hand
\[
|X_{out}| \le |X_{out}^*| + |X_{out} \cap X_{in}^*| = z + O(\eps z). 
\]
On the other hand, we can bound
\[
\phi(X  \setminus X_{out}, C) 
= \tau_\Theta(X  \setminus X_{out}, C) 
\le \tau_\Theta(X , C) =  O(\OPT + z \Theta) = O(\OPT / \eps).
\]

%Let us first compare \cref{thm:aggarwal_informal} with \cref{thm:bicriteria_kmpp}.  \cref{thm:bicriteria_kmpp} would in this case only guarantee that $\tau_\Theta(X,C) = O(\tau_\Theta(X, C^*)) = O(\phi(X^*_{in}, C^*) + z \cdot \Theta) = O(\OPT / \eps)$. 
%\cref{thm:aggarwal_informal} guarantees more in this case. First, $\tau_\Theta(X^*_{in}, C) = O(\OPT)$ also implies that $\tau_\Theta(X, C) = \tau_\Theta(X^*_{in}, C) + \tau_\Theta(X^*_{out}, C) = O(\OPT + z \cdot \Theta) = O(\OPT/\eps)$. 
%Second, similarly as in \cref{lem:penalties_to_outliers}, we can define a set $X' \subseteq X^*_{in}$ of points $x$ in $X^*_{in}$ for which $\tau_\Theta(X') = \Theta$. 
%We then have $|X'| \le \frac{\tau_\Theta(X^*_{in}, C)}{\Theta} = O(\OPT / \Theta) = O( \eps z)$. This implies $\phi(X^*_{in} \setminus X', C) = \tau_\Theta(X^*_{in} \setminus X', C) = O(\OPT)$, i.e., up to a small set $X'$ we can also bound $\phi(X^*_{in}, C)$ by $O(\OPT)$. 
Hence, we obtain an $O(1/\eps)$ approximation guarantee while outputting just $(1+\eps)z$ outliers.\footnote{In this particular case, we can even get an $O(1)$-approximation guarantee by labelling the furthest $(1+O(\eps))z$ points as outliers, since the set $X_{out}\cup X_{out}^*$ has size at most $(1+O(\eps))z$ and $\phi(X \setminus (X_{out}\cup X_{out}^*), C) \le \tau_\Theta(X_{in}^*, C) = O(\OPT)$. }
By itself, \cref{thm:aggarwal_informal} is still not satisfactory, as it requires us to oversample the number of centers by a factor of $O(1/\eps)$. 
However, we next show three directions of improvement that lead to more interesting results.

\section{Fast Sequential Algorithm}
\label{sec:paper_localsearch}
%\mtodo{Again, need a starting sentence, where you say what new result is explained in this section. Think like this "this section shows blah; this is based on combining the approach of Lattanzi-Sohler with the reduction outlined in the previous section". then you review LS on its own. Then you say what you get with reduction. Separate clearly. Each paragraph should be about only one of these topics.}

In this section, we present a simple sequential sampling-based algorithm for $k$-means with outliers that achieves an $O(1/\eps)$ approximation and outputs $(1+\eps)z$ outliers: 
\begin{theorem}
\label{thm:lattanzi_informal}
For every $ 0 < \eps < 1$, there exists an $(O(1/\eps),1+\eps)$-approximation algorithm for $k$-means with outliers with running time $\tilde{O}(nk/\eps)$.
\end{theorem}
The algorithm is based on ideas of a recent paper of Lattanzi and Sohler \cite{lattanzi2019better} who proposed augmenting \kmpp with $\tilde{O}(k)$ \emph{local search} steps \cite{kanungo2004local} as follows: 
their algorithm first invokes \kmpp to obtain an initial set of $k$ centers. 
Afterwards, in each local search step, the algorithm samples a $(k+1)$-th point from the same distribution as \kmpp.   
After sampling that point, the algorithm iterates over all $k+1$ current centers and takes out the one whose deletion raises the cost the least (see \cref{alg:localsearch}). 
Running \kmpp followed by $O( k)$ steps of \lsp is known to yield an $O(1)$-approximation for the cost function $\phi = \tau_\infty$ \cite{choo2020kmeans}, with a positive constant probability.

\begin{algorithm}[h]
	\caption{One step of \lsp}
	\label{alg:localsearch}
	{\bfseries Input:} $X$, $C$, threshold $\Theta$
	\begin{algorithmic}[1]
		\STATE Sample $c \in X$ with probability $\tau_\Theta(c, C) / \tau_\Theta(X, C)$
		\STATE $c' \leftarrow \argmin_{d \in C \cup \{c\}} \tau_\Theta(X, C \setminus \{ d \} \cup \{ c \})$
		%\IF{$\tau_\Theta(X, C \setminus \{ c' \} \cup \{ c \} ) < \tau_\Theta(X, C)$}
		%\STATE $C \leftarrow C \setminus \{ c' \} \cup \{ c \}$
		%\ENDIF
		\STATE \textbf{return} $C \setminus \{ c' \} \cup \{ c \}$
	\end{algorithmic}
\end{algorithm}

We get \cref{thm:lattanzi_informal} by proving that $\tilde{O}(k/\eps)$ iterations of \cref{alg:localsearch} with $\OPT/(2\eps z) \le \Theta \le \OPT/(\eps z)$ result in an $(O(1/\eps), 1+\eps)$-approximation guarantee. 
The analysis deals with new technical challenges and is deferred to \cref{sec:appendix_localsearch}. 
%Proof of \cref{thm:lattanzi_informal} requires a more intricate analysis to we can prove the stronger guarantee of \Cref{thm:lattanzi_informal}, by running the local search algorithm $\tilde{O}(k/\eps)$ steps. 
We note that with minor changes to the original analysis of Lattanzi and Sohler, one can show that their result generalizes for arbitrary $\tau_\Theta$, similarly as \cref{thm:kmpp}, and, hence, by \cref{lem:penalties_to_outliers} one obtains an $(O(1),O(1))$-approximation algorithm. % for \km with outliers. 
Our refined analysis, crucially, does \emph{not} use \cref{lem:penalties_to_outliers} as a black box. Instead, we use the idea of the lemma as a building stone for the rather intricate analysis of the algorithm. 

The intuition behind our proof of \cref{thm:lattanzi_informal} comes from \cref{thm:aggarwal_informal}: the difference between running $k$-means++ for $O(k/\eps)$ steps and the local search algorithm we described above is that the local search algorithm additionally  removes one point after each sampling step. 
One can still hope that the increase in cost due to the removals is dominated by the decrease in cost due to the newly sampled center making an unsettled cluster settled, as we have seen in the analysis of \cref{thm:aggarwal_informal}. This is indeed the case. 
%However, the analysis of \cref{thm:lattanzi_informal} faces new challenges and, thus, requires new ideas compared to the classical $k$-means setting \cite{lattanzi2019better}. 
We note that a local search based algorithm was also considered by \cite{gupta2017local}, though with substantially weaker guarantees than our algorithm.

\section{Distributed Algorithms for $k$-means with Outliers}
\label{sec:paper_distributed}

To model the distributed setting, we consider the coordinator model where the input data $X = X_1 \sqcup X_2 \ldots \sqcup X_m$ is split across $m$ machines. 
Each machine can first perform some local computation. Then, each machine sends a message to a global coordinator who computes the final set of cluster centers $C$. 
The main complexity measure is the total number of bits each machine needs to send to the coordinator. 
An informal version of our main distributed result states the following. 

\begin{theorem}
\label{thm:distributed_main}
There exists an $(O(1),1+\eps)$-approximate distributed algorithm in the coordinator model such that each machine sends at most $\tilde{O}(k/\eps)$ many bits. Moreover, each machine only needs to perform polynomial-time computations.
\end{theorem}

This improves on a recent result of \cite{guo2018}. They give an algorithm with a $(1+\eps, 1+\eps)$-approximation guarantee, but the required local running time is exponential. Other results similar to ours include \cite{Guha_distributed, chen2018practical}. 
Below, we sketch the high-level idea of the constructions and leave details to \cref{sec:appendix_distributed}. 
%We provide two distributed algorithms that achieve the guarantees of \cref{thm:distributed_main}. 
%The two algorithms are inspired by the \km algorithms in $\cite{guha2003streaming}$ and $\cite{bahmani2012scalable}$ , respectively. In this section, we give a high-level overview of the approach. The details can be found in $\cref{sec:appendix_distributed}$.

\paragraph{Simpler Construction}
In this paragraph we explain the high-level intuition behind a weaker result that provides an $(O(1/\eps), 1+\eps)$-approximation. 
This generalizes a classical construction of \cite{guha2003streaming} that works in the \km setting. 
%Nevertheless, the result might be of interest due to the efficient local computation: time complexity of each of the $m$ machines and the coordinator is bounded by $\tilde{O}(n k/m)$ and  $\tilde{O}(mk^2)$, respectively. 
For simplicity, we assume here that the machines know the value $\OPT$ and we define $\Theta := \OPT/(\eps z)$. This assumption can be lifted at the cost of a logarithmic overhead in the time - and message complexity.

Each machine starts by running \cref{alg:kmp} for $\tilde{O}(k/\eps)$ steps with input $X_j$ to obtain a set of centers $Y_j$. The sets of centers satisfy the property $\sum_j \tau_\Theta(X_j \cap \Xin^*, Y_j) = O(\OPT)$. 
This follows from an argument very similar to \cref{thm:aggarwal_informal} and will be important later on.
Now, let $X_{out,j}$ denote all points in $X_j$ that have a squared distance of at least $\Theta$ to the closest point in $Y_j$. 
All points in $X_{out,j}$ are declared as outliers and the remaining points in $X_j \setminus X_{out,j}$ are moved to the closest point in $Y_j$ which creates a new, weighted instance $X'_j$. 
Each machine sends its weighted instance together with the number of declared outliers to the coordinator. The coordinator combines the weighted instances and finds an $(O(1/\eps),1+\eps)$-approximate clustering using the algorithm guaranteed by \cref{thm:localsearch} on the instance $X'$, but with the number of outliers equal to $z' = z - |X_{out}| + O(\eps z)$, where $\Xout := \bigcup_j X_{out,j}$. 
Let $C$ denote the corresponding set of cluster centers and $\Xout'$ denote the set consisting of the $z'$ outliers. 
In \cref{thm:coreset_formal} in \cref{sec:appendix_distributed} we prove that $(C, \Xout \cup \Xout')$ is an $(O(1/\eps), 1+\eps)$-approximation. 

The (simple) analysis follows from two observations. 
First, the number of inliers incorrectly labelled as outliers in the first step, i.e., $|\bigcup_j (X_{out, j} \cap X^*_{in})|$, is bounded by $O(\eps z)$. This follows from $\sum_j \tau_\Theta(X_j \cap \Xin^*, Y_j) = O(\OPT)$ and $\tau_\Theta(x, Y_j) = \Theta$ for every $x \in X_{out,j}$. 
This ensures that in the end we output only $O(\eps z)$ more outliers than if the coordinator ran \cref{alg:noisylocalsearch} on the full dataset with original parameter $z$. 

Second, the total movement cost of changing the instance $X$ to instance $X'$ is bounded by $\sum_j \phi(X_j \setminus X_{out, j}, Y_j) \le \sum_j \tau_\Theta(X_j, Y_j) 
\le z\Theta + \sum_j \tau_\Theta(X_j \cap \Xin^*, Y_j) = O(\OPT / \eps)$. 
Hence, the total cost changes additively by $O(\OPT / \eps)$ compared to the case where the coordinator runs \cref{alg:noisylocalsearch} on the full dataset.

\paragraph{Refined Version}
To improve the approximation factor from $O(1/\eps)$ down to $O(1)$ in \cref{thm:distributed_main}, we need to perform two changes. First, the coordinator runs the polynomial-time $(O(1), 1)$-approximation algorithm of \cite{Krishnaswamy2018} on the weighted instance instead of \cref{alg:noisylocalsearch}. 
Second, each machine records not only  the number of points in $X_j \setminus X_{out,j}$ closest to  each point $y \in Y_j$, but it additionally sends for each integer $k \in O(\log \Delta)$, how many of those points have a distance between $2^k$ and $2^{k+1}$ to $y$. The coordinator then constructs an instance based on this refined information. 

\paragraph{$k$-means$\|$}
Adapting the popular $k$-means$\|$ algorithm to the setting with outliers \cite{bahmani2012scalable} can also be accomplished with a construction similar to the one explained above. 
The only difference is that instead of getting the weighted instance as a union of weighted instances from each machine, it is constructed in $O(\log n)$ sampling rounds by using the \kmpipe idea: in each round we sample from the same distribution as in \cref{alg:kmp}, but we sample $\tilde{O}(k/\eps)$ points instead of just one point.

\section{Tight Bounds for $(O(1), O(1))$-approximation Algorithms}
\label{sec:paper_nk2z}

In this section, we discuss tight bounds on the complexity of finding an $(O(1), O(1))$-approximation for the $k$-means with outliers objective. 
In \cref{sec:appendix_metropolis} we prove the following result. 

\begin{restatable}{theorem}{thmmetropolis}
\label{thm:fast_algorithm}
There is an $(O(1), O(1))$-approximation algorithm for $k$-means with outliers that runs in time $\tilde{O}(nk \cdot \min(1, k/z))$ and succeeds with positive constant probability. 
\end{restatable}

This result is based on \cref{thm:aggarwal_informal}, but to speed up the algorithm from $\tilde{O}(nk)$ to $\tilde{O}(nk^2/z)$, we use the Metropolis-Hastings algorithm (with uniform proposal distribution) which was used in \cite{bachem2016approximate, bachem2016fast} for $k$-means. 
Here, the idea is that in one sampling step of \cref{alg:kmp}, instead of computing $\tau_\Theta(x, C)$ for all the points, we first subsample $\tilde{O}(n/z)$ points uniformly at random and only from those points we then sample roughly proportional to their current $\tau_\Theta$ cost by defining a certain Markov chain. 
After speedup, we keep the guarantees of \cref{thm:aggarwal_informal}, but now the running time is $\tilde{O}(nk^2/z)$. 

Note that instead of using Metropolis-Hastings algorithm, we could also just take a uniform subsample and run \cref{alg:kmp} on it. The result of \cite{huang2018epsilon} would give that uniform subsampling the number of points to $\tilde{O}(k(n/z)^2)$ would lose only a constant factor in approximation guarantees. This leads to a $\tilde{O}(k^2(n/z)^2)$ running time and an algorithm with similar running time also based on uniform subsampling for $k$-median was given by \cite{meyerson2004k}. 

%Both approaches use uniform subsampling; the latter is more general, but the second power of $n/z$ is suboptimal. 
%The speed-up by Metropolis algorithm with uniform proposal distribution can be also interpreted as speeding up by uniform subsampling, but in our case we do not need to prove uniform subsampling works for any algorithm, only for a variant of \cref{alg:kmp}. 

\paragraph{Lower Bound}
Next, we provide a matching lower bound. 
%Naturally, we need to restrict ourselves to a certain class of algorithms. 
To that end, we restrict ourselves to the class of algorithms that work in an arbitrary metric space $\fM$ and access the distances of the space only by asking an \textit{oracle} that upon getting queried on two points $x,y \in \fM$ returns their distance. 

Virtually all algorithms with guarantees we know of are of this type, possibly up to a constant loss in their approximation guarantee. In \cref{sec:appendix_lemmas} we verify that this is the case also for algorithms that we consider here. 
%The only change that needs to be made is to use \cref{lem:2apx_metric} instead of \cref{lem:2apx} in all proofs. This loses only a factor of $2$ in the approximation guarantee. 

For the classical problems of \km, $k$-median, and $k$-center, there is an $\Omega(nk)$ lower bound (i.e., showing that so many queries to the oracle are necessary) for metric space algorithms \cite{mettu2004optimal, mettu2002approximation}. This essentially matches the complexity of \kmpp \cite{arthur2007k} or the \lsp algorithm of Lattanzi and Sohler \cite{lattanzi2019better}. 
The lower bound holds also in the setting with outliers, if the output of the algorithm is an assignment of each of the $n$ points to an optimal cluster, together with the list of the outliers. 
On the other hand, \cref{thm:fast_algorithm} gives an $(O(1), O(1))$ algorithm with complexity $\tilde{O}(nk^2/z)$. 
The catch is that the output of this algorithm is just a set of centers and it does not compute for all the points their respective assignment to a closest center. 
Also, it does not label all the outliers. 
For this type of algorithms that only output the set of centers, we prove a matching lower bound of $\Omega(nk^2/z)$ in the metric space query model. The following theorem is proved in \cref{sec:appendix_lower_bounds}. 

\begin{theorem}
\label{thm:lower_bound_informal}
Any randomized algorithm for the $k$-means/$k$-median/$k$-center problem with outliers in the setting $k \ge C$, $z \ge Ck \log k$, and $n \ge Cz$ for an absolute constant $C$ that with probability at least $0.5$ gives an $(O(1), O(1))$-approximation in the general metric space query model, needs $\Omega(nk^2/z)$ queries.  
\end{theorem}

Let us provide a brief intuition of the proof. The construction that yields \cref{thm:lower_bound_informal} is the following: we have $k/2$ large clusters and $k/2$ small clusters. All clusters are well separated and small clusters together contain $\Theta(z)$ points. 
As the algorithm can output only $\Theta(z)$ outliers, it needs to ``find'' $\Omega(k)$ small clusters. 
However, a point chosen from the input set uniformly at random is from a small cluster with probability $O(z/n)$ and we expect to need $\Omega(k)$ queries until we find out whether a point is from a small or a big cluster. This leads to the lower bound of  $\Omega(k \cdot \frac{n}{z} \cdot k) = \Omega(nk^2/z)$ rounds.  

\paragraph{Why it makes sense to consider bicriteria approximation}
Let us also observe a different lower bound against $(O(1), 1)$-approximation algorithms that motivates why we are concerned with algorithms that can output $(1+\eps)z$ outliers. 
The following construction is, e.g., in \cite{indyk1999sublinear}. 
Consider an input metric space with $z = n-1$ and $k=1$, where any two points have a distance of $1$, up to a single pair $x_1, x_2$ whose distance is $0$ (or some small $\eps$). Any approximation algorithm (even a randomized one) that outputs exactly $z$ outliers needs to ``find'' the pair $x_1, x_2$ and $\Omega(z^2)$ queries are needed for this.  
The example shows that there is a fundamental limit to the speed of $(O(1), 1)$-approximation algorithms: for $z$ linear in $n$ they even need $\Omega(n^2)$ time. 
This should be contrasted with the $(O(1), 1+\eps)$-approximation algorithms that only need $\poly(k/\eps)$ time for $z = \Theta(n)$ that follows from \cite{huang2018epsilon}.  

%Of course, second reason, why it makes sense to consider algorithms with bicriteria guarantees, is the same as why we consider $k$-means with outliers in the first place: while $k$-means solution is unstable in presence of few outliers, solution to $k$-means with $z$ outliers is again unstable when the number of true outliers is actually slightly bigger than $z$. One reason for studying sampling-based algorithms is the hope that in practice they will be robust against the problems like the actual number of outliers being somewhat bigger than the parameter $z$. \todo{new}

\section{Experiments}
\label{sec:paper_experiments}

We tested the following algorithms on the datasets \emph{kdd} (KDD Cup 1999) subsampled to $10\,000$ points with $38$ dimensions and \emph{spam} (Spambase) with $4601$ points in $58$ dimensions \cite{Dua:2019}. 
We set the number of outliers $z$ to be $10$ percent of the dataset. 

\iffalse
\texttt{Lloyd}: Variant of Lloyd's algorithm\cite{lloyd1982least, kmeans-} that handles outliers with random initialization ($10$ iterations);
\texttt{$k$-means++}: $k$-means++ seeding\cite{arthur2007k};
\texttt{$k$-means++ with penalties}: $k$-means++ with penalties\cite{bhaskara2020, li_penalties2020};
\texttt{Metropolized $k$-means++ with penalties}: $k$-means++ with penalties, sped up by Metropolis subsampling with $100$ steps (\cref{sec:appendix_metropolis});
\texttt{Distributed $k$-means++ with penalties}:  simplified variant of \cref{alg:streaming} in \cref{sec:appendix_distributed} -- input partitioned in $10$ subinputs, $k$-means++ with penalties run on each of them with $\ell = 2k$, weighted instances sent to coordinator who runs $k$-means++ with penalties again to obtain the final $k$ centers);
\texttt{Sped up local search}: $k$-means++ with penalties followed by additional $k$ local search steps (for the objective with penalties). 
\fi

\begin{itemize}
    \item \texttt{Lloyd}: Variant of Lloyd's algorithm \cite{lloyd1982least, kmeans-} that handles outliers with random initialization ($10$ iterations);
    \item \texttt{$k$-means++}: $k$-means++ seeding\cite{arthur2007k};
    \item \texttt{$k$-means++ with penalties}: $k$-means++ with penalties\cite{bhaskara2020,li_penalties2020};
    \item \texttt{Metropolized $k$-means++ with penalties}: $k$-means++ with penalties, sped up by Metropolis-Hastings algorithm with $100$ steps (see \cref{sec:appendix_metropolis});
    \item \texttt{Distributed $k$-means++ with penalties}:  simplified variant of \cref{alg:streaming} -- input is partitioned in $10$ subinputs, $k$-means++ with penalties is run on each of them with $\ell = 2k$, and weighted instances are sent to coordinator who runs $k$-means++ with penalties again to obtain the final $k$ centers);
    \item \texttt{Sped up local search}: Variant of \cref{alg:noisylocalsearch} -- $k$-means++ with penalties followed by $k$ additional local search steps (for the objective with penalties). 
\end{itemize}

To guess the value of $\Theta$ in all except the first two algorithms, we tried $10$ values from $1$ to $10^{10}$, exponentially separated. 
The best solution was then picked and we followed by running $10$ Lloyd iterations on it with the number of outliers for these iterations set to $z$ (the same for the second $k$-means++ algorithm). The results for this setup for $k\in \{5, 10, \dots, 50\}$ are in \cref{fig:kdd,fig:spam}.

\begin{figure}[h!]
    \centering
    %\begin{subfigure}[b]{.49\textwidth}
        \centering
        \includegraphics[width=.95\textwidth]{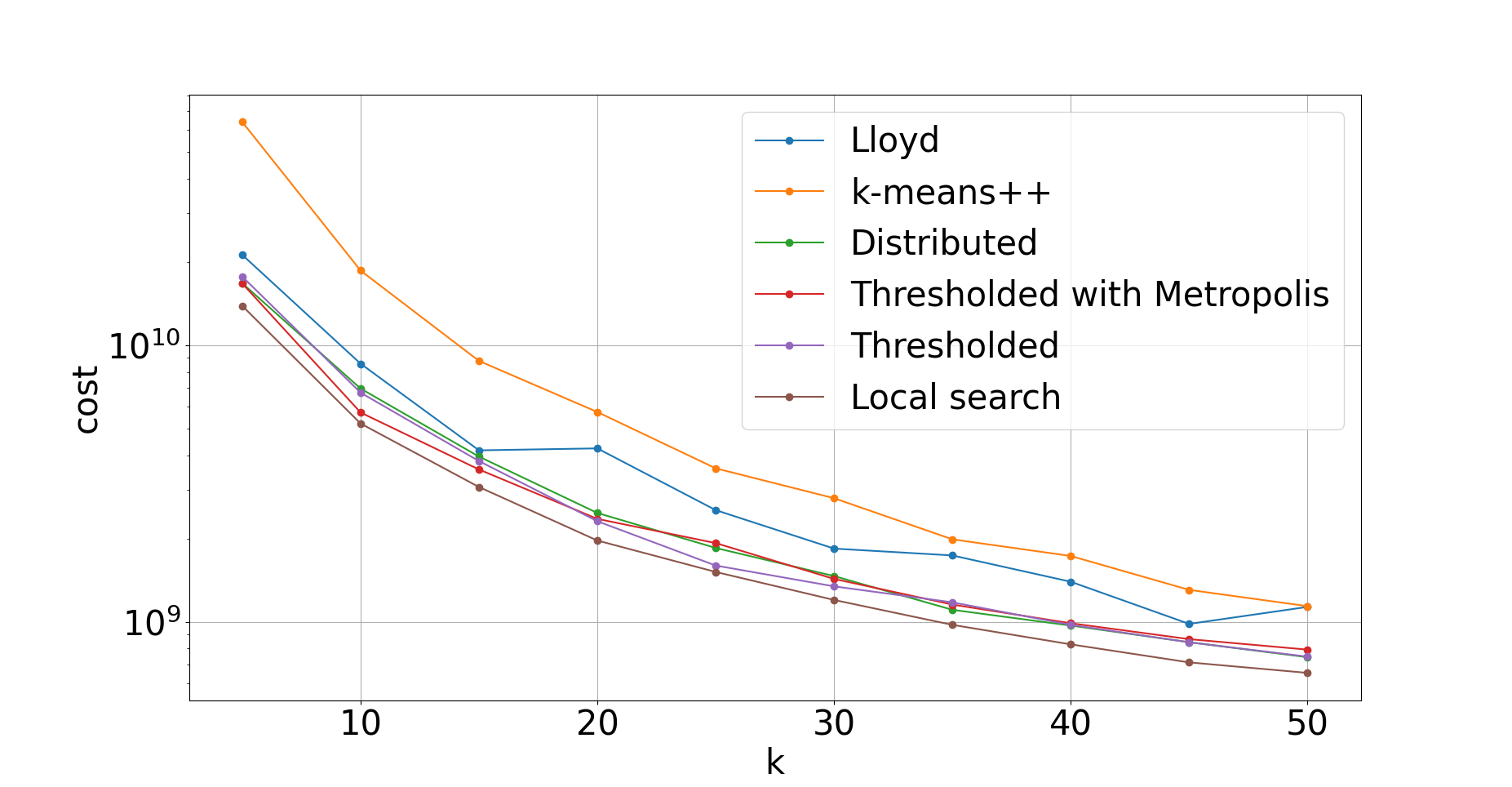}
        \caption{Experiments on kdd}
        \label{fig:kdd}
    %\end{subfigure}
    %\begin{subfigure}[b]{.49\textwidth}
        \includegraphics[width=.95\textwidth]{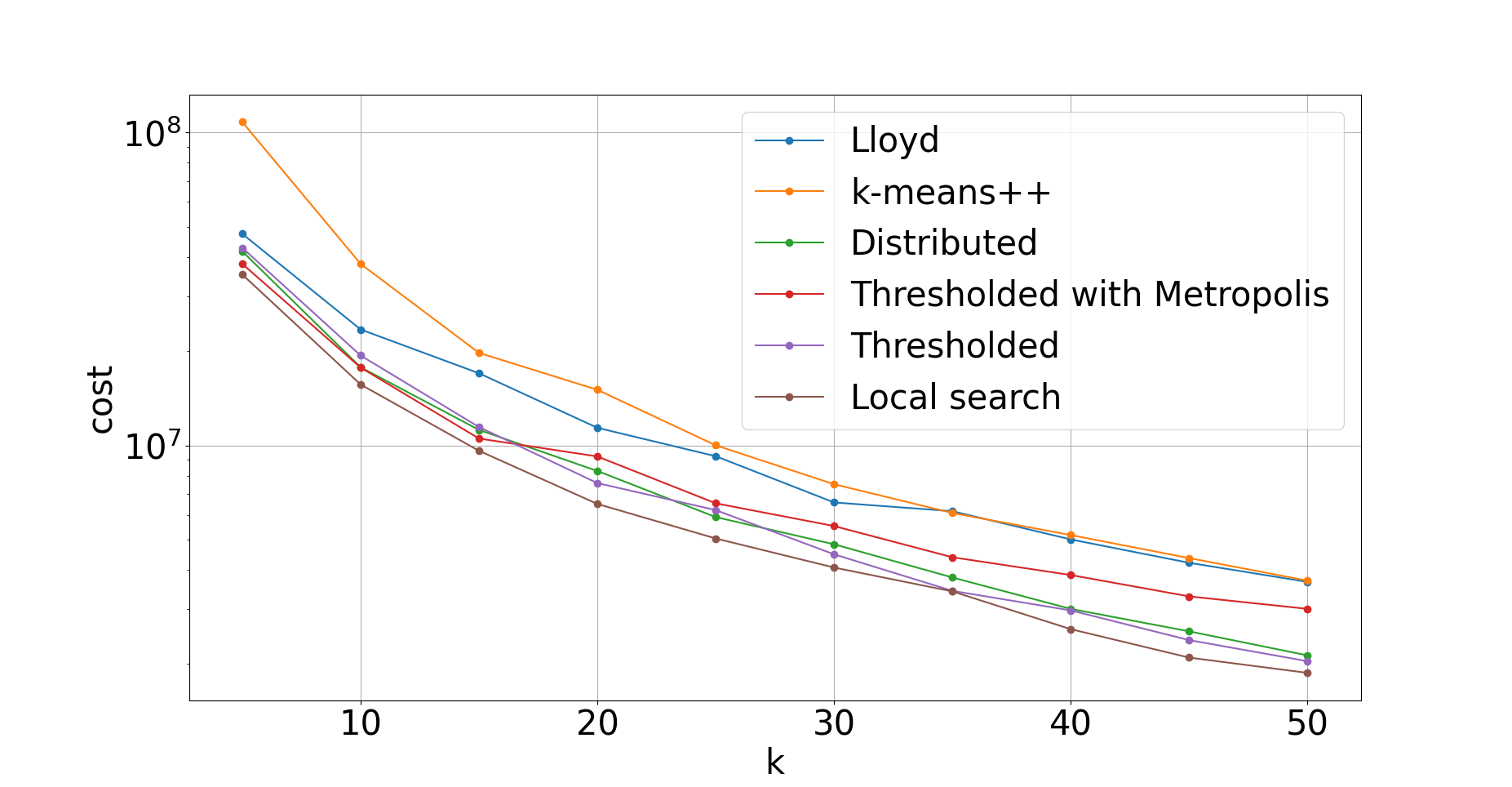}
        \caption{Experiments on spam}
        \label{fig:spam}
    %\end{subfigure}
    %\caption{Experiments}
    \label{fig:experiments}
\end{figure}

\texttt{$k$-means with penalties} outperforms the first two baseline algorithms on average by around $40\%$ in both datasets. Surprisingly, $k$-means++ seeding leads to consistently worse solutions than random initialization. We believe this indicates that the datasets indeed contain outliers that $k$-means++ picks preferably due to their large distance from other points. 
Distributed and metropolized variants are on par with \texttt{$k$-means++ with penalties}, except for the metropolized variant on the \emph{spam} dataset. 
\texttt{Sped up local search} consistently outperforms \texttt{$k$-means++ with penalties} by around $12\%$ in both datasets. It is also significantly slower, but we implemented only its simple $\tilde{O}(nk^2)$ implementation instead of the best possible $\tilde{O}(nk)$.

\section{Conclusion}

We have shown that several simple sampling-based algorithms for $k$-means can be adapted to handle outliers and retain strong theoretical guarantees, while still being similarly simple to implement. 
As a theoretical application, we settled the complexity of finding an $(O(1), O(1))$-approximation for $k$-means with outliers to $\tilde{\Theta}(nk^2/z)$ in the query model. 
%An application of the idea of speeding up local search by sampling of \cite{lattanzi2019better} suggest their technique is useful beyond the classical $k$-means problem. 

\section*{Acknowledgment}

We thank Davin Choo, Mohsen Ghaffari, Saeed Ilchi, Andreas Krause, and Julian Portmann for engaging discussions. In particular, we thank Mohsen for his numerous helpful remarks and Davin and Saeed for sharing their code with us. The authors were supported by the European Research Council (ERC) under the European Unions Horizon 2020 research and innovation programme (grant agreement No. 853109).

\bibliography{ref}
\bibliographystyle{alpha}

% \appendix

\newpage

\begin{appendices}
\crefalias{section}{appendix}
\section{Preparatory lemmas}
\label{sec:appendix_lemmas}

\paragraph{More notation} 
We start by setting up some more notation. 
Note that the optimal clustering $C^*$ splits the dataset $X$ into two parts: the set of \emph{inliers} $\Xin^*$ with $\phi(\Xin^*, C^*) = \OPT$ and the set of \emph{outliers} $\Xout^*$ with $|\Xout^*| = z$. 
Moreover, the set $\Xin^*$ is naturally split in $k$ sets $X_1^*, X_2^*, \dots, X_k^*$, where $X_j^*$ is the set of points $x \in X_{in}^*$ for which $\phi(x, C) = \phi(x, c_j^*)$. 
When using the notation $\tau_\Theta(X, C)$, we often drop the subscript $\Theta$ when it is clear from context and write just $\tau(X, C)$. 
For a single point $c$, we drop parentheses and write $\phi(X,c)$ instead of $\phi(X,\{c\})$ and similarly for $\tau$. 
We define $\phi^{-z}(X, C) = \min_{X_{out}, |X_{out}|=z} \phi(X \setminus X_{out}, C)$. 
We use $\textrm{out}_\Theta(X, C)$ to denote the subset of points $x \in X$ such that $\tau_\Theta(x, C) = \Theta$ and  $\textrm{in}_\Theta(X, C)$ to denote the subset of points $x \in X$ such that $\tau_\Theta(x, C) < \Theta$. 
We assume that the minimum distance between any two points in $X$ is lower bounded by $1$ and their maximum distance is upper bounded by $\Delta$. 

In this section we collect several useful lemmas that are either well-known or are variations of well-known results. First, we recall the well-known Chernoff bounds and then, for completeness, prove a generalization of sampling lemmas that were used to analyze \kmpp \cite{arthur2007k}. 
These were, in slightly different forms, also proved in \cite{li_penalties2020} or \cite{bhaskara2020}. 

\begin{theorem}[Chernoff bounds]
\label{thm:chernoff}
Suppose $X_1, \dots, X_n$ are independent random variables taking values in $\{0, 1\}$. Let $X$ denote their sum. Then for any $0 \le \delta \le 1$ we have
\[
\P(X \le (1-\delta)\E[X]) \le \e^{-\E[X]\delta^2 / 2}
\]
and
\[
\P(X \ge (1+\delta)\E[X]) \le \e^{-\E[X]\delta^2 / 3}.
\]
Moreover, for any $\delta \ge 1$ we have
\[
\P(X \ge (1+\delta)\E[X]) \le \e^{-\E[X]\delta / 3}.
\]
\end{theorem}

Next, we state basic facts about metric spaces and $\mathbb{R}^d$ used to prove the sampling lemmas below. 

\begin{fact}[cf. \cite{arthur2007k}]
\label{fact:mean_minimizes}
Let $A$ be a subset of $\mathbb{R}^d$. Then 
\[
\inf_{\mu \in \mathbb{R}^d} \phi(A, \mu) = \phi(A, \mu(A))
\]
where $\mu(A) = \sum_{x \in A}x/|A|$ is the mean of $A$. 
\end{fact}

\begin{fact}[cf. \cite{arthur2007k}]
\label{fact:generalised_triangle_inequality}
Let $a, b, c \in M$ be three points in an arbitrary metric space. Then 
\[
\phi(a, c) \le 2(\phi(a, b) + \phi(b, c))
\]	
and, more generally, for $C \subseteq M$ and $\eps > 0$ we have 
\[
\phi(a,C) - \phi(b,C) \leq \eps \cdot \phi(b,C) + \left(1 + \frac{1}{\eps} \right)\phi(a,b),
\]
where the former inequality is a special case for $\eps = 1$ and $C = \{c\}$. 
\end{fact}
\begin{proof}
Let $d$ denote the distance function of $M$. For $c \in C$ such that $d(b,C) = d(b,c)$ we have
\begin{align*}
&\phi(a, C) = d^2(a, C) = d^2(a, c)
\le (d(a,b) + d(b,c))^2\\
&= (d(a,b) + d(b,C))^2
= d^2(a,b) + d^2(b,C) + 2d(a,b)d(b,C)\\
&\le d^2(a,b) + d^2(b,C) + \frac{1}{\eps} d^2(a,b) + \eps d^2(b,C)\\
&= \phi(a,b) + \phi(b,C) + \frac{1}{\eps} \phi(a,b) + \eps \phi(b,C)
\end{align*}
which is the the needed inequality, up to rearrangement. We used that
\[
2d(a,b)d(b,C) \le \frac{1}{\eps} d^2(a,b) + \eps d^2(b,C)
\]
is equivalent to
\[
(\sqrt{\eps} d(b,C) - d(a,b)/\sqrt{\eps})^2 \ge 0. 
\]

\end{proof}

\begin{fact}
\label{fact:thresholded_triangle_inequality}
Let $M$ be an arbitrary metric space with distance function $d$. 
Then for any constant $T > 0$, the function $d': d'(a, b) = \min(d(a, b), T)$ is also a distance function. 
\end{fact}
\begin{proof}
For any $a, b, c \in M$, we need to prove $d'(a, c) \le d'(a, b) + d'(b, c)$. If $d(a,b) \ge T$ or $d(b,c) \ge T$, respectively, we have $d'(a,b) = T$ or $d'(b,c) = T$, respectively, and, hence, $d'(a, b) + d'(b, c) \ge T \ge d'(a,c)$, as needed. Otherwise, we have $d'(a,c) \le d(a,c) \le d(a,b) + d(b,c) = d'(a,b) + d'(b,c)$, as needed. 
\end{proof}
\cref{fact:thresholded_triangle_inequality} for example implies that in \cref{fact:generalised_triangle_inequality} we can replace $\phi$ by $\tau_\Theta$, as $\tau_\Theta(a,b) = \min(d(a,b), \sqrt{\Theta})^2$. 
It also implies \cref{thm:kmpp}, as the analysis of \cite{arthur2007k} holds for an arbitrary metric space. 
However, this loses a factor of $2$ in approximation guarantee of \cref{thm:kmpp} as we explain next. 
Proofs of \cref{thm:kmpp} in \cite{li_penalties2020, bhaskara2020} instead generalize Lemmas 3.1 and 3.2 in \cite{arthur2007k} for $\tau$ costs to get \cref{thm:kmpp}; this can recover the same constants as in \cite{arthur2007k}. 
For completeness, we also reprove these lemmas here as \cref{lem:2apx,lem:8apx} since we refer to them later in the proofs. Also, we prove a version of \cref{lem:2apx} for metric spaces as \cref{lem:2apx_metric}. This is the reason why arguments for arbitrary metric spaces lose additional factor of $2$ in approximation guarantees. The relevance of \cref{lem:2apx_metric} comes from the fact that it implies our algorithms work in arbitrary metric spaces, which is important for applications in  \cref{sec:paper_nk2z}. 
%The first sampling lemma below is saying that sampling a random point from a cluster $A$ results in a $2$-approximation of the optimal cost in expectation. 
%This was \cite{arthur2007k}. 

\begin{lemma}[cf. Lemma 3.1 in \cite{arthur2007k}]
\label{lem:2apx}
For any cluster $A \subseteq \mathbb{R}^d$ we have
%with optimum center $m$ let $\ain \subseteq A$ be a subset of points $c$ for which $\phi(c, m) \le \Theta$ and $\aout = A \setmiforφwenus \ain$. 
%We have $\tau(A, m) = \tau(A, \mu(\ain))$ and 
\[
\frac{1}{|A|} \sum_{c \in A} \tau_\Theta(A, c) \le 2 \inf_{\mu \in \mathbb{R}^d} \tau_\Theta(A, \mu).
\]
\end{lemma}
\begin{proof}
Recall that for the cost function $\phi$ we have $\frac{1}{|A|}\sum_{c \in A}\phi(A, c) = 2\phi(A, \mu(A))$ by Lemma 3.1 in \cite{arthur2007k}. 
Let us fix an arbitrary $\mu\in \mathbb{R}^d$ and denote $\ain = \inn_\Theta(A, \mu)$ and $\aout = \outt_\Theta(A, \mu)$. 
Note that $\tau_\Theta(A, \mu(\ain)) \le \tau_\Theta(A, \mu)$ due to \cref{fact:mean_minimizes}. 
We have 
\begin{align*}
\frac{1}{|A|} \sum_{c \in A} \tau_\Theta(A, c) 
&= \frac{1}{|A|} ( \sum_{c \in \ain}\tau_\Theta(\ain, c) + \sum_{c \in \ain}\tau_\Theta(\aout, c) +  \sum_{c \in \aout}\tau_\Theta(A, c) )\\
&\le \frac{1}{|A|} ( \sum_{c \in \ain}\phi(\ain, c) + |\ain||\aout|\Theta + |\aout||A|\Theta )\\
&= \frac{1}{|A|} ( 2|\ain|\phi(\ain, \mu(\ain)) + |\ain||\aout|\Theta + |\aout||A|\Theta )&&\text{Lemma 3.1 in \cite{arthur2007k}}\\
&\le 2\phi(\ain, \mu(\ain)) + 2|\aout|\Theta\\
&\le 2\tau_\Theta(A, \mu(\ain))
\le 2\tau_\Theta(A, \mu).
\end{align*}
\end{proof}

%We also prove \cref{lem:2apx_metric}: a version of \cref{lem:2apx} for metric spaces. 
%One can make our algorithms work in arbitrary metric spaces (cf. \cref{sec:paper_nk2z} for the motivation) with higher constant approximation factor by using \cref{lem:2apx_metric} instead of \cref{lem:2apx} in the analysis. 

\begin{lemma}
\label{lem:2apx_metric}
For any metric space $M$ and subset of its points $A$ we have
\[
\frac{1}{|A|}\sum_{c \in A}  \phi(A, c) \le 4 \inf_{\mu \in M} \phi(A, \mu).
\]
\end{lemma}
\begin{proof}
For any $\mu \in M$, we can write
\begin{align*}
& \frac{1}{|A|} \sum_{c \in A} \phi(A, c)
= \frac{1}{|A|} \sum_{c \in A} \sum_{d \in A} \phi(d, c)\\
&\le \frac{2}{|A|} \sum_{c \in A}  \sum_{d \in A}  \phi(d, \mu) + \phi(\mu, c)&&\text{\cref{fact:generalised_triangle_inequality}}\\
&= 4\phi(A, \mu),
\end{align*}
as needed. 
\end{proof}

The second sampling lemma states that sampling a point from a cluster $A$ proportional to its $\tau$-cost with respect to the current clustering $C$ results in an $8$-approximation of the cost of $A$. 

\begin{lemma}[cf. Lemma 3.2 in \cite{arthur2007k}, Lemma 4 in \cite{li_penalties2020} and Lemma 6 in \cite{bhaskara2020}]
\label{lem:8apx}
For any cluster $A \subseteq \mathbb{R}^d$ and an arbitrary point set $C$, we have 
\[
\sum_{c \in A} \frac{\tau_\Theta(c, C)}{\tau_\Theta(A, C)} \tau_\Theta(A, C \cup \{c\}) \le 8 \inf_{\mu \in \mathbb{R}^d} \tau_\Theta(A, \mu).  
\]
\end{lemma}
\begin{proof}
Fix an arbitrary $\mu \in \mathbb{R}^d$ and $c, d \in A$. We have
\[
\tau_\Theta(c, C) 
\le  2( \tau_\Theta(c, d) + \tau_\Theta(d, C))
\]
by \cref{fact:generalised_triangle_inequality}. 
For a given $c \in A$ we can then average over all $d \in A$ to get
\begin{align}
\label{eq:blue_effect}
\tau_\Theta(c, C) 
\le \frac{2}{|A|} \sum_{d\in A} (\tau_\Theta(c, d) + \tau_\Theta(d, C) )
= \frac{2}{|A|} ( \tau_\Theta(A, c) + \tau_\Theta(A, C) ).
\end{align}
This implies
\begin{align*}
&\sum_{c \in A} \frac{\tau_\Theta(c, C)}{\tau_\Theta(A, C)} \tau_\Theta(A, C \cup \{c\}) \\ 
%&= \sum_{c \in A} \frac{\min(\Theta, \phi(c, C))}{\sum_{d \in A} \min(\Theta, \phi(d, C))} \sum_{d\in A} \min( \Theta, \phi(d, C), \phi(d, c)) \\
&\le \sum_{c \in A} \frac{\frac{2}{|A|} ( \tau_\Theta(A, c) + \tau_\Theta(A, C) )}{\tau_\Theta(A, C)} \sum_{d\in A} \min( \Theta, \phi(d, C), \phi(d, c)) && \text{\cref{eq:blue_effect}} \\
&\le \sum_{c \in A} \frac{\frac{2}{|A|} \tau_\Theta(A, c) }{\tau_\Theta(A, C)} \sum_{d\in A} \min( \Theta, \phi(d, C)) 
+ \sum_{c \in A} \frac{\frac{2}{|A|} \tau_\Theta(A, C)}{ \tau_\Theta(A, C)} \sum_{d\in A} \min( \Theta, \phi(d, c)) \\
&= \sum_{c \in A} \frac{2}{|A|} \tau_\Theta(A, c)
+ \frac{2}{|A|} \sum_{c \in A} \tau_\Theta(A, c) \\
&= \frac{4}{|A|} \sum_{c \in A} \tau_\Theta(A, c)
\le 8 \tau_\Theta(A, \mu). && \text{\cref{lem:2apx}}
\end{align*}

\end{proof}

Both \cref{lem:2apx} and \cref{lem:8apx} were proven in the same way as corresponding lemmas in \cite{arthur2007k}, and  \cref{thm:kmpp} follows from these lemmas in the same way as Lemma 3.3 in \cite{arthur2007k} follows from Lemmas 3.1 and 3.2 in \cite{arthur2007k} (cf. \cite{bhaskara2020, li_penalties2020}). 
Finally, we will use a simple corollary of \cref{lem:8apx}. 

\begin{corollary}
\label{cor:10apx}
For any cluster $A \subseteq \mathbb{R}^d$ and an arbitrary point set $C$, sampling a point $c \in A$ proportional to $\tau_\Theta(c, C)$ results in  
\[
\tau_\Theta(A, C \cup \{c\}) \le 10 \inf_{\mu \in \mathbb{R}^d} \tau_\Theta(A, \mu),  
\]
with probability at least $\frac{1}{5}$. 
\end{corollary}
\begin{proof}
Follows from \cref{lem:8apx} by Markov inequality. 
\end{proof}

\section{Basic Oversampling Result}
\label{sec:appendix_tricriteria}
In this section we prove a formal and more general version of \cref{thm:aggarwal_informal}. The full generality of the theorem is needed at a later point in time. To obtain \cref{thm:aggarwal_informal}, plug in some arbitrary $\Theta  \in  [\OPT/(2\eps z),\OPT/(\eps z)]$ and $\delta = 0.5$ in the following theorem. 

\begin{theorem}[cf. \cite{aggarwal2009adaptive}]
\label{thm:aggarwal}
Let $\delta, \eps \in (0,1]$ be arbitrary and suppose we run \cref{alg:kmp} for $\ell = O(k/\eps \cdot \log 1/\delta)$ steps  to get output $C$. Then, with probability at least $1 - \delta$ we have
\[
\tau_\Theta(X^*_{in}, C) = O(\OPT + \eps z \Theta). 
\]
\end{theorem}

\begin{proof}

We refer to a cluster $X_i^*$ as unsettled if $\tau(X_i^*, C) \ge 10 \tau(X_i^*, \mu(X_i^*))$. 
Fix one iteration of \cref{alg:kmp} and let $C$ be its current set of centers. 
Suppose that $\tau(X^*_{in}, C) \ge 20(\OPT + \eps z \Theta)$, since otherwise we are already done. We sample a new point $c$ from $X^*_{in}$ with probability
\[
\frac{\tau(X^*_{in}, C)}{\tau(X^*_{in}, C) + \tau(X^*_{out}, C)}  \ge \frac{20\eps z \Theta}{20\eps z \Theta + z \Theta}  \ge \eps/2,
\]
where we used $\tau(X^*_{in}, C) \ge 20\eps z \Theta$, $\tau(X^*_{out}, C) \le |X^*_{out}|\Theta = z\Theta$, and $\eps \leq 1$. 

Moreover, given that $c$ is sampled from $X^*_{in}$, the probability that it is sampled from an unsettled cluster is at least
\[
\frac{\tau(X^*_{in}, C) - 10\OPT}{\tau(X^*_{in}, C)} 
\ge \frac{20 \OPT - 10 \OPT}{20\OPT}
\ge \frac12,
\]
where we used $\tau(X^*_{in}, C) \ge 20\OPT$. 
Given that $c$ is sampled from an unsettled cluster according to the $\tau$-distribution, \cref{cor:10apx} tells us that the cluster becomes settled with probability at least $\frac15$. 
Hence, we make a new cluster settled with probability at least $(\eps/2) \cdot \frac12 \cdot \frac15 = \eps/20$. 

Finally, consider all $O(k/\eps \cdot \log1/\delta)$ iterations. 
We call each iteration \emph{good}, if either $\tau(X^*_{in}, C) \le 20(\OPT + \eps z \Theta)$, or we make a new cluster settled. 
As each iteration is good with probability at least $\eps/20$, the expected number of good iterations is $\Omega(k \log 1/\delta)$. 
By \cref{thm:chernoff}, at least $k$ iterations are good with probability at least $1 - \e^{-\Omega(\log1/\delta)} \ge 1- \delta$. 
This implies that in the end either $\tau(X^*_{in}, C) \le 20(\OPT + \eps z \Theta)$ and we are done, or all clusters are settled. But this implies $\tau(X_{in}^*, C) \le 10 \sum_i \tau(X_i^*, \mu(X_i^*)) \le 10 \sum_i \phi(X_i^*, \mu(X_i^*)) = 10 \OPT$, so we are again done.  
\end{proof}

\section{Proof of  \cref{thm:localsearch}}
\label{sec:appendix_localsearch}

In this section we explain in more detail how to adapt the local-search algorithm of \cite{lattanzi2019better} in order to obtain an $(O(1/\eps),1+\eps)$-approximation algorithm for \km with outliers.

As explained in the main part of the paper, the main idea is to run the local-search algorithm with penalties. That is, each point is sampled as a new cluster center proportional to its $\tau_\Theta$-cost, where $\Theta = \Theta(\frac{\OPT}{\eps z})$. In total, we run $O(k \log \log k + k\log(1/\eps)/\eps)$ local-search steps. Hence, the algorithm can be implemented in time $\tilde{O}(nk/\eps)$ (cf. \cite{choo2020kmeans}, Section 3.3). 
After the first $O(k \log \log k)$ search steps, we show that the current cost is with constant probability at most $O(\OPT/\eps)$. The analysis is quite similar to the analysis of \cite{lattanzi2019better}. Afterwards, within the next $O(k\log(1/\eps)/\eps)$ steps, we show that the number of points with a squared distance of at least $10\Theta$ (The extra factor of $10$ is needed in the analysis for technical reasons) drops below $(1+\eps)z$ with constant probability. 
Hence, we can declare all those points as outliers and the $\phi$-cost of each remaining point is then only at most $10$ times larger than its $\tau_\Theta$ cost. 
Therefore, $\phi(X, C)$ can still be bounded by $O(\OPT/\eps)$, as needed. 
The formal algorithm description is given below as \cref{alg:noisylocalsearch}. 
Recall that $\Theta$ depends on $\OPT$ and as $\OPT$ is unknown to the algorithm, it needs to be "guessed". Also note that due to technicalities in the analysis, the algorithm does not simply output the solution one obtains after running all the local-search steps, but is also considers all "intermediate" solutions as final solutions. The reason is that, unlike the total cost, the number of points with a squared distance of at least $10\Theta$ is not necessarily decreasing monotonically. 

\begin{algorithm}[]
	\caption{\nlsp}
	\label{alg:noisylocalsearch}
	Input: $X$, $k$, $z$, $\eps$, $\Delta$ \\ Assumptions: $k + z < |X|$, $\eps \in (0,1]$
	\begin{algorithmic}[1]
	    \STATE $C \leftarrow \emptyset$ 
	    
	    \STATE
	    $\Xout \leftarrow \emptyset$
		\FOR{$i = 0$ to $\lceil \log(n\Delta^2)\rceil$}
		\STATE $OPT_{guess} \leftarrow 2^i$, 
		$\beta \leftarrow 300/\eps$, 
		$\Theta  \leftarrow \frac{\beta OPT_{guess}}{z}$
        \STATE $C_{i,0} \leftarrow \kmpp(X, k, \Theta)$  (\cref{alg:kmp})
        \FOR{$j \leftarrow 1, 2, \dots, O(k \log \log k  + k \cdot \frac{\log(1/\eps)}{\eps})$}
        \STATE $C_{i,j} \leftarrow \lsp\textrm{step}(X, C_{i,j-1}, \Theta)$ (\cref{alg:localsearch})
        \STATE $\Xout^{i,j} \leftarrow \outt_{10\Theta}(X,C)$
        \IF{$|\Xout^{i,j}| \leq (1+\eps)z$ and $\phi(X \setminus \Xout^{i,j},C_{i,j}) < \phi(X \setminus \Xout,C) $}
        \STATE $C \leftarrow C_{i,j}$ \STATE $\Xout \leftarrow \Xout^{i,j}$
        \ENDIF
        \ENDFOR
		\ENDFOR
		\RETURN $(C,\Xout)$
	\end{algorithmic}
\end{algorithm}

\begin{theorem}[Formal version of \cref{thm:lattanzi_informal}]
\label{thm:localsearch}
Let $\eps \in (0,1]$ be arbitrary and $X \subset \mathbb{R}^d$ be a set of $n$ points. 
Furthermore, let $k \in \mathbb{N}$ and $z \in \mathbb{N}_0$ be such that $k + z < n$ (otherwise, the cost is $0$). 
Then, with positive constant probability, \cref{alg:noisylocalsearch} returns a set $\Xout$ containing at most $(1+\eps)z$ points and a set of $k$ cluster centers $C$ such that
$$\phi(X \setminus \Xout,C) = O(\OPT/\eps) $$
for $\OPT = \min_{C',\Xout' \colon |C'| = k,|\Xout'| = z} \phi(X \setminus \Xout', C')$. 
%Moreover, $\cref{alg:noisylocalsearch}$ can be implemented to run in time $\tilde{O}(n k  \log(\Delta /\eps)/\eps)$.
\end{theorem}
Before going into the actual proof of \cref{thm:localsearch}, we first need to introduce additional notation. As before, $X$ denotes the set of all input points and $\Xout^*$ denotes the set of all input points that are declared as outliers by the optimal solution under consideration. %$X^*_{in} = X \setminus X^*_{out}$.
Moreover, let $X^*_{far} \subseteq X \setminus \Xout^*$ denote the set consisting of those points in $X \setminus X_{out}^*$ that have a squared distance of at least $\Theta$ to the closest optimal cluster center. Now, we define $\Xin^* = X \setminus (\Xout^* \cup X^*_{far})$. In other words, we split $X = X_{in}^* \sqcup X_{far}^* \sqcup X_{out}^*$ (in \cref{sec:paper_previous_work} we used $X = X_{in}^*  \sqcup X_{out}^*$).

For a fixed \emph{candidate clustering} $C$ and for $i \in [k]$, $X_i \subseteq X^*_{in}$ denotes the set consisting of those points in $\Xin^*$ for which the closest center $c \in C$ is the center $c_i$. 
Similarly, let $X_i^*$ denote the set consisting of those points in $\Xin^*$ that are clustered to the center $c_i^*$. 
Moreover, $X_{far,i}$ denotes the set of points in $X^*_{far}$ that are closest to $c_i$ among all candidate centers in $C$.  

We say that the optimal cluster $X_i^*$ is \emph{ignored by $C$} if for every $x \in X_i^*$ we have $\tau_\Theta(x, C) = \Theta$. We define $I$ to be the set of indices corresponding to ignored clusters. 
Next, we define a function $f: [k] \setminus I \mapsto [k]$ with $f(i) = \arg \min_{j \in k} d(c_i^*,c_j)$ for every $i \in [k] \setminus I$. Intuitively, the function $f$ maps each optimal cluster that is not  ignored to the closest candidate center (cf. \cref{fig:matching}). Now, for each candidate center, either zero, one or more than one optimal cluster maps to it. If none of the optimal clusters map to it, then we call the candidate center \emph{lonely} and we denote the indices of all the lonely candidate centers by $L$. 
If exactly one optimal cluster maps to a given candidate center, then we call the candidate center \emph{matched} and $H$ denotes the set of indices of all the matched candidate centers. 
Without loss of generality, we assume that $f(i) = i$ for every $i \in H$. Finally, if more than one optimal cluster maps to a candidate center, then we call it \emph{popular}. We  refer to an optimal cluster $X_i^*$ with cluster index $i \in [k] \setminus I$ as a $\emph{cheap}$ cluster if $X_i^*$ maps to some popular candidate center with corresponding index $j$ and, moreover, $X_i^*$ has the smallest cost with respect to $C$ among all optimal clusters that are mapped to the popular candidate center with index $j$. 
We let $T$ denote the set of indices corresponding to cheap clusters. Let $b: [k] \setminus T \mapsto H \cup L$ be an arbitrary bijection with $b(i) = i$ for every $i \in H$ (cf. \cref{fig:matching}). That is, each optimal cluster center that is not cheap gets uniquely assigned to a candidate center which is not popular while preserving the matching between matched cluster centers. Note that such a bijection exists since the number of cheap optimal clusters is equal to the number of popular candidate centers.

\begin{figure}[h!]
    \centering
    \includegraphics[width=\textwidth]{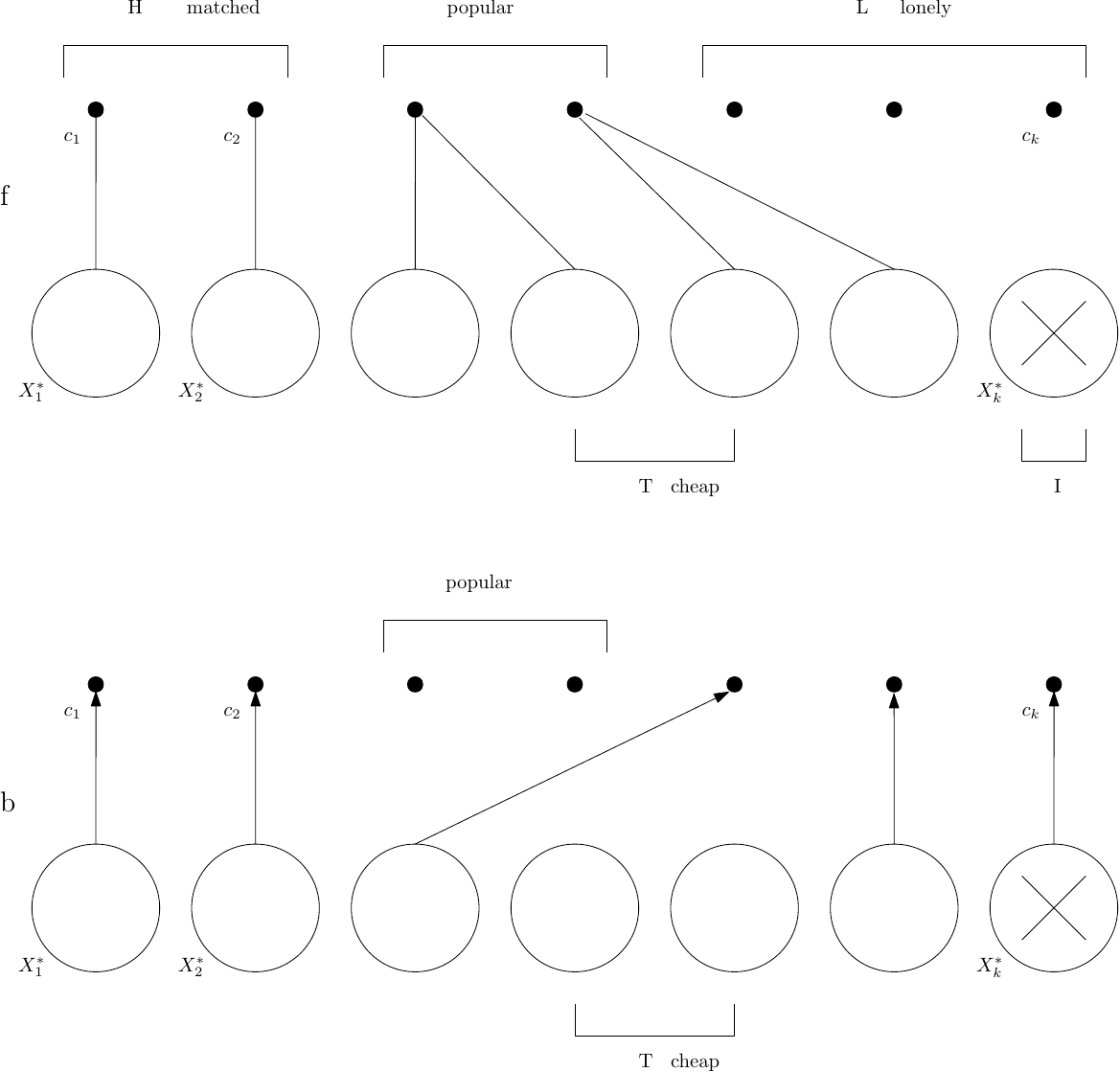}
    \caption{The figure captures several definitions. First, function $f$ maps each optimal cluster $X_i^*$ to its closest candidate center $c_j \in C$. 
    Candidate centers with exactly one optimal cluster matched to them are denoted by $H$ and for convenience we assume that $f$ maps $X_i^*$ to $c_i$ there. Then there are popular centers with more than one match and for each such center we identify the cheapest matched cluster and put it in the set $T$. Finally, there are lonely centers with no matches. Ignored optimal clusters are not matched to any center by $f$. 
    Second, we construct a bijection $b$ that is used later for a double counting argument. In this bijection, we set $b(i) = f(i) = i$ for clusters $i$ matched to centers in $H$. Then, we omit cheap clusters $T$ and popular candidate centers, and extend $b$ to an arbitrary bijection between the remaining clusters and lonely candidate centers. }
    \label{fig:matching}
\end{figure}

For $i \in H \cup L$, we define 
$$\out(c_i) = \{x \in \inn_{10\Theta}(\Xout^*,C): i = \arg \min_{j \in [k]} \phi(x,c_i)\}.$$ 
That is, $\out(c_i)$ contains the real outliers that are closest to the candidate center $c_i$ and which have a squared distance of at most $10\Theta$ to $c_i$. 
%It might be much more natural to consider points with a squared distance of at mot $\Theta$ instead of $10\Theta$. 
The factor of $10$ comes into play because of the following lemma. 

\begin{lemma}
	\label{lemma:ignored}
	Let $C$ denote the current set of candidate centers and assume there exists some $x \in X^*_i$ with $\phi(x,C) \geq 10\Theta$. Then, $X^*_i$ is ignored by $C$.
\end{lemma}
\begin{proof}
Let $x' \in X_i^*$ be arbitrary and let $c'$ denote the candidate center in $C$ that is closest to $x'$. We have

\begin{align*}
	||x'-c'|| &\geq ||x-c'|| - ||x'-x|| \\
			  &\geq ||x-c'|| - ||x'- c_i^*|| - ||x- c_i^*|| \\  &\geq\sqrt{10\Theta} - \sqrt{\Theta} - \sqrt{\Theta}
			   &&x,x' \notin X^*_{far} \\
			   & > \sqrt{\Theta}
\end{align*}
and therefore $\phi(x',C) > \Theta$. Hence, $X_i^*$ is  ignored.
\end{proof}

Now we define the notion of reassignment cost, similarly as it was done in \cite{lattanzi2019better}. For each lonely candidate center $c_\ell$, $\ass(X,C,c_\ell)$ is an upper bound for the increase in $\tau$ cost due to removing $c_\ell$ from the current set of candidate centers $C$. Similarly, for each matched candidate center $c_h$, $\ass(X,C,c_h)$ is an upper bound for the increase in cost due to removing $c_h$ from the current set of candidate centers $C$, ignoring the cost increase for all points in the optimal cluster $X_h^*$.

\begin{definition}
Let $h \in H$ be arbitrary and let $X^*_h$ be the optimal cluster that is captured by the center $c_h$. The reassignment cost of $c_h$ is defined as

\begin{align*}
&\ass(X,C,c_h) = \cost{\Xin^* \setminus X_h^*}{C \setminus \{c_h\}} - \cost{\Xin^* \setminus X_h^*}{C}  + \Theta |X_{far,h}| \\
&+ \left( \Theta |\out(c_h
)|- \tau(\out(c_h), c_h)\right) .
\end{align*}

For $\ell \in L$, the reassignment cost of $c_\ell$ is defined as 

\begin{align*}
    &\ass(X,C,c_\ell) = \cost{\Xin^*}{C \setminus \{c_\ell\}} - \cost{\Xin^* }{C} + \Theta |X_{far,\ell}| \\& 
    + \left( \Theta |\out(c_\ell)|- \tau(\out(c_\ell), c_\ell)\right).
\end{align*}

\end{definition}

Intuitively, the first two terms in the definition correspond to the increase of the cost for all the points in $X_{in}^* \setminus X_h^*$ and $X_{in}^*$, respectively. The third term corresponds to the increase of the cost for all points in $X^*_{far}$. This increase is bounded by $\Theta$, their maximal cost. Finally, the last two terms in the definition below upper bound the increase of the cost of outliers. Here we again crudely upper bound the cost of each outlier after removing the candidate center by $\Theta$.

\begin{fact}
\label{fact:reassign}
For a matched cluster $h$ we have
\[
\ass(X,C,c_h) \ge \tau(X \setminus X_h^*,C \setminus \{c_h\}) - \tau(X \setminus X_h^*,C).
\]

Similarly, for a lonely cluster $\ell$ we have

\[
\ass(X,C,c_\ell) \ge \tau(X,C \setminus \{c_\ell\}) - \tau(X,C). 
\]
\end{fact}
\begin{proof}
See above discussion. 
\end{proof}

\begin{lemma}[cf. Lemma 4 in \cite{lattanzi2019better}]
\label{lem:reassignmentcost}
For $i \in H \cup L$ we have

\begin{align*}
\ass(X,C,c_i) &\leq \frac{21}{100} \tau(\inn_\Theta(X_i,C),C) + 24 \phi(\inn_\Theta(X_i, C) ,C^*) + \Theta |X_{far,i}|  \\
&+\Theta|\out(c_i
)|  - \tau(\out(c_i), c_i).
\end{align*}
\end{lemma}

\begin{proof}
We only present the case $i \in H$, as the case $i \in L$ is similar and even easier. We observe that $\ass(X,C,c_i) = \tau(X_i \setminus X^*_i,C \setminus \{c_i\}) - \tau(X_i \setminus X_i^*,C) + \Theta |X_{far,i}| + \Theta|\out(c_i
)|  - \tau(\out(c_i), c_i)$, since vertices in clusters other than $X_i$ will still be assigned to their current center. Hence, it suffices to show that $$\cost{X_i \setminus X^*_i} {C \setminus \{c_i\}} - \tau(X_i \setminus X_i^*,C) \leq \frac{21}{100} \tau(\inn_\Theta(X_i,C),C) + 24 \tau(\inn_\Theta(X_i,C), C^*).$$ 
As

\[\cost{\outt_\Theta(X_i \setminus X^*_i,C)} {C \setminus \{c_i\}} - \cost{\outt_\Theta(X_i \setminus X_i^*,C)}{C} = 0,\]

this is equivalent to showing that

\begin{align*}
    &\cost{\inn_\Theta(X_i \setminus X^*_i,C)} {C \setminus \{c_i\}} - \cost{\inn_\Theta(X_i \setminus X_i^*,C)}{C} \\
    &\leq \frac{21}{100} \tau(\inn_\Theta(X_i,C),C) + 24 \tau(\inn_\Theta(X_i,C), C^*).
\end{align*}

The cost contribution of each point in $\inn_\Theta(X_i \setminus X_i^*,C)$ with respect to $C$ is equal to the squared distance to the closest center in $C$ and hence the analysis of Lattanzi and Sohler more or less directly applies. 
For the sake of completeness, we still present their analysis, which makes use of the following lemma.

To upper bound the cost increase one incurs by removing $c_i$, we assign each point $p \in \inn_\Theta(X_i \setminus X_i^*,C) \cap X^*_j$, $j \neq i$, to the candidate center that captures the cluster $X^*_j$. By computing an upper bound for the clustering cost with respect to that assignment, we directly get an upper bound for $\cost{\inn_\Theta(X_i \setminus X_i^*,C)}{C \setminus \{c_i\}}$. We upper bound the cost increase in two steps. First, we move each point $p \in \inn_\Theta(X_i \setminus X_i^*,C) \cap X^*_j$, $j \neq i$, to the cluster center $c^*_j$ of $X^*_j$ and we call the resulting multiset $Q_i$. Let $q_p$ denote the point in $Q_i$ corresponding to $p$. Among the centers in $C$, the closest center to $q_p$ is now the center that captured the optimal center $X^*_j$. Thus, it is not equal to $c_i$, as $c_i$ captures exactly one optimal center, namely $X^*_i$. Thus, we obtain

\begin{align}
\label{eq:vladimir_misik}
\cost{q_p}{C \setminus \{c_i\}} - \cost{p}{C}  &= \cost{q_p}{C} - \cost{p}{C} && \text{$q_p$ not assigned to $c_i$}\\
&\leq \frac{1}{10} \cost{p}{C} + 11 \cost{p}{q_p} && (\text{\cref{fact:generalised_triangle_inequality,fact:thresholded_triangle_inequality} with $\eps = 1/10$})\nonumber\\
&= \frac{1}{10} \cost{p}{C} + 11 \cost{p}{C^*} . \nonumber
 \end{align}
 This implies
 
 \begin{align}
 \label{eq:radim_hladik}
     \tau(q_p,C \setminus \{c_i\}) = \tau(q_p,C)  \leq \frac{11}{10} \tau(p,C) + 11 \tau(p,C^*).
 \end{align}

 Next, we upper bound the cost increase by moving each point $q_p \in Q_i$ back to its original position $p$. Formally, we bound
 
 \begin{align}
 \label{eq:michal_prokop}
     &\cost{p}{C \setminus \{c_i\}} - \cost{q_p}{C \setminus \{c_i\}}  \\
     &\leq \frac{1}{10}\cost{q_p}{C \setminus \{c_i\}} + 11 \cdot \cost{p}{q_p} && \text{\cref{fact:generalised_triangle_inequality,fact:thresholded_triangle_inequality}}\nonumber\\
     &\leq \frac{1}{10}\left(\frac{11}{10}\cost{p}{C} + 11 \cdot \cost{p}{C^*}\right) + 11 \cdot \cost{p}{C^*} && \text{\cref{eq:radim_hladik}}\nonumber\\
     &= \frac{11}{100} \cost{p}{C} + 13 \cdot \cost{p}{C^*}\nonumber
 \end{align}

Combining the two inequalities \cref{eq:vladimir_misik,eq:michal_prokop} yields

\begin{align}
\label{eq:john_lennon}
&\cost{p}{C \setminus \{c_i\}} - \cost{p}{C} \\
&= \left(   \cost{p}{C \setminus \{c_i\}} - \cost{q_p}{C \setminus \{c_i\}}\right) + \left( \cost{q_p}{C \setminus \{c_i\}} - \cost{p}{C}\right) \nonumber\\
&\leq  \frac{11}{100} \cost{p}{C} + 13 \cdot \cost{p}{C^*} +  \frac{1}{10} \cost{p}{C} + 11 \cdot \cost{p}{C^*} &&\text{\cref{eq:vladimir_misik,eq:michal_prokop}} \nonumber\\
&\leq \frac{21}{100}\cost{p}{C} + 24 \cdot \cost{p}{C^*}.\nonumber 
\end{align}

Summing up the inequality \cref{eq:john_lennon} for all points $p \in  \inn_\Theta(X_i \setminus X_i^*,C)$ then yields

\begin{align*}
    &\cost{ \inn_\Theta(X_i \setminus X_i^*,C)}{C \setminus \{c_i\}} - \cost{ \inn_\Theta(X_i \setminus X_i^*,C)}{C} \\
    &\leq \frac{21}{100}\cost{ \inn_\Theta(X_i \setminus X_i^*,C)}{C} + 24 \cost{ \inn_\Theta(X_i \setminus X_i^*,C)}{C^*}\\
    &\leq \frac{21}{100}\cost{ \inn_\Theta(X_i \setminus X_i^*,C)}{C} +  24 \phi(\inn_\Theta(X_i, C), C^*)\\
\end{align*}

and therefore

\begin{align*}
\ass(X,C,c_i) 
&\leq \frac{21}{100} \tau(\inn_\Theta(X_i,C),C) + 24 \phi(\inn_\Theta(X_i, C), C^*) \\
&+ \Theta|X_{far,i}| + \Theta|\out(c_i
)|  - \tau(\out(c_i), c_i).
\end{align*}

as needed.
\end{proof}

\subsection{Decreasing Cost to $O(1/\eps)$}
\label{sec:appendix_localsearch_easy}

In this subsection, we will prove that while the cost of our solution is $\Omega(\OPT/\eps)$, one local search step decreases the cost by a factor of $1 - \Theta(1/k)$ with constant probability (\cref{lemma:easyonestep}). This means that after $\tilde{O}(k)$ steps, we have $\tau_\Theta(X, C) = O(\OPT/\eps)$ with constant probability (\cref{lem:phase1}). This, by itself, is the easy part and our proof is essentially the analysis of Lattanzi and Sohler \cite{lattanzi2019better}. The hard part is then the generalization to \cref{lemma:onestephard} in \cref{sec:appendix_localsearch_hard}, which, roughly speaking, shows that while we have $(1+\eps)z$ points with a big cost, we still do substantial progress in each step.  

For the next definition, recall that $b$ is the bijection between optimal clusters and candidate cluster centers that we defined earlier.

\begin{definition}
\label{def:awesome}
A cluster index $i \in [k] \setminus T$ is called \emph{awesome}, if 
\[
\cost{X_i^*}{C} - \ass(X, C, c_{b(i)}) - 9 \cost{X_i^*}{ c_i^*} > \frac{\cst(X, C)}{100 k}.\]
\end{definition}

Intuitively, awesome clusters are such that sampling a point $c$ from them leads to high improvement in cost if the sampled center is sampled close to the mean $c_i^*$. 
In particular, if $\tau(X_i^*, c) \le 9\tau(X_i^*, c^*_i)$, swapping $c_{b(i)}$ with $c$ leads to improvement of $\tau(X,C)/(100k)$ in cost. 
The next lemma shows that we sample from an awesome cluster with constant probability. 

\begin{lemma}
\label{lem:sample_from_awesome}
Let $C$ denote the current set of candidate centers. Let $\eps \in (0,1]$ be arbitrary and $\beta := \frac{300}{\eps}$.
If $\cst(X, C) \ge 100\beta \opt$ and $\Theta \le 2\frac{\beta\opt}{z} $, then
\[
\sum_{\text{$i \in [k] \setminus T$, $i$ is awesome} } \cost{X_i^*}{C} \ge \frac{1}{10} \cst(X, C). 
\]
\end{lemma}

\begin{proof}
We bound the cost of clusters that are not awesome. 

\begin{align*}
    &\sum_{\text{$i \in [k] \setminus T$, $i$ is not awesome} } \cost{X_i^*}{C}\\
     &\le  \sum_{i \in [k] \setminus T} \ass(X, C, c_{b(i)})  + 9\opt + \frac{\tau(X,C)}{100} &&\text{\cref{def:awesome}} \\
    &\le \sum_{i \in [k] \setminus T} \Big( \frac{21}{100} \cst(\inn_\Theta(X_{b(i)},C), C) + 24 \phi(\inn_\Theta(X_{b(i)}, C) ,C^*) + \Theta |X_{far,b(i)}|  && \text{\cref{lem:reassignmentcost}} \\
    & +\Theta \cdot |\out(c_{b(i)})|- \tau(\out(c_{b(i)}),c_{b(i)})\Big) + 9\opt + \frac{\tau(X,C)}{100} \\
    &\leq \sum_{i \in H \cup L} \Big( \frac{21}{100} \cst(\inn_\Theta(X_{i},C), C) + 24 \phi(\inn_\Theta(X_i), C) ,C^*) + \Theta |X_{far,i}| \\ & + \Theta \cdot |\out(c_{i})|- \tau(\out(c_{i}),c_{i})\Big) + 9\opt + \frac{ \tau(X,C)}{100} && \text{$b([k]\setminus T) = H\cup L$} \\
    &\le \frac{21}{100} \tau(X,C) + 24 \OPT + \opt + z\Theta +  9 \opt  + \frac{\tau(X,C)}{100} && \Theta|X_{far,i}| \leq \OPT, X_i \subseteq X_{in}^*\\
    &\leq \frac{22}{100} \tau(X,C) + (2\beta + 34)\opt && \Theta \leq 2\frac{\beta \OPT}{z} \\
    &\leq \frac{30}{100} \tau(X,C). && \beta \geq 300, \tau(X,C) \geq 100\beta \OPT
\end{align*}

Observe that from the definition of cheap clusters it follows that 
\begin{align}
\label{eq:paul_mccarthy}
    \sum_{i \in T} \cst(X_i^*,C) \le \frac{1}{2}\tau(X,C)
\end{align}
because every cheap cluster can be matched with a different cluster with at least as big cost that is also matched to the same popular cluster center. 

Finally, we bound
\begin{align*}
    &\sum_{\text{$i \in [k] \setminus T$, $i$ is awesome} } \cost{X_i^*}{C} \\
    &= \cst(\Xin^*, C) - \sum_{\substack{\text{$i \in [k] \setminus T$},\\ \text{$i$ is not awesome}} } \cost{X_i^*}{C} - \sum_{i \in T} \cst(X_i^*,C) \\
    &\geq \tau(X,C) - \tau(\Xout^*,C) - \tau(X_{far}^*,C) \\& - \frac{30}{100}\tau(X,C) - \frac{1}{2}\tau(X,C) && \text{\cref{eq:paul_mccarthy}}  \\
    &\geq \frac{20}{100}\tau(X,C) - z\Theta - \opt && \tau(\Xout^*,C) \le z\Theta, \tau(X_{far}^*,C) \le \OPT \\
    &\geq \frac{1}{10}\tau(X,C). && \cst(X, C) \ge 100\beta \opt, \Theta \le 2\frac{\beta\opt}{z}
\end{align*}

%Note that $\sum_{i \in T} \cst(X_i^*,C) \leq  \frac{1}{2}\tau(X,C)$ more or less directly follows from the definition of cheap clusters.

\end{proof}

\begin{lemma}
\label{lemma:easyonestep}

Let $C$ denote the current set of candidate centers. Let $\eps \in (0,1]$ be arbitrary and $\beta := \frac{300}{\eps}$. Suppose that $\cost{X}{C} \ge100\beta \cdot \OPT$ and $\Theta \leq 2 \frac{\beta \opt}{z}$. Then, with probability at least $1/100$, one local search step of \cref{alg:noisylocalsearch} results in a new set of candidate centers $C'$ with $\cst(X, C') \le (1 - 1/(100k)) \cst(X, C)$. 
\end{lemma}

\begin{proof}
Let $c'$ denote the point sampled by \nlsp. The probability that $c' \in X^*_i$ for some $i$ with $i$ being an awesome cluster index is at least
$$\frac{(1/10) \cost{X}{C}}{\cost{X}{C}} = \frac{1}{10}$$
according to \cref{lem:sample_from_awesome}.
Let $i$ denote an arbitrary awesome cluster index. \cref{lem:8apx} implies that
$$\E[\cost{X_i^*}{c'}|c' \in X_i^*] \leq 8 \cdot \cost{X_i^*}{c_i^*}.$$
Hence, by a simple application of Markov's inequality, we obtain
$$\P[\cost{X_i^*}{c'} \leq 9 \cdot \cost{X_i^*}{c_i^*}|c' \in X_i^*]  \geq \frac{1}{9}.$$
Putting things together, the probability that the sampled point $c'$ is contained in $X_i^*$ for some awesome cluster index $i$ and it additionally holds that 
\begin{align}
\label{eq:queen}
    \cost{X_i^*}{c'} \leq 9 \cdot \cost{X_i^*}{c_i^*}
\end{align} 
is at least
$$\frac{1}{10} \cdot \frac{1}{9} > \frac{1}{100}.$$
Consider the case that  $i \in H$. We can upper bound $\cost{X}{C'}$ as follows:

\begin{align*}
\cost{X}{C'} &\leq \cost{X}{(C \setminus \{c_{b(i)}\}) \cup \{c'\}} && \text{swap $c_{b(i)}$ and $c'$}\\
&= \cost{X}{C} + \left( \cost{X}{(C \setminus \{c_{b(i)}\}) \cup \{c'\}}- \cost{X}{C}\right) \\
&=  \cost{X}{C} +  \cost{X \setminus  X_i^*}{(C \setminus \{c_{b(i)}\}) \cup \{c'\}} - \cost{X \setminus X_i^*}{C} \\
& + \cost{X_i^*}{(C \setminus \{c_{b(i)}\}) \cup \{c'\}}- \cost{X_i^*}{C} \\
&\le  \cost{X}{C} + \ass(X,C,c_i) +  \tau(X_i^*,c') - \tau(X_i^*,C) && \text{\cref{fact:reassign}}\\
&\le  \cost{X}{C} + \ass(X,C,c_i) +  9\tau(X_i^*,c_i^*) - \tau(X_i^*,C) && \text{\cref{eq:queen}} \\
&\leq \tau(X,C) - \frac{\tau(X,C)}{100k} && \text{\cref{def:awesome}}
\end{align*}

The case $i \notin H$ is very similar. Thus, we obtain that with probability at least $1/100$ we have

$$\cost{X}{C'} \leq \left(1 - 1/(100k)\right)\cost{X}{C},$$

as desired.
\end{proof}

\begin{lemma}
\label{lem:phase1}
Let $C$ denote the current set of candidate centers. Let $\eps \in (0,1]$ be arbitrary and $\beta := \frac{300}{\eps}$. Let $\alpha > 0$ be arbitrary such that $\cost{X}{C} \leq \alpha \cdot \beta \cdot \opt $. Furthermore, assume that $\Theta \leq 2 \frac{\beta \opt}{z}$. Let $C_0 := C$ and for $t \geq 0$, let $C_{t+1}$ denote the set of centers obtained by running \cref{alg:noisylocalsearch} with input centers $C_t$. Then, with positive constant probability $p >0$, we have $\cost{X}{C_T} \leq 100 \beta \opt$ with $T := \lceil 1000 \cdot 100k \log \alpha \rceil = O(k \log (\alpha))$.
\end{lemma}
\begin{proof}
For $t \geq 0$, let $X_t$ denote the indicator variable for the event that $\cost{X}{C_{t+1}} > \max((1-1/(100k))\cost{X}{C_{t}}, 100\beta \opt)$. According to \cref{lemma:easyonestep}, $\E[X_t] \leq \frac{99}{100}$. Hence, with $X := \sum_{t=0}^{T-1} X_t$, we have $\E[X] \leq \frac{99}{100}T$. Thus, Markov's Inequality implies that there exists a constant $p > 0$ such that $Pr[X \leq \frac{999}{1000}T] \geq p$. In that case, there are at least $100k\log \alpha$ iterations $t$ for which $X_t = 0$ and therefore

\begin{align*}
\cost{X}{C_T} &\leq \max(\left(1 - 1/(100k)\right)^{100k \log \alpha}\cost{X}{C_0},100\beta\opt) \\
&\leq \max(e^{-\log \alpha}\alpha\beta\opt,100\beta\opt)  \\
&= 100\beta\opt,
\end{align*}
as desired.
\end{proof}

\subsection{Decreasing Number of Outliers to $(1+\eps)z$}

\label{sec:appendix_localsearch_hard}

We are now ready to prove \cref{lemma:onestephard}. It roughly states that with probability $\Omega(\eps)$ we sufficiently improve our solution, provided that there are at least $(1+\eps)z$ points with a squared distance of at least $10\Theta$ to the current solution. \cref{lem:phase2} then follows as a simple consequence of \cref{lemma:onestephard}. Roughly speaking, it states that starting with a set of centers $C$ with a cost of $O((1/\eps) \OPT)$, within $O(k/\eps)$ local search steps, one obtains a set of centers $C'$ such that at most $(1+\eps)z$ points have a squared distance of at least $10\Theta$ to $C'$, with positive constant probability. 

Let $g$ be the number of points in $\Xout^*$ that have a squared distance of less than $10\Theta$ to the closest center in $C$. Note that by our definition of $\out(c)$, we have
\begin{align}
    \label{eq:rammstein}
    %\sum_{i \in H \cup L} |\out(c_i)| \le 
    \sum_{i \in [k]} |\out(c_i)| = g \le z.
\end{align}

\begin{definition}
\label{def:good}

A cluster index $i \in [k] \setminus T$ is called \emph{good}, if %there exists $\ell \in L$ such that
\[
\cost{X_i^*}{C} - \ass(X, C, c_{b(i)})  - 9 \cost{X_i^*}{c_i^*} > \frac{\tau(X,C) - z \Theta}{100k}.
\]
\end{definition}

%Assume $\sum_{h \in H} \cost{P_h^*}{C} > \frac{1}{3} \cost{\Xin}{C}$. 

\begin{lemma}
\label{lem:sample_from_good}
%If $\sum_{h \in H} \cost{P_h^*}{C} > \frac{1}{3} \cost{\Xin}{C}$, 
Let $C$ denote the current set of candidate centers. Let $\eps \in (0,1]$ be arbitrary and $\beta := \frac{300}{\eps}$. If for some $\eps' \ge \eps$, there are $(1+\eps')z$ different points with a squared distance of at least $10\Theta$ to the closest candidate center in $C$ and %$(1-\eps/3)\frac{\beta\opt}{z} \le 
$\Theta \ge \frac{\beta\opt}{z} $, then
\[
\sum_{\text{$i \in [k] \setminus T$, $i$ is good} } \cost{X_i^*}{C} \ge \frac{\tau(X,C) - z \Theta}{10}.\]
\end{lemma}

\begin{proof}

Recall our assumption that exactly $(1+\eps')z$ points have a squared distance of at least $10\Theta$ to the set $C$ of candidate centers, i.e., 
\begin{align}
    \label{eq:paul_lennon}
    |\outt_{10\Theta}(X,C)| = (1+\eps')z.
\end{align}
Moreover, we defined $g$ as the number of points in $X_{out}^*$ such that they have squared distance of less than $10\Theta$ to $C$, so 
\begin{align}
\label{eq:john_mccarthy}
    |\outt_{10\Theta}(X_{out}^*,C)| = z-g. 
\end{align}
Hence
\begin{align}
    \label{eq:Bach}
    &|\outt_{10\Theta}(\Xin^*,C)|\\
    &= |\outt_{10\Theta}(X,C)| - |\outt_{10\Theta}(\Xout^*,C)| - |\outt_{10\Theta}(X^*_{far},C)|\nonumber\\
    &\geq   (1+\eps')z - (z -g) - |X^*_{far}| && \text{\cref{eq:paul_lennon,eq:john_mccarthy}} \nonumber \\
    &\geq \eps' z + g - \frac{\opt}{\Theta}. && \text{$\OPT \ge |X_{far}^*|\Theta$} \nonumber
\end{align}

We will use the following bound
\begin{align}
    \label{eq:ringo_star}
    &\frac{\tau(X,C) - z\Theta}{100} \\
    &\le \frac{\tau(\Xin^*,C) + \tau(X^*_{far},C) + \tau(\inn_{10\Theta}(\Xout^*,C),C) - g\Theta}{100}\nonumber\\
    &\le \frac{\tau(\Xin^*,C) +  \tau(\inn_{10\Theta}(\Xout^*,C),C) - g\Theta}{100} + \OPT/100. \nonumber
\end{align}

Now we bound
\begingroup
\allowdisplaybreaks
\begin{align*}
    &\sum_{\text{$i \in [k] \setminus T$, $i$ is not good} } \cost{X_i^*}{C}\\
    &\le \Big( \sum_{i \in [k] \setminus T} \ass(X, C, c_{b(i)}) \Big)  + 9\opt + \frac{\tau(X,C) - z\Theta}{100} &&\text{\cref{def:good}} \\
    &\le \sum_{i \in [k] \setminus T} \Big( \frac{21}{100} \cst(\inn_\Theta(X_{b(i)},C), C) + 24 \phi(\inn_\Theta(X_{b(i)},C), C^*) \\& + \Theta |X_{far,b(i)}|  +\Theta \cdot |\out(c_{b(i)})| \\ &  - \tau(\out(c_{b(i)}),c_{b(i)})\Big) + 9\opt + \frac{\tau(X,C) - z\Theta}{100}  &&\text{\cref{lem:reassignmentcost}} \\
    &= \sum_{i \in H \cup L} \Big( \frac{21}{100} \cst(\inn_\Theta(X_{i},C), C) + 24 \phi(\inn_\Theta(X_{i},C), C^*) \\& +   \Theta |X_{far,i}|  + \Theta \cdot |\out(c_{i})| \\& - \tau(\out(c_{i}),c_{i})\Big) + 9\opt + \frac{\tau(X,C) - z\Theta}{100}  && b([k]\setminus T) = H\cup L\\
    &\le \frac{21}{100} \tau(\inn_\Theta(\Xin^*,C), C) + 24 \opt + \opt  \\& + \sum_{i \in H \cup L} \left( \Theta \cdot |\out(c_i)| - \tau(\out(c_i),c_i)\right)  + 9 \opt \\& + \frac{\tau(X,C) - z\Theta}{100} && X_i \subseteq X_{in}^*\\
   &\le \frac{21}{100} \tau(\inn_\Theta(\Xin^*,C), C) + \sum_{i \in [k]} \left( \Theta \cdot |\out(c_i)| - \tau(\out(c_i),c_i)\right)  \\& + 34 \opt  + \frac{\tau(X,C) - z\Theta}{100} \\
      &= \frac{21}{100} \tau(\inn_\Theta(\Xin^*,C), C) + g\Theta -\tau(\inn_{10\Theta}(\Xout^*,C),C)  \\& + 34 \opt +  \frac{\tau(X,C) - z\Theta}{100} && \text{\cref{eq:rammstein}}\\
      &\le \frac{21}{100} \tau(\inn_\Theta(\Xin^*,C), C) + g\Theta  - \tau(\inn_{10\Theta}(\Xout^*,C),C) + 34 \opt + \\& \frac{\tau(\Xin^*,C) +  \tau(\inn_{10\Theta}(\Xout^*,C),C) - g\Theta}{100} + \OPT/100 && \text{\cref{eq:ringo_star}}\\
      & \le \frac{21}{100} \tau(\inn_\Theta(\Xin^*,C), C) + \Theta g  - \frac{1}{10}\tau(\inn_{10\Theta}(\Xout^*,C),C)\\&  + 35 \opt + \frac{\tau(\Xin^*,C) - g\Theta}{100} 
\end{align*}
\endgroup

Next, we make use of the inequality

\begin{align}
\label{eq:Mozart}
\sum_{i \in T} \tau(X_i^*,C) \leq \frac{1}{2}\sum_{i \in [k], \text{$i$ is not ignored}} \tau(X_i^*,C).    
\end{align}

This follows as a cheap cluster is not ignored, and each cheap cluster can be matched with a different cluster that is not ignored, has at least the same cost and is matched to the same popular cluster center.

Crucially, we now use \cref{lemma:ignored} that says that whenever a cluster $X_i^*$ is not ignored, we have $X_i^* = \inn_{10\Theta}(X_i^*, C)$. 

Putting things together, we get
\begingroup
\allowdisplaybreaks
\begin{align*}
    &\sum_{\text{$i \in [k] \setminus T$, $i$ is good} } \cost{X_i^*}{C}\\
    &= \cst(\Xin^*, C) - \sum_{\text{$i \in [k] \setminus T$, $i$ is not good} } \cost{X_i^*}{C} - \sum_{i \in T} \tau(X_i^*, C)\\
    &\ge \tau(\inn_{10 \Theta}(\Xin^*,C), C) + \tau(\outt_{10\Theta}(\Xin^*,C), C)  \\ &- \Big( \frac{21}{100} \tau(\inn_\Theta(\Xin^*,C), C) + \Theta g  - \frac{1}{10}\tau(\inn_{10\Theta}(\Xout^*,C),C) \\& + 35 \opt + \frac{\tau(\Xin^*,C) - g\Theta}{100} \Big)  - \frac{1}{2} \sum_{i \in [k], \text{ $i$ is not ignored}} \tau(X_i^*,C) && \text{\cref{eq:Mozart}}\\
     &\ge \frac{79}{100}\tau(\inn_{10 \Theta}(\Xin^*,C), C) + ((\eps'z + g)\Theta - \OPT) \\ &- \left(\Theta g   - \frac{1}{10}\tau(\inn_{10\Theta}(\Xout^*,C),C) + 35 \opt + \frac{\tau(\Xin^*,C) - g\Theta}{100} \right) - \\& \frac{1}{2}\tau(\inn_{10\Theta}(X^*_{in},C),C) && \text{ \cref{lemma:ignored}} \\
    &\ge  \frac{1}{10}\tau(\inn_{10\Theta}(\Xin^*,C), C) +  \frac{1}{10}\tau(\inn_{10\Theta}(\Xout^*,C),C) + \eps'z \Theta \\&  - 100 \opt - \frac{\tau(\outt_{10\Theta}(\Xin^*,C),C) - g\Theta}{100} \\
    &\ge \frac{1}{10}\tau(\inn_{10\Theta}(X,C), C) + \eps'z\Theta - 101 \opt - \frac{\eps'z \Theta}{100} && \text{\cref{eq:Bach}} \\
    &\geq \frac{1}{10} \left(\tau(\inn_{10\Theta}(X,C),C) + \eps'z\Theta \right) && \Theta \ge \frac{300\opt}{\eps z}\\
    &= \frac{1}{10}\left(\tau(\inn_{10\Theta}(X,C),C) + \tau(\outt_{10\Theta}(X,C),C) - z\Theta \right) && \text{\cref{eq:paul_lennon}} \\
    &= \frac{\tau(X,C) - z\Theta}{10},
\end{align*}
\endgroup
as needed. 
\end{proof}

\begin{lemma}
\label{lemma:onestephard}
Let $C$ denote the current set of candidate centers. Let $\eps \in (0,1]$ be arbitrary and $\beta := \frac{300}{\eps}$. Suppose that for out current candidate centers $C$ we have at least $(1+\eps)z$ points in $X$ that have a squared distance of at least $10\Theta$ to the closest center in $C$. Furthermore, assume that $\Theta \geq \frac{\beta \opt}{z}$. 
Then, with probability at least $\eps/200$, one step of \nlsp results in a new set of candidate centers $C'$ with $\cst(X, C') - z \Theta \le \left(1 - 1/(100k) \right) \left(\cst(X, C) - z\Theta \right)$.
\end{lemma}

\begin{proof}
Let $c'$ denote the point sampled by \nlsp. The probability that $c' \in X^*_i$ for some $i$ with $i$ being a good cluster index is at least

\begin{align*}
&\frac{\frac{\tau(X,C) - z\Theta}{10}}{\tau(X,C)} \geq \frac{(1+\eps)z\Theta - z\Theta}{10(1+\eps)z\Theta} && \tau(X, C) \ge (1+\eps) z \Theta \\
&\geq \frac{\eps}{20}. && \eps \le 1
\end{align*}

Let $i$ denote an arbitrary good cluster index. \cref{lem:8apx} implies that

$$\E[\cost{X_i^*}{c'}|c' \in X_i^*] \leq 8 \cdot \cost{X_i^*}{c_i^*}.$$

Hence, by a simple application of Markov's inequality, we obtain

$$\P[\cost{X_i^*}{c'} \leq 9 \cdot \cost{X_i^*}{c_i^*}|c' \in X_i^*]  \geq \frac{1}{9}.$$

Putting things together, the probability that the sampled point $c'$ is contained in $X_i^*$ for some good cluster index $i$ and it additionally holds that 
\begin{align}
    \label{eq:skrillex}
    \cost{X_i^*}{c'} \leq 9 \cdot \cost{X_i^*}{c_i^*}
\end{align} is at least

$$\frac{\eps}{20} \cdot \frac{1}{9} \geq \frac{\eps}{200}.$$

First, consider the case that  $i \in H$. In that case, we can upper bound $\cost{X}{C'}$ as follows:

\begin{align*}
\cost{X}{C'} &\leq \cost{X}{(C \setminus \{c_{b(i)}\}) \cup \{c'\}} && \text{swap $c_{b(i)}$ and $c'$} \\
&= \cost{X}{C} + \left( \cost{X}{(C \setminus \{c_{b(i)}\}) \cup \{c'\}}- \cost{X}{C}\right) \\
&=  \cost{X}{C} + \cost{X \setminus  X_i^*}{(C \setminus \{c_{b(i)}\}) \cup \{c'\}}- \cost{X \setminus X_i^*}{C}  \\
& + \cost{X_i^*}{(C \setminus \{c_{b(i)}\}) \cup \{c'\}}- \cost{X_i^*}{C}\\
&\le  \cost{X}{C} + \ass(X,C,c_i) + \tau(X_i^*,c') - \tau(X_i^*,C) && \text{\cref{fact:reassign}} \\
&\le  \cost{X}{C} + \ass(X,C,c_i) +  9\tau(X_i^*,c_i^*) - \tau(X_i^*,C) && \text{\cref{eq:skrillex}} \\
&\leq \tau(X,C) - \frac{\tau(X,C) - z\Theta}{100k}. && \text{\cref{def:good}}
\end{align*}

The case $i \notin H$ is very similar. Thus, we obtain

\begin{align*}
    &\cost{X}{C'} - z\Theta 
    \leq \left( \cost{X}{C} - \frac{\tau(X,C) - z\Theta}{100k}\right) - z\Theta \\
    &= \left(1 - 1/(100k)\right)\left( \cost{X}{C} - z\Theta\right),
\end{align*}
as desired.

\end{proof}

\begin{lemma}
\label{lem:phase2}
Let $C$ denote the current set of candidate centers. Let $\eps \in (0,1]$ be arbitrary and $\beta := \frac{300}{\eps}$. Let $\alpha > 0$ be arbitrary such that $\cost{X}{C} \leq \alpha\opt$. Assume that $\Theta \geq \frac{\beta \opt}{z}$. Let $C_0 := C$ and for $t \geq 0$, let $C_{t+1}$ denote the set of centers obtained by running \nlsp with input centers $C_t$. Then, with positive constant probability $p > 0$, there are at most $(1+\eps)z$ points in $X$ that have a distance of at least $10\Theta$ to the closest candidate center in $C_t$ for some $t \leq T$ with $T = O( \log (\alpha)k/\eps)$.
\end{lemma}
\begin{proof}
For $t \geq 0$, we define the potential $\Phi(t) = \cost{X}{C_t} - z\Theta$. Now, let $Y_t$ denote the indicator variable for the event that there are less than $(1+\eps)z$ points that have a squared distance of at least $10\Theta$ to the closest center in $C_t$ or $\Phi(t+1) \leq \left( 1 - 1/(100k) \right)\Phi(t)$. Note that the random variables $Y_0,Y_1,\ldots,Y_{T-1}$ are not  independent. However, \cref{lemma:onestephard} implies that 

$$\E[Y_i|Y_0,Y_1,\ldots,Y_{i-1}] \geq \frac{\eps}{200}$$
for every $i \in \{0,\ldots,T-1\}$ and every realization of random variables $Y_0,Y_1,\ldots Y_{i-1}$.
Thus, the random variable $Y = \sum_{i=0}^{T-1} Y_i$ stochastically dominates the random variable $Y' = \sum_{i=0}^{T-1} Y_i' $, where the $(Y'_i)$'s are independent Bernoulli variables that are equal to $1$ with probability $\frac{\eps}{200}$. Now, let $T' := \lceil 100k \cdot \log (\alpha) + 1\rceil$ and $T := \lceil \frac{2T'}{\eps / 200} \rceil$. By a standard application of the Chernoff bound in \cref{thm:chernoff}, we get

$$\P[Y' \leq T'] \leq P[Y' \leq (1-1/2)\E[Y']] \leq e^{-\Omega(\E[Y'])} = e^{-\Omega(1)}.$$

Thus, there exists a positive constant $p$ such that with probability at least $p$, we have

$$\P[Y \geq T'] \geq p.$$

Assume that $Y \geq T'$. First, consider the case that there exists a $t \in \{0,\ldots,T-1\}$ for which there are less than $(1+\eps)z$ points in $X$ that have a squared distance of at least  $10\Theta$ to the closest center in  $C_t$. In that case we are done. Thus, assume that this does not happen. In particular, this implies that 

$$\Phi(T-1) \geq (1+\eps)z \Theta - z\Theta = \eps \Theta z > \opt.$$

However, we also have

\begin{align*}
\Phi(T-1) & \leq\left( 1 - 1/(100k)\right)^{T'-1} \Phi(0) \\
&\leq \left(1 - 1/(100k) \right)^{100k \cdot \log (\alpha)}\alpha\opt \\
&\leq \opt,
 \end{align*}
a contradiction. This concludes the proof. 
\end{proof}

\subsection{Putting things together}

\begin{lemma}
\label{lemma:good_iteration}
Consider some iteration $i$ for which $\opt \leq \OPT_{guess} \leq 2\opt$. Then, with positive constant probability, there exists some iteration $j$ of the inner loop such that $|\Xout^{i,j}| \leq (1+\eps)z$ and $\phi(X \setminus \Xout^{i,j},C_{i,j})= O(1/\eps) \opt$.
\end{lemma}
\begin{proof}
According to \cite{bhaskara2020, li_penalties2020}, running \cref{alg:kmp} results in a set of centers $C_{i,0}$ such that 
\begin{align*}
    &\E[\tau(X,C_{i,0})] = O(\log k) \cdot \min_{C' \colon |C'| = k} \tau(X,C') \\
    &\leq O(\log k) \left(\opt + z \cdot \Theta \right) \leq O(\log k) \beta \opt.
\end{align*}

By using Markov's inequality, we therefore have $\tau(X,C_{i,0}) = O(\log k)\beta \opt$ with positive constant probability. Conditioned on that event, we use \cref{lem:phase1} to deduce that $\tau(X,C_{i,T}) \leq 100 \beta \opt$ for some $T = O(k \log (O(\log k))) = O(k \log \log k)$ with positive constant probability. Then, we use \cref{lem:phase2} to deduce that  with positive constant probability, there exists some $T' \in O( \log(1/\eps)k/\eps)$ such that there are at most $(1+\eps)z$ points in $X$ with a squared distance of at least $10\Theta$ to the closest center in $C_{i,T + T'}$ and furthermore $$\tau(X,C_{i,T+T'}) \leq \tau(X,C_{i,T}) \leq 100 \beta \opt = O(1/\eps)\opt. $$
Hence, $|\Xout^{i,T+T'}| \leq (1+\eps)z$ and $$\phi(X \setminus \Xout^{i,T+T'},C_{i,T+T'}) \leq 10 \cdot  \tau(X \setminus \Xout^{i,T+T'},C_{i,T+T'}) \leq O(1/\eps)\opt$$ with $T+T' = O(k \log \log k + \log(1/\eps)k/\eps)$, as set in the inner loop of \cref{alg:noisylocalsearch}.
\end{proof}

Finally, we are ready to finish the analysis of \cref{alg:noisylocalsearch} by proving \cref{thm:localsearch}. 

\begin{proof}[Proof of \cref{thm:localsearch}.]
As $k + z < n$, we have 

$$1 \leq \min_{C',\Xout' \colon |C'| = k,|\Xout'| = z} \phi(X \setminus \Xout', C) \leq n\Delta^2.$$

For $i = 0$, we have $2^i = 1$ and for $i = \lceil \log(n\Delta^2))\rceil$, we have

$$2^i = 2^{ \lceil \log(n\Delta^2 ))\rceil}  \geq n\Delta^2.$$

Thus, there exists some $i \in \{0,1,\ldots,\lceil \log(n\Delta^2)\rceil\}$ such that in the $i$-th iteration we have
$OPT_{guess} \in [\opt,2\opt]$ and the statement follows from \cref{lemma:good_iteration}.
\end{proof}

\section{Using the  Metropolis-Hastings algorithm to speed up \cref{alg:kmp}}
\label{sec:appendix_metropolis}
Here we explain how to speed up \cref{alg:kmp} by using the  Metropolis-Hastings algorithm. This was first done in \cite{bachem2016approximate, bachem2016fast} for the classical $k$-means++ algorithm, but there it only leads to rather weak additive error guarantee, which turns out not to be the case for $k$-means with outliers. 

We alter every step of \cref{alg:kmp} as follows. 
Instead of sampling a new point proportional to its $\tau_\Theta$-weight, we instead sample a new point from the following Markov chain: First, we sample a point $x$ according to a \emph{proposal} distribution $q$. %, i.e., $q(x) = 1/n$ for all $x \in X$. 
Then, in the following $T$ steps, we always sample a new point $y$ from $q$ and set $x \leftarrow y$ with probability $\min\left(1, \frac{q(x) \pi(y)}{q(y) \pi(x)} \right)$, where $\pi$ is the target distribution, so in our case, the $\tau_\Theta(\cdot, C) / \tau_\Theta(X, C)$-distribution. 
In other words, we define a Markov chain with transition probability 
\[
P(x,y) = q(y) \cdot \min\left(1, \frac{q(x) \pi(y)}{q(y) \pi(x)} \right).
\]
In \cite{bachem2016approximate}, a uniform proposal distribution $q(x) = 1/n$ for all $x \in X$ was chosen. The advantage of the uniform distribution is that one can sample from it easily, but only arguably rather weak guarantees are known for the resulting algorithm in the context of speeding up $\kmpp$. 
In \cite{bachem2016fast}, $q(x) = \phi(x, c)/ \phi(X,c)$ for a random point $c$ was chosen. 
This distribution can be precomputed in $O(n)$ time and the guarantee of the resulting algorithm matches the guarantee of the original \kmpp algorithm up to an additional additive error term $\delta \phi(X, \mu(X))$, if $T = \tilde{O}(1/\delta)$ is chosen. 

\begin{remark}
One can see (but we will not prove it here) that the analysis of \cite{bachem2016fast} directly generalizes to  $\tau_\Theta$-costs for each non-negative $\Theta$. Choosing $\Theta = \OPT/z$, one obtains $\delta \tau_\Theta(X, \mu(X)) \le \delta \cdot n\Theta = \frac{\delta n \OPT}{z}$. Hence, the choice of $\delta = O(z/n)$ yields that the additional incurred error is of the same order as $\OPT$. Hence, generalizing the result of \cite{bachem2016fast} and using \cref{lem:penalties_to_outliers}, we obtain an $(O(\log k), O(\log k))$-approximation algorithm with time complexity of $\tilde{O}(n + k \cdot (k/\delta)) = \tilde{O}(n + nk^2/z)$, because we need to precompute the proposal distribution $q$ in $O(n)$ time and sampling one point takes $\tilde{O}(1/\delta)$ steps, each implementable in time $O(k)$. Changing proposal distribution to uniform would recover the same guarantees. 
\end{remark}

While in the case of classical $k$-means, one needs to use the more complicated proposal distribution $q(x) = \phi(x, c)/ \phi(X,c)$ to get some guarantees, it turns out that for $k$-means with outliers, uniform proposal suffices. 
To get \cref{thm:fast_algorithm}, we start by proving that \cref{alg:kmp} with $\ell = O(k)$ can be sped up with Metropolis-Hastings algorithm with uniform proposal. 
%In what follows, we always assume $z \ge k$ (otherwise, \cref{thm:fast_algorithm} follows from \cref{thm:localsearch}). 
\begin{theorem}
\label{thm:metropolis_oversampling}
A version of \cref{alg:kmp} with $\ell = O(k)$, $\Theta = \Theta(\OPT/z)$ and where each sampling step is approximated by the Metropolis-Hastings algorithm running for $T = O(n/z)$ steps and a uniform proposal distribution will run in time $\tilde{O}(\frac{nk^2}{z})$ and produce a solution that with positive constant probability satisfies
\[
\tau_\Theta(X_{in}^*, C) = O(\OPT). 
\]
\end{theorem}

To prove \cref{thm:metropolis_oversampling}, we first recall a classical result about the Metropolis-Hastings algorithm \cite{metropolis1953equation, liu1996metropolized}. If the proposal distribution $q$ and the target distribution $\pi$ satisfy $q(x)\ge \delta \pi(x)$ for all $x$ and some $\delta > 0$, then after $T = O(1/\delta)$ steps, the distribution of the Metropolis-Hastings chain dominates $\pi/2$. 

\begin{lemma}[one Metropolis-Hastings step \cite{liu1996metropolized}]
\label{lem:metropolis}
Let $\pi$ denote the target distribution we want to approximate with  Metropolis-Hastings and $p_t = \alpha \pi + (1 - \alpha) p_t'$ for some distribution $p_t'$. Let $q$ be a proposal distribution with $q(x) \ge \delta \pi(x)$ for all $x$. Then, after one step of Metropolis-Hastings, we are in  a distribution $p_{t+1}$ which can be written as $p_{t+1} = \alpha' \pi + (1- \alpha') p_{t+1}'$ for $\alpha' = \alpha + \delta(1-\alpha)$. 
\end{lemma}
\begin{proof}
For any fixed vertex $x$, given that $x$ is the current vertex after $t$ steps, the probability of moving from $x$ to $y$ is 
\[
P(x,y) = q(y) \cdot \min\left(1, \frac{q(x) \pi(y)}{q(y) \pi(x)} \right).
\]
Now, if $q(x)\pi(y) \ge q(y)\pi(x)$, we have $$P(x,y) = q(y) \ge \delta \pi(y)$$ by the definition of $\delta$. 
On the other hand, if $q(x)\pi(y) < q(y)\pi(x)$, we have $$P(x,y) = \frac{q(x) \pi(y)}{\pi(x)} \ge \frac{q(x) \pi(y)}{q(x) / \delta} = \delta \pi(y).$$ 

So, for any $x$ and any $y$ we have $P(x,y) \ge \delta\pi(y)$, and since $\pi$ is a stationary distribution, for any $x$ we have $p_{t+1}(x) \ge \alpha \pi(x) + (1-\alpha)\delta \pi(x)$, as needed. 
\end{proof}

We now prove \cref{thm:metropolis_oversampling}. 

\begin{proof}
Fix one sampling step of \cref{alg:kmp} and suppose that $\tau(\Xin^*, C) \ge 40\OPT$. For $T=O(n/z)$, we will show that we sample from a distribution $p_T$ such that for $\pi(x) = \tau(x, C) / \tau(X, C)$ we have $p_T(x) \ge \pi(x)/2$ for all $x$. Hence, we may interpret sampling from $p_T$ as sampling from $\pi$ with probability $1/2$ and sampling from some other distribution otherwise. Thus, we have essentially reduced the analysis to the one in the proof of \cref{thm:aggarwal}, only loosing a $2$-factor in the number of required iterations. 

The running time of the algorithm is $O(k \cdot \frac{n}{z} \cdot k)$, as we need to sample $O(k)$ points in total, $O(n/z)$ iterations of the Metropolis-Hastings algorithm are necessary for each point sample and in each step we need to compute $\tau_\Theta(x, C)$, which can be done in time proportional to the number of points already sampled.

We now prove that $p_T(x) \ge \pi(x)/2$ for all $x$, provided that $\Theta \in [\OPT/z,2\OPT/z]$. 
First note that $\pi(x) = \tau(x, C) / \tau(X, C) \le \Theta / \tau(X, C)$, so for $\delta = \frac{z}{n}$ we have for any $x$ that 
\[
\delta \pi(x) 
\le \frac{z}{n} \frac{2\OPT/ z}{\tau(X,C)} 
= q(x) \frac{2\OPT}{\tau(X,C)}
\le q(x).
\]
Hence, we may use \cref{lem:metropolis} to conclude that after $O(1/\delta) = O(n/z)$ steps, $p_T(x) \ge \pi(x)/2$ for all $x$, as needed. 
\end{proof}

We are now ready to prove \cref{thm:fast_algorithm} that we restate here for convenience. 

\thmmetropolis*

\begin{proof}
If $z = O(k \log k)$ or $z = O(\log^2 n)$, we can run \cref{alg:noisylocalsearch} with $\eps = 1$ to obtain an $(O(1), O(1))$-approximation in time $\tilde{O}(nk)$ (\cref{thm:localsearch}). Otherwise, we run the Metropolis-Hastings based variant of \cref{alg:kmp} from \cref{thm:metropolis_oversampling} with $\ell = O(k)$ that runs in $\tilde{O}(nk^2/z)$ time for $O(\log n)$ values $\Theta_0 = 2^0/z, \Theta_1 = 2^1/z, \dots, \Theta_{\log(\Delta^2n)} = \Delta^2n/z$. Hence, we get a sequence of sets of centers $Y_0, Y_1, \dots$ with $|Y_i|= O(k)$. 
\cref{thm:metropolis_oversampling} ensures that there will exist one value $\tilde{\Theta}$ with $\OPT/(2z) \le \tilde{\Theta} \le \OPT/ z$ in the set of values $\Theta_0, \dots, \Theta_{\log(\Delta^2n)}$ such that the sped up \cref{alg:kmp} returns, with constant probability, a set $Y$  with $\tau_{\tilde{\Theta}}(X_{in}^*, Y) = O(\OPT)$.  

In \cref{thm:fast_coreset} we give an algorithm with running time $\tilde{O}(nk^2/z)$ that turns the pair $(\tilde{Y}, \tilde{\Theta})$ into a set of cluster centers $\tilde{C}, |\tilde{C}| = k$, such that $\phi^{-A z}(X, \tilde{C}) = O(\OPT)$ for some universal constant $A$, with positive constant probability. 

We apply the algorithm from \cref{thm:fast_coreset} to all pairs $(Y_i, \Theta_i)$ and get a sequence of $O(\log n)$ sets $C_0, C_1, \dots$ with $|C_i|=k$ such that, by \cref{thm:fast_coreset}, with positive constant probability it contains a set $\tilde{C}$  with $\phi^{-A z}(X, \tilde{C}) = O(\OPT)$. 

To conclude, we describe how in time $\tilde{O}(nk/z)$ we can get an estimator $\zeta_C$ of $\phi^{-A z}(X, C)$ for some set of $k$ centers $C$ such that the following two properties are satisfied. 

\begin{align}
    \label{eq:met_prop1}
    \zeta_C \le 4\phi^{-Az}(X,C)\;\;\;\;\;\;\;\;\;\;\text{  with probability $1-1/\poly(n)$}
\end{align}
and

\begin{align}
    \label{eq:met_prop2}
    \zeta_C \ge \frac{1}{4} \phi^{-10Az}(X,C) \;\;\;\;\;\;\;\;\;\; \text{  with probability $1-1/\poly(n)$.}
\end{align}

Then, we can simply estimate the cost of each set of centers $C_i$ and return the one with the smallest estimated cost. With positive constant probability, there is at least one set of centers $C_i$ with  $\phi^{-Az}(X,C_i) = O(\OPT)$. Hence, by \cref{eq:met_prop1} there also exists an estimator $\zeta_{C_i}$, with positive constant probability, such that $\zeta_{C_i} = O(\OPT)$.
In that case, we choose a set of centers $C_\fA$ with $\phi^{-10Az}(X,C_\fA) \le 4\zeta_{C_\fA}  = O(\OPT)$ with constant positive probability, by \cref{eq:met_prop2}. \\

The estimator $\zeta_C$ for $\phi^{-A z}(X,C)$ is constructed as follows. We sample $N = O((n/z)\log^2 n)$ points from $X$ uniformly and independently at random with replacement. We denote the resulting multiset as $X'$. Then, we compute the value $\phi^{-4Az(N/n)}(X',C)$ in time $O(|X'|k) = \tilde{O}(nk/z)$ and set $\zeta_C = \frac{n}{N} \phi^{-4Az(N/n)}(X',C)$. 

We define $X_{in}$ as the closest $Az$ points to the set $C$ and $X_{out} = X \setminus X_{in}$. 
Let us split the points of $X_{in}$ into $\log\Delta$ sets $X_1, X_2, \dots$ such that $x \in X_i$ if $2^i \le \phi(x, C) \le 2^{i+1}$. 
We say a set $X_i$ is \emph{big} if $|X_i| \ge z/\log\Delta$ and \emph{small} otherwise. 
The intuition here is that either a set $X_i$ is big enough so that the number of sampled points from $X_i$ is concentrated around its expectation, or the set is small enough with respect to the number of outliers $z$ we consider. 

More precisely, for every big set $X_i$ we can compute that $\E[|X_i \cap X'|] \geq \frac{z}{n\log\Delta} \cdot \Omega( \frac{n}{z} \log^2(n)) = \Omega(\log n)$, so by \cref{thm:chernoff}, we have 
\begin{align}
    \label{eq:met_concentration}
    \P&\left( \E[|X_i \cap X'|]/2 \le  |X_i \cap X'| \le 2\E[ |X_i \cap X'|]\right) \\
    &\ge 1 - \e^{-\Theta(|X_i \cap X'|)} \nonumber\\
    &=  1 - 1/\poly(n).\nonumber
\end{align} 
Similarly, we can compute
\begin{align}
    \label{eq:met_concentration_outliers}
    \P&\left( \E[|X_{out} \cap X'|]/2 \le  |X_{out} \cap X'| \le 2\E[ |X_{out} \cap X'|] \right) \\
    &\le 1 - \e^{-\Theta(|\Xout \cap X'|)} \nonumber\\
    &=  1 - 1/\poly(n).\nonumber
\end{align}
Furthermore, we have $|\cup_{\text{$i$ is small}} X_i| \leq \log (\Delta) \frac{z}{\log (\Delta)} = z$.
As $\frac{z}{n} \cdot N =  \Omega(\log^2n)$, we can conclude that $X'$ contains at most  $2\frac{z}{n} \cdot N$ points from $\cup_\text{$i$ is small} X_i$ with probability $1 - 1 / \poly(n)$. We now condition on the events from \cref{eq:met_concentration} and \cref{eq:met_concentration_outliers} and we additionally condition on the event that $X'$ contains at most  $2\frac{z}{n} \cdot N$ points from $\cup_\text{$i$ is small} X_i$.

We first prove \cref{eq:met_prop1}. We have $\zeta_C = \frac{n}{N} \phi^{-4Az(N/n)}(X',C) \le 4 \phi^{-Az}(X,C)$. By \cref{eq:met_concentration_outliers}, we have  $|X_{out} \cap X'| \le 2Az(N/n)$, so we can label points in $X_{out} \cap X'$ as outliers. Furthermore, we can additionally label all points in $\left( \cup_{\text{$i$ is small}} X_i \right) \cap X'$ as outliers, as $|\left( \cup_{\text{$i$ is small}} X_i \right) \cap X'| \leq 2z(N/n)$. The number of sampled points from each big $X_i$ is at most $2\left(N/n\right)|X_i|$  by \cref{eq:met_concentration} and the costs of points in $X_i$ differ only by a 2-factor. 

Next, we prove \cref{eq:met_prop2}. Let $z_i'$ be the number of outliers one needs to choose from $X_i \cap X'$ to optimize the cost $\phi^{-4Az(N/n)}(X', C)$. Let us define a set of outliers $X_{out}^\fA$ such that it contains the set $X_{out}$, all points from small sets $X_i$ and $2z_i'(n/N)$ arbitrary points from each big set $X_i$. Then, $|X_{out}^\fA| \le z + z + 2(4Az(N/n))(n/N) \leq 10Az$. To bound the cost $\phi(X \setminus X_{out}^\fA, C)$, we bound the contribution of each set $X_i$. 
We have 
\begin{align*}
(|X_i| - 2z_i'(n/N)) \cdot 2^{i+1}
&\le (2 \cdot \frac{n}{N} |X' \cap X_i| - 2z_i' \frac{n}{N})2^{i+1} && \text{\cref{eq:met_concentration}}\\
&\le 4 \frac{n}{N} 2^i (|X' \cap X_i| - z_i').
\end{align*}
Summing up these contributions, we get $\phi^{-10Az}(X, C) \le 4\frac{n}{N}\phi^{-4Az(N/n)}$, as desired. 
\end{proof}

For the purpose of getting a set of centers that provides an $(O(1), O(1))$-approximation in $\tilde{O}(nk^2/z)$ time, we next provide a simple adaptation of the ``a clustering of a clustering is a clustering'' result of \cite{guha2003streaming}. 

\begin{theorem}
\label{thm:fast_coreset}
Let $X$ be an input set, $\OPT$ the optimal solution to the $k$-means with $z$ outliers objective on $X$ and $\Theta$ such that  $\OPT/(2z) \le \Theta \le \OPT/z$. 
Suppose we have a set $Y, |Y| = O(k)$ such that $\tau_\Theta(X,Y) = O(\OPT)$. Then, in time $\tilde{O}(nk^2/z)$ we can compute a set $C_\fA, |C_\fA|=k$, such that $\tau_\Theta(X,C_\fA) = O(\OPT)$ with positive constant probability. 
\end{theorem}
\begin{proof}
Let $C^*$ be the clustering of an optimal solution. Recall that $\tau_\Theta(X, C^*) \le \OPT + z\Theta = O(\OPT)$. 
We subsample $N = \Theta(\frac{nk}{z}\log n)$ points uniformly and independently with replacement from $X$, thus obtaining a multiset of points $X'$. 
For each point $x \in X'$ we compute its distance to $Y$. This is done in time $\tilde{O}(nk^2/z)$. For $y \in Y$ we define $B'(y)$ as the set of points $x \in X'$ for which $y$ is the closest center and $B(y)$ as the set of points $x \in X$ for which $y$ is the closest center. Moreover, we define $\hat{w}(y) = (n/N)|B'(y)|$ and $w(y) = |B(y)|$. Note that $\E[\hat{w}(y)] = w(y)$. We denote with $y_x$ the point $y \in Y$ for which $x \in B(y)$. 
For each $y \in Y$ we either have $w(y) < z/k$, and then we call the respective set $B(y)$ \emph{small}, or $w(y) \ge z/k$ and the respective set $B(y)$ is called \emph{big}. For big sets we have  $\E[|B'(y)|] \geq (N/n) \cdot z/k = \Omega(\log n)$. 
In this case, an application of Chernoff bounds (\cref{thm:chernoff}) gives that
\begin{align}
\label{eq:avicii}
w(y)/2 \le \hat{w}(y) \le 2w(y)
\end{align}
with probability $1 - \e^{-\Omega(\E[|B'(y)|])} = 1 - 1/\poly(n)$. We condition that this event happens for all $y \in Y$. 
We know that $\sum_{\text{$y$ small}} w(y) = O(z)$, so $\E[\sum_{\text{$y$ small}}\hat{w}(y)] = O(z)$ and by Markov's inequality $\sum_{\text{$y$ small}}\hat{w}(y) = O(z)$ with constant probability. Again, we condition on this event.

We now run some $O(1)$-approximate algorithm for $k$-means with penalties on the weighted instance $(Y, \hat{w}(y))$ that runs in $\tilde{O}(n'k)$ time for input of size $n'$ (e.g., \cref{alg:noisylocalsearch} can be interpreted as such an algorithm by the easy part of the analysis captured by  \cref{lem:phase1};  the hard part of the analysis -- \cref{lem:phase2} -- then talks about the guarantees of the algorithm for $k$-means with outliers). 
As our instance has $n' = \tilde{O}(nk/z)$ points, the time complexity of the algorithm will be $\tilde{O}(nk^2/z)$. The algorithm returns a set $C_\fA$ of $k$ centers that we return. 
We have
\begin{align}
\label{eq:elmundo}
    &\tau_\Theta((Y, \hat{w}), C^*)\\
    &\le \sum_{\text{$y$ small}} \hat{w}(y) \Theta + \sum_{\text{$y$ big}} \hat{w}(y) \tau_\Theta(y, C^*)\nonumber\\
    &\le O(z)\cdot \Theta + 2\sum_{\text{$y$ big}} w(y) \tau_\Theta(y, C^*)&& \text{\cref{eq:avicii}}\nonumber\\
    &\le O(\OPT) + 4 \sum_{\text{$y$ big}} \sum_{x\in B(y)} \tau_\Theta(y, x) + \tau_\Theta(x, C^*)&&\text{\cref{fact:generalised_triangle_inequality}}\nonumber\\
    &\le O(\OPT) + O(\tau_\Theta(X,Y)) + O(\tau_\Theta(X, C^*))
    = O(\OPT).\nonumber
\end{align}
Hence, 
\begin{align*}
    \tau_\Theta(X, C_\fA) 
    &= \sum_{x \in X} \tau_\Theta(x, C_\fA)\\
    &\le 2 \sum_{x \in X} (\tau_\Theta(x, y_x) + \tau_\Theta(y_x, C_\fA) ) &&\text{\cref{fact:generalised_triangle_inequality}} \\
    &= 2\tau_\Theta(X, Y) + 2\tau_\Theta((Y,w), C_\fA)\\
    &\le O(\OPT) + 2\sum_{\text{$y$ small}} w(y) \tau_\Theta(y, C_\fA) + 2\sum_{\text{$y$ big}} w(y) \tau_\Theta(y, C_\fA) \\
    &\le O(\OPT) + O(z)\cdot \Theta + 4\sum_{\text{$y$ big}} \hat{w}(y) \tau_\Theta(y, C_\fA)&& \text{\cref{eq:avicii}}\\
    &= O(\OPT) + O(\tau_\Theta((Y,\hat{w}), C^*)) = O(\OPT),&&\text{\cref{eq:elmundo}}
\end{align*}
as needed. 
\end{proof}

\begin{remark}
We believe, but do not prove, that refining the above proofs and using the results from \cref{sec:appendix_localsearch,sec:appendix_distributed}, we can refine \cref{thm:fast_algorithm} to achieve $(O(1/\poly(\eps)), 1+\eps)$-approximation algorithm with running time $\tilde{O}(nk^2/(z\poly(\eps)))$. 
\end{remark}

\section{Lower bounds}
\label{sec:appendix_lower_bounds}

	In the following, we consider algorithms for the $k$-means/$k$-median/$k$-center with outliers problems in the general metric space query model. The algorithm is only required to output a set of $k$ centers and does not need to explicitly declare points as outliers.

\begin{theorem}[Formal version of \cref{thm:lower_bound_informal}]
	\label{thm:lower_bound_formal}
	Let $\mathcal{A}$ denote a (randomized) algorithm for the $k$-means/$k$-median/$k$-center problem with outliers with a query complexity of $o(nk^2/z)$.  Then, for large enough values $k$, $z$, and $n$ such that $10000k \log k \leq z \leq \frac{1}{10000}n$, there exists a concrete input instance for the $k$-means/$k$-median/$k$-center problem with parameters $k$, $z$ and $n$ such that $\mathcal{A}$ outputs an $(\alpha, \beta)$-approximation with a probability  less than $0.5$ for every $\alpha,\beta \in O(1)$.
\end{theorem}

If one does not assume that the ratio between the maximum and minimum distance is bounded, one can show that the statement even holds for an arbitrary multiplicative factor $\alpha$. For simplicity, we assume in the proof that $\beta = 2$, but the proof works for an arbitrary constant $\beta$.

\begin{proof}
We roughly follow the proof from \cite{mettu2002approximation}. 
Applying Yao's principle\cite{yao1977probabilistic}, it suffices to provide a distribution over metric spaces such that any deterministic algorithm needs $\Omega(nk^2/z)$ queries to succeed with probability $0.5$. 
	
To that end, let $k,z$ and $n$ be given such that $10000k\log k \leq z \leq \frac{1}{10000}n$ and we assume that $k$ is sufficiently large. Let $C = \{c_1,\ldots,c_k\}$ be a set of points such that the distance between any two points in $C$ is $1$. 
Next, we define a probability distribution $\fD$ over $C$. Our  random instance then consists of a multiset of $n$ points $X = \{x_1,\ldots,x_n\}$ with each point being sampled independently according to $\fD$. More precisely, we call $S = \{c_1, \dots, c_{k/2}\}$ the set of \emph{small} points and $B = \{c_{k/2+1}, \dots, c_k\}$ the set of \emph{big} points (for simplicity, we assume that $k$ is even). 
The probability of sampling each small point is set to $\frac{100z}{n\cdot k/2}$ and the probability of sampling each big point is set to $\frac{n-100z}{n\cdot k/2}$. Note that $\frac{k}{2} \cdot \frac{100z}{n \cdot k/2} + \frac{k}{2} \cdot \frac{n-100z}{n\cdot k/2} = 1$. Note that it can (and will) happen that we sample the same point twice, but the algorithm does not know it unless it asks the oracle $\fM$ that for a query $(x,y)$ returns the distance $d(x,y)$ of the two points in the constructed metric space. 
One could perturb this instance slightly in order to get an instance with a lower bound on the minimum distance between two points. We refer to the multiset of points that are a clone of $c_j$ as a small cluster if $j \in [k/2]$ and as a big cluster otherwise. We denote with $E_{i,j}$ the event that the $i$-th point is a clone of $c_j$.\\
	
We consider an arbitrary deterministic algorithm $\fA$ that asks at most $T = \frac{nk^2}{100000z}$ queries to the distance oracle $\fM$. We refer to a query as \emph{informative} if the distance between the two queried points is $0$ and \emph{uninformative} otherwise. We first start by proving a helper lemma.

\begin{lemma}
\label{lem:helper}
	 Assume that among the first $t$ queries of $\fA$, at most $0.1k$ queries involved the $i$-th point and all of those queries were uninformative. Then, conditioned on the first $t$ queries and the answer to those queries, the probability that the $i$-th point belongs to some arbitrary \emph{fixed} small cluster with index $j_{small} \in [k/2]$ is at most $\frac{500z}{kn}$. Moreover, the statement still holds if we additionally condition on an arbitrary cluster assignment of all points, except the $i$-th one, that is consistent with the oracle answers to the first $t$ queries.
\end{lemma}
\begin{proof}
	%We use a Bayesian perspective. 
	Let $Q_t$ denote the event corresponding to the observed interaction between $\fA$ and the distance oracle restricted to the first $t$ queries. To prove the lemma, we will not only condition on $Q_t$, but we additionally fix the randomness of all points except the $i$-th one in an arbitrary way. We denote the corresponding event by $R_{-i}$ and we assume that $\P[Q_t \cap R_{-i}] > 0$. Among the first $t$ queries, the $i$-th point was queried at most $0.1k$ times and all queries were uninformative. Hence, conditioned on $Q_t$ and $R_{-i}$, there are at least $0.5k - 0.1k = 0.4k$ big clusters that the $i$-th point could still be contained in. Let $j_{big}$ denote an arbitrary index of one of those big clusters. Using Bayes rule, we have

	\begin{align*}
		\P[E_{i,j_{small}}|Q_t \cap R_{-i}] & \leq \frac{\P[E_{i,j_{small}}|Q_t \cap R_{-i}]}{0.4k \cdot \P[E_{i,j_{big}}|Q_t \cap R_{-i}]} \\
		&= \frac{\P[Q_t \cap R_{-i}|E_{i,j_{small}}]\P[E_{i,j_{small}}]}{0.4k \cdot \P[Q_t \cap R_{-i}|E_{i,j_{big}}]\P[E_{i,j_{big}}]} \\
		&= \frac{\P[R_{-i}|E_{i,j_{small}}]\P[E_{i,j_{small}}]}{0.4k \cdot \P[R_{-i}|E_{i,j_{big}}]\P[E_{i,j_{big}}]} \\
		&= \frac{\P[E_{i,j_{small}}]}{0.4k\P[E_{i,j_{big}}]} \\
		&= \frac{100z}{0.4k \cdot(n-10z)} \\
		&\leq \frac{500z}{k n}.
	\end{align*}

\end{proof}

Let $\fC_{\fA}$ denote the set of centers that $\fA$ outputs. To simplify the notation, we assume that $\fC_{\fA}$ is a set of $k$ indices between $1$ and $n$ instead of $k$ points. We define $\fC_{\fA,small} = \fC_{\fA} \cap [k/2]$ as the set of indices corresponding to small clusters. The main ingredient of our lower bound argument is to show that $\E[|\fC_{\fA,small}|]$ is small, namely at most $0.1k$. To that end, we partition $\fC_{\fA,small} = \fC_1 \sqcup \fC_2 \sqcup \fC_3$ into three different sets

\begin{align*}
	&\fC_1 := \{i \in \fC_{\fA,small} : \text{the $i$-th point was involved in less than $0.1k$ queries,}\\
	&\text{all being uninformative}\}, \\
	&\fC_2 := \{i \in \fC_{\fA,small} : \text{the $i$-th point was involved in at least $0.1k$ queries,}\\
	&\text{ the first $0.1k$ of those being uninformative}\}, \\
	&\fC_3 := \{i \in \fC_{\fA,small} : \text{Among the first $0.1k$ queries the $i$-th point was involved with,}\\
	&\text{ at least one was informative}\}.
\end{align*}

We prove that the expected size of all three sets is small. We start by bounding $\E[|\fC_1|]$. To that end, let $Q$ denote the event corresponding to the complete observed interaction between $\fA$ and the oracle. As $\fA$ is deterministic, conditioned on $Q$, the set $\fC_{\fA}$ is completely determined. Consider some arbitrary $i \in \fC_{\fA}$ such that the $i$-th point was involved in at most $0.1k$ queries, all being uninformative. \cref{lem:helper} together with a union bound over all small clusters implies that 
\begin{align}
\label{eq:bob_marley}
\P[\text{the $i$-th point belongs to a small cluster} \mid  Q] \leq 0.5k \cdot \frac{500z}{  kn} = \frac{250z}{n}.
\end{align}
Hence, using linearity of expectation, we obtain $\E[|\fC_1|] \leq k \cdot \frac{250z}{n} \leq 0.025k$.

Next, we bound the expected size of $\fC_2$. Note that there can be at most $2 \cdot \frac{nk^2/(100000z)}{0.1k} = \frac{nk}{5000z}$ points that are involved in at least $0.1k$ queries, the first $0.1k$ of those being uninformative. Among those points, the expected fraction  of points belonging to small clusters is at most $\frac{250z}{n}$. This is again a consequence of \cref{eq:bob_marley}, as at the moment the $i$-th point was queried exactly $0.1k$ times, all queries being uninformative, the probability that the point belongs to a small cluster is at most $\frac{250z}{n}$. Hence, we can conclude that $\E[|\fC_2|] \leq \frac{250z}{n} \cdot \frac{nk}{5000z} = 0.05k$. 

Thus, it remains to bound $\E[|\fC_3|]$. To that end, for $t \in [T]$, let $X_t$ denote the indicator variable for the following event: The $t$-th query of $\fA$ is informative, involves a point that was queried less than $0.1k$ times before and all those previous queries were uninformative. 
Again, we use \cref{lem:helper} to obtain that $\E[X_t] \leq \frac{500z}{k \cdot n}$. As $|\fC_3| \leq 2 \cdot \sum_{t \in T} X_t$, linearity of expectation implies that $$\E[|\fC_3|] \leq 2 \frac{nk^2}{100000z} \cdot \frac{500z}{k \cdot n} \leq 0.01k.$$ 

Putting things together, we obtain that $\E[|\fC_{\fA,small}|] \leq 0.1k$. By Markov's inequality, this implies that $\P(|\fC_{\fA,small}| \geq 0.3k) \leq 1/3$. As the expected size of each small cluster is $\frac{200z}{k} \geq 200\log k$, a Chernoff Bound followed by a Union Bound implies that each small cluster contains at least $100z/k$ points with a probability of at least $0.9$. Hence, with a probability of at least $0.9 - 1/3 > 0.5$, each small cluster contains at least $100z/k$ points and $\fA$ outputs at most $0.3k$ centers belonging to small clusters. As $(0.5k - 0.3k) \cdot \frac{100z}{k} > 2z$, this implies that the clustering cost with respect to $\fC_{\fA}$ is not zero, unlike the clustering cost of an optimal solution. 
This finishes the proof.
\end{proof}

\begin{remark}
Note that the same construction can be used to give an $\Omega(k^2/\eps)$ lower bound for $k$-means algorithms that output a solution $C$ with $\phi(X,C) = O(\phi(X, C^*) + \eps\phi(X, \mu(X))$. This shows that the complexity of the \kmc algorithm \cite{bachem2016fast} is essentially the best possible. We omit details. 
\end{remark}

\section{Distributed algorithms}
\label{sec:appendix_distributed}

In this section, we show how to adapt classical streaming algorithm of \cite{guha2003streaming} and popular distributed algorithm \kmpipe to the setting with outliers by extension of \cref{thm:aggarwal}. 
In \cref{sec:appendix_distributed_sites}, we show two ways of, roughly speaking, distributing \cref{thm:aggarwal}. Then, in \cref{sec:appendix_distributed_coordinator}, we show how this set can be used to get distributed algorithms that output $(1+\eps)z$ outliers and are $O(1/\eps)$ or even $O(1)$-approximation. 

The most natural distributed model for the following algorithms is the coordinator model. In this model, the data $X = X_1 \sqcup X_2 \sqcup \dots \sqcup X_m$ is evenly split across $m$ machines. In 1-round distributed algorithm such as \cref{alg:streaming}, the machines perform some computation on their data, and send the result to the coordinator who computes the final clustering. In $t$-round distributed algorithm such as \cref{alg:kmpipe_seeding}, there are $t$ rounds of communication between machines and the coordinator. 

\subsection{Distributed algorithms with bicriteria guarantee}
\label{sec:appendix_distributed_sites}

In this section, we start by presenting \cref{alg:streaming}, a simple variant of the algorithm of \cite{guha2003streaming}. In this algorithm, each machine runs \cref{alg:kmp} with oversampling and sends the resulting clustering, together with some additional information, to the coordinator. Analysis of this algorithm follows from our \cref{thm:aggarwal}. 

\begin{algorithm}[htb]
    \caption{Overseeding from Guha et al. \cite{guha2003streaming}}
   \label{alg:streaming}
   {\bfseries Require} data $X$, threshold $\OPT/(2\eps) \le \Theta \le \OPT/\eps$
\begin{algorithmic}[1]
    \STATE split $X$ arbitrarily into sets $X_1, X_2, \dots, X_{m}$. 
    \FOR{$j \leftarrow 1, 2, \dots, m$ in parallel}
        \STATE $Y_j \leftarrow $ \cref{alg:kmp} with $X_{A\text{\ref{alg:kmp}}} = X_j$, $\ell_{A\text{\ref{alg:kmp}}} = \tilde{O}(k/\eps)$, and $\Theta_{A\text{\ref{alg:kmp}}} = \Theta$. 
        \STATE For each $y \in Y_j$, let $w(y)$ be the number of points $x \in X_j, \tau_\Theta(x, Y_j)< \Theta$, such that $y = \argmin_{y' \in Y_j} \tau_\Theta(x, y')$ with ties handled arbitrarily. Let $X_{out, j}$ be the set of points $x\in X_j$ with $\tau(x, Y_j) = \Theta$.  
    \ENDFOR
    \STATE Each site $j$ sends $(Y_j, w), |\bigcup_j X_{out, j}|$ to the coordinator, who computes the weighted set $(Y, w) = \bigcup_{j = 1}^{m} (Y_j, w)$ and the cardinality of $X_{out} = \bigcup_j X_{out, j}$ as $|X_{out}| = \sum_j |X_{out, j}|$.  %who returns the set of centers $C$ gotten by running \cref{alg:noisylocalsearch} on $(Y,w)$ with $z + O(\eps z) - |\cup_j X_{out,j}|$ outliers.   
\end{algorithmic}
\end{algorithm}

Here, by weighted set $(Y, w)$ we mean a set $Y = \{y_1, y_2, \dots, y_{|Y|}\}$ together with a weighting function $w: \{1, 2, \dots, |Y|\} \rightarrow \mathbb{N}$. Set operations can be naturally extended to the weighted setting and in the next subsection we observe that $k$-means algorithms too.

\begin{theorem}
\label{thm:streaming}
Let $\OPT/(2\eps z) \le \Theta \le \OPT/(\eps z)$. 
Suppose we run \cref{alg:streaming} on $m$ machines with $\ell = O(k/\eps \log(m/\delta))$ and let $Y := \bigcup_j Y_j$. Then with probability $1 - \delta$ we have
\[
\sum_{j=1}^{m} \tau(X_j \cap X^*_{in}, Y_j) =  O(\OPT). 
\]
\end{theorem}
\begin{proof}
Let $z_j = |X_j|$ and recall $C^*$ is the optimal solution for the whole instance $X$. 
By \cref{thm:aggarwal} we have, with probability at least $1 - m\cdot \delta/m = 1 - \delta$ that for every $j$
\[
\tau(X_j \cap X^*_{in}, Y_j)
= O(\inf_{C, |C| = k} \tau(X_j \cap X^*_{in}, C) + \eps z_j\Theta)
= O(\tau(X_j \cap X^*_{in}, C^*) + \eps z_j\Theta). 
\]

This implies that we have
\begin{align*}
&\sum_{j=1}^{m} \tau(X_j \cap X^*_{in}, Y_j)
= \sum_{j=1}^{m} O(\tau(X_j \cap X^*_{in}, C^*) + \eps z_j\Theta)\\
&= O(\tau(X^*_{in}, C^*) + \eps z \Theta) = O(\OPT).
\end{align*}

\end{proof}

We leave the description of the algorithm run by the coordinator to \cref{sec:appendix_distributed_coordinator} and now we show a different way of getting the weighted set $(Y, w)$ with essentially the same guarantees, by adapting the \kmpipe algorithm. 

\paragraph{The \kmpipe algorithm} 

We now show a simple adaptation of the \kmpipe algorithm \cite{bahmani2012scalable}: a popular distributed variant of \kmpp. In \kmpipe, the goal is to get $\tilde{O}(k)$ centers in few distributed steps so as to achieve constant approximation guarantee. This is done by changing \cref{alg:kmp} as follows: in each round we sample $k$ points instead of just one point proportional to the cost, but the algorithm is run only for $O(\log n)$ steps. Our adaptation of \kmpipe algorithm to outliers is below.  

\begin{algorithm}[htb]
    \caption{\kmpipe overseeding}
   \label{alg:kmpipe_seeding}
   {\bfseries Require}~data~$X$,\#~rounds~$t$,~sampling~factor~$\ell$, threshold $\Theta$
\begin{algorithmic}[1]
    \STATE Uniformly sample $x \in X$ and set $Y = \{ x \}$.
    \FOR{$i \leftarrow 1, 2, \dots, t$}
        \STATE $Y' \leftarrow \emptyset$
        \FOR{$x \in X$}
            \STATE Add $x$ to $Y'$ with probability $\min\left( 1, \frac{\ell \tau_\Theta(x, Y)}{\tau_\Theta(X, Y)} \right)$
        \ENDFOR
        \STATE $Y \leftarrow Y \cup Y'$
    \ENDFOR
     \STATE For each $y \in Y$, let $w(y)$ be the number of points $x \in X_j$ with $\tau_\Theta(x,Y) < \Theta$ and such that $y = \argmin_{y' \in Y} \tau_\Theta(x, y')$ with ties handled arbitrarily.  
     Let $X_{out}$ be the set of points $x$ with $\tau(x, Y) = \Theta$. 
    \STATE {\bfseries Return } $(Y, w), |X_{out}|$
\end{algorithmic}
\end{algorithm}

We now prove an analogue of \cref{thm:streaming}. Again, it can be seen as extension of our basic result \cref{thm:aggarwal}. 

\begin{theorem}
\label{thm:kmpipe}
Let $\Theta \in [\OPT/(2\eps),\OPT/\eps]$ be arbitrary.
Suppose we run \cref{alg:kmpipe_seeding} for $t = O(\log (n/\eps))$ rounds with $\ell = O((k/\eps) \log n)$ to obtain output $Y$. Then, with probability at least $1 - 1/\poly(n)$, we have that
\[
\tau_\Theta(X^*_{in}, Y) = O(\OPT). 
\]
\end{theorem}
The following proof is a simple adaptation of the analysis of \kmpipe by Rozhon \cite{rozho2020simple}. 
\begin{proof}
Let $Y_0 \subseteq Y_1 \subseteq \dots \subseteq Y_t$ be the cluster centers that are gradually built by \cref{alg:kmpipe_seeding}. 
We have $Y_0 = \{ y \}$ for a point $y$ picked uniformly at random. Whatever point we picked, we have $\tau(X, Y_0) \le n \Theta = \poly(n/\eps)$, as we assume $\Delta = \poly(n)$. 
We call a cluster $X^*_i$ \emph{unsettled} in iteration $j$ if $\tau(X^*_i, Y_j) > 10\tau(X^*_i, C^*)$. We define $\tau_U(X, Y_j)$ as  
\[
\tau_U(Y, Y_j) = \sum_{\text{$i$, $X^*_i$ unsettled in iter. $j$}} \tau(X^*_i, Y_j). 
\]

We will now prove that with probability $1 - 1/\poly(n)$ we have
\begin{align}
\label{eq:piano_guys}
\tau_U(X, Y_{j+1}) \le \frac{1}{2} \tau_U(X, Y_{j}) + \OPT.
\end{align}

To prove \cref{eq:piano_guys}, fix an iteration $j$. If $\tau_U(X, Y_j) \le \OPT$, the claim certainly holds, so assume the opposite. 
We split unsettled clusters into two groups: a cluster $X^*_i$ is called \emph{heavy} if $\tau(X^*_i, Y_j) \ge \tau_U(X, Y_j) / (2k)$ and \emph{light} otherwise. 
The probability that we make a heavy cluster $X^*_i$ settled can be bounded as follows. 
Sampling a random point from $X^*_i$ proportional to $\tau(\cdot, Y_j)$ makes $X^*_i$ settled with probability at least $1/5$ by \cref{cor:10apx}. 
Hence, there is a subset $Z_i \subseteq X^*_i$ with $\tau(Z_i, Y_j) \ge \tau(X^*_i, Y_j) / 5$ such that sampling a point from $Z_i$ makes $X^*_i$ settled. 
If $Z_i$ contains a point $x$ with $\tau(x, Y_j) \ge \frac{\tau(X, Y_j)}{\ell}$, we sample $x$ with probability $1$ and this also makes $X^*_i$ settled. Otherwise,
\begin{align*}
&\P(X^*_i\text{ does not get settled}) \\
&\le \prod_{x \in Z_i}(1 - \ell \cdot \tau(x, Y_j)/\tau(X, Y_j))\\
&\le \exp(- O(k/\eps \log n) \sum_{x \in Z_i} \tau(x, Y_j)/\tau(X, Y_j) ) && \text{$1+x \le \e^{x}$}\\
&\le \exp(- O(k(\log n)/\eps)  \cdot \frac1{5} \tau(X^*_i, Y_j)/\tau(X, Y_j) )\\
&\leq \exp(- O((\log n)/\eps)  \cdot \frac1{10} \tau_U(X, Y_j)/\tau(X, Y_j) )&&\text{$X_i^*$ is heavy}\\
&\leq \exp \left( - O((\log n)/\eps) \cdot \frac1{10}\frac{ \tau_U(X, Y_j)}{\tau_U(X, Y_j) + 10\OPT + z\Theta} \right) &&\text{each point is unsettled, settled, or an outlier}\\
&\leq \exp \left(- O((\log n)/\eps) \cdot \frac1{10}\frac{ \OPT}{11\OPT + \OPT/\eps} \right)&&\tau_U(X, Y_j)\ge \OPT\\
&= 1 / \poly(n).\\
\end{align*}
Hence, every heavy cluster $X^*_i$ gets settled with probability $1 - 1/\poly(n)$. We now condition on that event.

Note that $\sum_{X_i^*\text{ light}} \tau(X_i^*, Y_j) \le k \cdot \tau_U(X, Y_j)/(2k) = \tau_U(X, Y_j)/2$. 
Hence, we get
\begin{align*}
\tau_U(X, Y_j) - \tau_U(X, Y_{j+1})
\ge \sum_{\text{$X_i^*$ heavy}} \tau(X_i^*, Y_j) 
\ge \tau_U(X, Y_j)/2 
\end{align*}
as needed to finish the proof of \cref{eq:piano_guys}. 
Now we can apply \cref{eq:piano_guys} $t = O(\log (n/\eps))$ times together with $\tau(X, Y_0) \le \poly(n/\eps)$ to conclude that 
\begin{align*}
\tau(X, Y_t)
&\le (1/2)^t \poly(n/\eps) + \OPT \cdot \sum_{j = 1}^t (1/2)^j
= O(\OPT), 
\end{align*}
with probability $1 - 1/\poly(n)$. 
\end{proof}

\subsection{Constructing final clustering}
\label{sec:appendix_distributed_coordinator}

In this subsection, we show how the output of \cref{alg:streaming,alg:kmpipe_seeding}, that is, a weighted set of centers $(Y, w)$, together with the number of found outliers $|X_{out}|$, can be used to get $(O(1), 1+\eps)$-approximation algorithms, via \cref{thm:streaming,thm:kmpipe}. First, in \cref{thm:coreset_formal}, we only show how to get $(O(1/\eps), 1+\eps)$-approximation guarantee by simply running \cref{alg:noisylocalsearch} on the weighted data with the parameter of number of outliers set to roughly $z - |X_{out}|$. The advantage of \cref{thm:coreset_formal} is its simplicity (we implement a variant of it in \cref{sec:paper_experiments}) and speed. 

Next, in \cref{thm:coreset_finer_formal} we refine the reduction to get $(O(1), 1+\eps)$-approximation guarantee. This result is interesting from the theoretical perspective, as we explained in \cref{sec:paper_distributed}. As a subroutine, the coordinator needs to run some $(O(1), 1+\eps)$-approximation algorithm, i.e., linear programming based algorithms \cite{chen2008constant,Krishnaswamy2018}, so the resulting distributed algorithm seems not practical. Also, we need to somewhat refine \cref{alg:streaming,alg:kmpipe_seeding} so that they pass somewhat more information than just $(Y,w)$ and $|X_{out}|$.

\begin{theorem}[Adaptation of \cite{guha2003streaming}]
\label{thm:coreset_formal}
Let $\eps \in [0,1]$ be arbitrary and $\Theta$ be such that $\OPT/(2\eps z) \le \Theta \leq \OPT/(\eps z)$. 
Let $(Y,w)$ denote the weighted point set that \cref{alg:kmpipe_seeding} and \cref{alg:streaming} compute, respectively. Furthermore, assume that $\tau_\Theta(X^*_{in}, Y) \le \alpha \OPT$ in case of \cref{alg:kmpipe_seeding}, and $\sum_j \tau_\Theta(X_j \cap X^*_{in}, Y_j) \le \alpha \OPT, Y = \bigcup Y_j$, in case of $\cref{alg:streaming}$, respectively. 
Now, suppose we then run a $(\beta,1+\eps)$-approximation algorithm $\fA$ for the weighted $k$-means with outliers formulation on the weighted instance $(Y,w)$ with $z^\fA := z - |\Xout| + 2\alpha \eps z$ outliers and let $C$ denote the resulting set of $k$ centers. Then,
\[
\phi^{-(1+5\alpha \eps)z}(X , C) = O((\alpha\beta + 1/\eps) \OPT). 
\]
\end{theorem}

Here, we define the weighted $k$-means with outliers as the following problem: for the weighted input $(Y, w)$, where each $w(y)$ is an integer, we want to choose a set of $k$ centers $C$ and a weight vector $w_{in}(y)$ for any $y$ such that $w_{in}(y)$ is a non-negative integer, for $w_{out}(y) = w(y) - w_{in}(y)$ we have $\sum_{y \in Y} w_{in}(y) \le z^\fA$ and the sum $\sum_{y \in Y} w_{in}(y) \phi(y, C)$ is minimized. 

Note that algorithms for $k$-means with outliers can be seen as algorithms for weighted $k$-means with outliers by viewing weighted points $(y, w(y))$ as $w(y)$ points on the same location (our assumption that pairwise distances of input points are at least one, is in our case only for simplicity of exposition). Moreover, sampling based algorithms such as \cref{alg:localsearch} can be implemented in time proportional not to $\sum_y w(y)$, but $|Y| = \tilde{O}(k)$ by additionally weighting the sampling distribution by the weight of input points. 

Plugging in our \cref{alg:localsearch} as $\fA$ in \cref{thm:coreset_formal} hence yields a distributed algorithm with approximation factor $O(\alpha\beta + 1/\eps) = O(1/\eps + 1/\eps) = O(1/\eps)$ that outputs $(1+O(\eps))z$ outliers. Moreover, the computational complexity of the coordinator is $\tilde{O}(k^2/\eps^2)$ if we use \cref{alg:kmpipe_seeding} and  $\tilde{O}(mk^2/\eps^2)$ if we use \cref{alg:streaming}. 

\begin{proof}
We start by showing that $|\Xin^* \cap \Xout| \leq 2\alpha\eps z$. That is, the number of points that \cref{alg:streaming} and \cref{alg:kmpipe_seeding}, respectively, declare as outliers in the first phase of the algorithm, but that are actually true inliers with respect to a given optimal solution is small. For \cref{alg:kmpipe_seeding}, we have
$$|\Xin^* \cap \Xout| \leq \frac{\tau_\Theta(\Xin^*,Y)}{\Theta} \leq \frac{\alpha \OPT}{\OPT/(2\eps z)} = 2\alpha \eps z,$$
and for \cref{alg:streaming}, we have
$$|\Xin^* \cap \Xout| \leq \sum_{j}\frac{\tau_\Theta(X_j \cap \Xin^*,Y_j)}{\Theta} \leq \frac{\alpha \OPT}{\OPT/(2\eps z)} = 2\alpha \eps z,$$
as desired.
In particular,  $|X_{out}| = |X_{out} \cap X^*_{out}| + |X_{out} \cap X^*_{in}| \le z + 2\alpha\eps z$. Hence, $z^\fA = z - |X_{out}| + 2\alpha\eps z \ge  0$ and therefore algorithm $\fA$ is run on a legal instance.
Let $\Xin := X \setminus \Xout$. Notice that for \cref{alg:kmpipe_seeding}, we have
\begin{align}
\label{eq:pipe}
    \phi(X_{in}^* \cap X_{in}, Y) = \tau_\Theta(X_{in}^* \cap X_{in}) \le \alpha \OPT,
\end{align}
and for \cref{alg:streaming}, we have
\begin{align}
\label{eq:streaming}
    \sum_j\phi(X_{in}^* \cap X_{in} \cap X_j, Y_j) = \sum_j \tau_\Theta(X_{in}^* \cap X_{in} \cap X_j,Y_j) \le \alpha \OPT.
\end{align}

%Recall that we run $\fA$ on $(Y,w)$ with  
%$z^\fA = z - |X_{out}| + \eps z $. 
%By above, we have $z^\fA \ge |X^*_{in} \cap X_{out}| + ?\eps z $. 

%The algorithm groups points from $X^*_{in, \fA}$, each to closest center on a machine, thus forming an instance $(Y, w)$ for $Y \subseteq X$. 
Let us call each set of points $B(y) \subseteq X$ that \cref{alg:streaming} or \cref{alg:kmpipe_seeding} groups around $y \in Y$ a \emph{blob} and for a given $x \in \Xin$, we denote with $y_x$ the point with $x \in B(y_x)$. We define (as the respective algorithm also does) $w(y) = |B(y)|$. 
Moreover, we define $w^*_{in}(y) = |B(y) \cap X^*_{in}|$ and $w^*_{out}(y) = |B(y) \cap X^*_{out}|$. 

%We say a blob $y$ is \emph{outlierish} if at most $\eps$ fraction of grouped vertices are from $X^*_{in}$. Otherwise, a blob is \emph{inlierish}. 
%The intuition is that if we discard outlierish blobs, we discard mostly true outliers and only few more inliers. On the other hand, the cost of inlierish clusters can accounted to inliers in that cluster. As at least $\eps$ fraction of vertices there are inliers (i.e., from $X^*_{in}$), by this we lose $1/\eps$ factor in approximation. 

%Note that the number of points from $X^*_{out}$ in outlierish blobs is bounded by $|X^*_{out}| - |X_{out}| = z - |X_{out}| = z^\fA - 2 \alpha\eps z$. 

Note that splitting the weights $w$ into $w^*_{in}$ and $w^*_{out}$ gives us a natural upper bound on the cost of the optimal solution on the instance $Y$ with $z^\fA$ outliers: On one hand,  we have 
\begin{align}
\label{eq:u2}
\sum_{y \in Y} w^*_{out}(y) 
&= |X^*_{out} \cap X_{in}|\\
&= |X^*_{out}| - |X^*_{out} \cap X_{out}|\nonumber\\
&= z - |X_{out}| + |X_{out} \cap X^*_{in}| \nonumber\\
&\le z - |X_{out}| + 2\alpha\eps z &&  \nonumber\\
&= z^\fA.\nonumber
\end{align}
Thus, if we label the weighted set $(Y, w^*_{out})$ as outliers, at most $z^\fA$ points are labeled as outliers. On the other hand, we will now bound $\phi((Y, w^*_{in}), C^*)$, where $C^*$ is the optimum clustering for the original problem defined on $X$. 

For any $y \in Y$, $x \in X^*_{in} \cap B(y)$, and $c^*_x = \argmin_{c^* \in C^*} \phi(x, c^*)$ we have 
\begin{align}
\label{eq:justin_bieber}    
\phi(y, C^*) 
\le \phi(y, c^*_x)
\le 2\phi(y, x) + 2\phi(x, c^*_x),
\end{align}
where the second inequality is due to \cref{fact:generalised_triangle_inequality}. 
%As there are at least $\eps w(y)$ inliers in the blob $B(y)$, we have 
%\[
%w(y) \cdot \phi(y, C^*)
%\le \frac{1}{\eps} \left( \sum_{x \in X^*_{in} \cap B(y)} \phi(y, C^*) \right)
%\le \frac{1}{\eps} \left( \sum_{x \in X^*_{in} \cap B(y)} 2\phi(y, x) + 2\phi(x, c^*_x) \right)
%\]
Hence, 
\begin{align}
\label{eq:imagine_dragons}
\phi((Y, w^*_{in}), C^*)
&= \sum_{\text{$y \in Y$}} w^*_{in}(y) \cdot \phi(y, C^*)\\
&= \sum_{\text{$y \in Y$}} \sum_{x \in X^*_{in} \cap B(y)} \phi(y, C^*)\nonumber\\
&\le  \sum_{\text{$y \in Y$}}  \sum_{x \in X^*_{in} \cap B(y)} 2\phi(y, x) + 2\phi(x, c^*_x)&&\text{\cref{eq:justin_bieber}}\nonumber\\  
&\le 2 \left( \sum_{x \in X^*_{in} \cap X_{in}} \phi(y_x, x) \right) +2  \left( \sum_{x \in X^*_{in} } \phi(x, c^*_x) \right)\nonumber\\
&\le 2 \alpha \OPT + 2\OPT \nonumber && \text{\cref{eq:pipe,eq:streaming}{}}\\
&= O(\alpha \OPT).  \nonumber
\end{align}

Now we consider the output $C^\fA$ and $w_{out}^\fA$  of $\fA$, i.e., for each $y \in Y$ the algorithm $\fA$ decides which integer weight $w_{out}^\fA$ of $y$ is labeled as outlier. Define $w_{in}^\fA(y) = w(y) - w_{out}^\fA(y)$.  
This solution of instance $(Y,w)$ with $\sum_{y\in Y}w^\fA_{out}(y)$ outliers naturally defines a solution of the original instance $X$ as follows. 
We leave the set of centers $C^\fA$ the same and whenever $w^\fA_{out}(y)$ is nonzero, we label $w^\fA_{out}(y)$ arbitrary points in $B(y)$ as outliers and call this set $X^\fA_{out}$ (note that the algorithm itself does not need to explicitly label the outliers, it suffices to prove that there is a labelling). Then we define points from $X_{out} \cup X^\fA_{out}$ as outliers. 
The number of points we label as outliers is bounded by
\begin{align}
|\Xout| + \sum_{y\in Y}w^\fA_{out}(y) 
&\leq |\Xout| + (1+\eps)z_\fA \\
&= |\Xout| + (1+\eps)(z - |\Xout| + 2\alpha \eps z) \nonumber\\
&\leq (1 + 5 \alpha \eps) z    \nonumber
\end{align}
Thus, it remains to bound the cost of the set $X \setminus (X_{out} \cup  X^\fA_{out}) $ with respect to $C^\fA$. 
We have
\begin{align*}
&\phi(X \setminus (X_{out} \cup X_{out}^\fA), C^\fA)\\
&= \sum_{y \in Y} \sum_{x \in B(y)\setminus \Xout^\fA} \phi(x, C^\fA)\\
&\le 2 \sum_{y \in Y} \sum_{x \in B(y) } \phi(x,y_x) 
+ 2 \sum_{y \in Y}\sum_{x \in B(y) \setminus \Xout^A} \phi(y_x, C^\fA) && \text{\cref{fact:generalised_triangle_inequality}} \\
&= 2 \sum_{y \in Y} \sum_{x \in B(y) \cap \Xin^*} \phi(x,y_x) + 2 \sum_{y \in Y} \sum_{x \in B(y) \cap \Xout^*} \phi(x,y_x) 
\\& + 2 \sum_{y \in Y}\sum_{x \in B(y) \setminus \Xout^A} \phi(y_x, C^\fA)  \\
&\leq 2\alpha \OPT + 2\phi(\Xin \cap \Xout^*,Y) + 2 \sum_{y \in Y} w^\fA_{in} \phi(y, C^\fA) && \text{\cref{eq:pipe}, \cref{eq:streaming}}\\
&\leq 2 \alpha \OPT + 2 z \Theta + 2 \sum_{y \in Y} w^\fA_{in} \phi(y, C^\fA) && |X^*_{out}|=z, x \in X_{in} \Rightarrow \phi(x, Y) \le \Theta \\
&\leq  2\alpha \OPT + 2(1 / \eps)\OPT + 2\beta \phi((Y, w^*_{in}), C^*) && \text{$\fA$ is $\beta$ approximation}\\
&\le O((1/\eps + \alpha\beta) \OPT).&& \text{\cref{eq:imagine_dragons}}  
\end{align*}
\end{proof}

\paragraph{Better construction}
Here we sketch a somewhat more elaborate construction that enables us to lose only a constant factor in approximation, instead of $O(1/\eps)$. Below, we assume existence of a $(\beta, 1+\eps)$-approximation algorithm that works in what we call an \emph{almost-metric space}, which is defined as a classical metric space, but without the axiom $d(x,y)=0 \Leftrightarrow x=y$. That is, even when an algorithm chooses some point $x$ as a cluster center, the clustering cost of $x$ might still be non-zero.
%We actually dropped the assumption $d(x,y) = 0 \Rightarrow x=y$ above when we discussed splitting one weighted point in many. Here we also drop $d(x,x)=0$ axiom, i.e., even if $x$ is chosen in the set of centers $C$, we still need to pay a nonzero cost for it. 
Our algorithms work in this more general setting and we believe that the algorithm of \cite{bhattacharya2019noisy} too. 
In any case, note that the problem of $k$-means with outliers in almost metric space can be reduced to the same problem in metric space by 
%by increasing the number of points by a $2k$-factor, one can create a new metric-space instance, where each solution for the \km with outliers objective in the metric-space naturally maps to a solution in the almost-metric space and vice versa, only loosing a constant-factor in the approximation guarantee. 
splitting each vertex $x$ in $2k$ clones with $x_1, \dots, x_{2k}$ and defining $d'(x_i, x_j) = d(x,x)$. This increases the number of points by a $2k$-factor and we lose only a factor of $2$ in our approximation guarantee. 
As in previous construction, we moreover assume a weighted version of the problem, when for integer weight $w(x)$ one is allowed to output an integer $0 \le w_{out}(x) \le w(x)$ and then only $w(x) - w_{out}(x)$ weight is used in the computation of the cost. 

\begin{theorem}
\label{thm:coreset_finer_formal}
Let $\eps \in (0,1]$ be arbitrary. Suppose that there is a $(\beta, 1+\eps)$-approximation algorithm $\fA$ that solves the weighted variant of $k$-means with outliers in an almost-metric space in polynomial time. Then, there is a distributed polynomial-time $(O(\beta), 1+O(\eps))$-approximation algorithm in the coordinator model with communication $\tilde{O}(k/\eps)$ per site that succeeds with positive constant probability.
\end{theorem}

\begin{proof}
Assume that $\Theta$ satisfies $\frac{\OPT}{2 \eps z} \leq \Theta \leq \frac{\OPT}{\eps z}$. We later discuss how to "guess" $\OPT$ in this distributed setting.
The algorithm needed is a variant \cref{alg:streaming}, so we only describe the change that needs to be done. 
When the variant of \cref{alg:streaming} aggregates all points to the closest center in $Y_j$, it will not just compute for each $y$ its weight $w(y)$, i.e., for how many points in $X_j \setminus \Xout$ the point $y$ is the closest, but it iterates over all these points $x_1, x_2, \dots, x_{w(y)}$ and rounds the distances $d(x_i, y)$ down to the closest power of two. 
Then, the additional information sent to the leader about $y$ is not just $w(y)$, but also the list $w_0(y), w_1(y), w_2(y), \dots$, where $w_k(y)$ is the number of points $x \in B(y)$ with distance $d(x, y)$ rounded down to $2^k$. 
For any $y \in Y_j$ and $k$, we define $B_k(y)$ as the set of points in $X_j \setminus \Xout$ for which $y$ is the closest point in $Y_j$ and $d(x,y)$ is rounded down to $2^k$. As before, for a given $x \in \Xin := X \setminus \Xout$, $y_x$ denotes the point $y$ for which $x \in B(y)$.

The instance $(Y',w')$ the leader is going to construct will not be just $(Y,w)$, but the leader creates a new almost-metric space $\fM'$ that includes the original input metric space $\fM$ that is only defined for points of $Y$ together with the original metric, but moreover, for each $y$ and $0 \le k \le \log\Delta$, it contains a point $y^k$ with  $d_{\fM'}(y, y^k) = 2^k$. Now, imagine a weighted graph with the set of vertices being equal to the points in the almost-metric space and there is an edge between two points if we explicitly stated the distance between the two points and the weight of the edge is equal to that distance. Then, the distance between each pair of vertices is just equal to the length of the shortest path between the two vertices in the weighted graph. In fact, we also define $d_{\fM'}(y^k, y^k) = 2 \cdot 2^k$ (which is not possible in a metric space) and only then the leader runs the $(\beta, 1+\eps)$-approximation algorithm $\fA$ in $\fM'$. We note that the defined distances satisfy the triangle inequality.

As in \cref{thm:coreset_formal}, we assume that $\sum_j \tau_\Theta(X_j \cap \Xin^*,Y_j) \leq \alpha \OPT$. We run $\fA$ on $\fM'$ (we will use $\phi'$ and $d'$ when talking about (squared) distances in $\fM'$) with the number of outliers being set to $z^\fA = z - |X_{out}| + 2\alpha \eps z$ and get as output a set of centers $\fC^\fA$ and a set of outliers $(Y, w^\fA_{out})$ with $\sum_{y\in Y} w^\fA_{out}(y) \le (1+\eps)z^\fA$. 

As in the proof of \cref{thm:coreset_formal}, we  would now like to argue that the optimal solution in the almost-metric space has a cost of at most $O(\alpha \OPT)$ by considering the optimal set of centers $C^*$ for the original problem. However, there is one small technical issue. Namely, the points in $C^*$ might not be contained in the almost-metric space. To remedy this situation, we do the following: We naturally extend the almost-metric space to also include all the points in $C^*$ and prove that the defined distances still satisfy the triangle inequality. Next, we show that in this extended metric space, the cost of the optimal solution (with outliers defined appropriately) has a cost of $O(\alpha \OPT)$. By using the uniform sampling lemma (\cref{lem:2apx_metric}), this implies that there also exists a solution in the original almost-metric space (without the points in $C^*$) with a cost of at most $4 \cdot O(\alpha \OPT) = O(\alpha \OPT)$, as desired. 

We define $w^*_{in}(y^k) = |B_k(y) \cap X^*_{in}|$ and $w^*_{out}(y^k) = |B_k(y) \cap X^*_{out}|$. 
We start by extending the almost-metric space to points in $C^*$. For each $c^* \in C^*$ and $y \in Y$, we define $d_{\fM'}(y,c^*) = d(y,c^*)$. That is, the distance between points in $Y \cup C^*$ is simply equal to the original distance. 
The distance between $c^* \in C^*$ and $y^k$ is defined as $d_{\fM'}(c^*,y^k) = d_{\fM'}(c^*,y) + d_{\fM'}(y, y^k) = d_{\fM'}(c^*,y) + 2^k$. This extension still satisfies the triangle inequality, which more or less directly follows from the fact that the original metric satisfies the triangle inequality.
Next, we define a solution for the weighted \km with $z^{\fA}$ outliers objective in this extended almost metric-space, whose cost gives us an upper bound of $O(\alpha\OPT)$ for the optimal solution. We consider the optimal set of centers $C^*$ together with the set of outliers $\Xout^*$. We set $w^*_{out}(y^k) = |B_k(y) \cap \Xout^*|$ and $w^*_{in}(y^k) = w_k(y) - w^*_{out}(y^k)$. Splitting the weights $w$ into $w^*_{in}$ and $w^*_{out}$ together with the set of centers $C^*$ gives us a natural upper bound (up to a factor of $4$, as mentioned before) on the cost of the optimal solution on the instance $(Y',w')$ with $z^\fA$ outliers. First, one needs to verify that the number of declared outliers $\sum_{y \in Y}\sum_k w^*_{out}(y^k)$ in the solution is at most $z_{\fA}$. This follows in the exact same way as in the proof of \cref{thm:coreset_formal} (cf. \cref{eq:u2}). Thus, it remains to bound $\phi'((Y', w^*_{in}), C^*)$.

Let $x \in B_k(y)$ be arbitrary for some $k$ and $y \in Y$. Moreover, let $c^*_x := \argmin_{c^* \in C^*} \phi(x, c^*)$. We have
\begin{align}
\label{eq:one_direction}    
\phi'(y^k, C^*) 
&\le 2\phi'(y^k, y) + 2\phi'(y, C^*)\\
&= 2\phi'(y^k, y) + 2\phi(y, C^*)\nonumber\\
&\le 2\phi'(y^k, y) + 4\phi(y, x) + 4\phi(x, c^*_x)\nonumber\\
&\le 6\phi(y, x) + 4\phi(x, c^*_x) && \phi'(y^k, y) \le \phi(y, x),\nonumber
\end{align}
where we repeatedly used \cref{fact:generalised_triangle_inequality}. 
%As there are at least $\eps w(y)$ inliers in the blob $B(y)$, we have 
%\[
%w(y) \cdot \phi(y, C^*)
%\le \frac{1}{\eps} \left( \sum_{x \in X^*_{in} \cap B(y)} \phi(y, C^*) \right)
%\le \frac{1}{\eps} \left( \sum_{x \in X^*_{in} \cap B(y)} 2\phi(y, x) + 2\phi(x, c^*_x) \right)
%\]
Hence, 
\begin{align}
\label{eq:imagine_wagons}
\phi'((Y', w^*_{in}), C^*)
&= \sum_{\text{$y^k \in Y'$}} w^*_{in}(y^k) \cdot \phi'(y^k, C^*)\\
&= \sum_{\text{$y^k \in Y'$}} \sum_{x \in X^*_{in} \cap B_k(y)} \phi'(y^k, C^*)\nonumber\\
&\le \sum_{\text{$y^k \in Y'$}} \sum_{x \in X^*_{in} \cap B_k(y)}  6\phi(y, x) + 4\phi(x, c^*_x)&&\text{\cref{eq:one_direction}}\nonumber\\
&\le 6 \left( \sum_{x \in X^*_{in} \cap X_{in}} \phi(y_x, x) \right) + 4  \left( \sum_{x \in X^*_{in} } \phi(x, c^*_x) \right)\nonumber \\
&\le 6 O(\alpha \OPT) + 4\OPT \nonumber\\
&= O(\alpha \OPT). &&  \nonumber
\end{align}

Now we consider the output $C^\fA$ and $w_{out}^\fA$  of $\fA$, i.e., for each $y^k \in Y'$ the algorithm $\fA$ decides which integer weight $w_{out}^\fA(y^k)$ of $y^k$ is labeled as an outlier. Define $w_{in}^\fA(y^k) = w_k(y) - w_{out}^\fA(y^k)$.  
This solution of instance $(Y',w)$ in $\fM'$ with $\sum_{y^k\in Y'}w^\fA_{out}(y^k)$ outliers naturally defines a solution of the original instance $X$ with $|X_{out}| + \sum_{y\in Y}w^\fA_{out}(y)$ outliers: for each center $c \in C^\fA$ in $\fM'$ we define $f(c) \in \fM$ as follows: if $c = y^k$ for some $y$ and $k$, we let $f(c) = y$. Otherwise, $f(c) = y$ for some $y$ and we simply set $f(c) = c$. Now, we choose $f(C^\fA) := \{f(c) \colon c \in C^\fA\} \subseteq Y \subseteq X$ as our set of centers. Moreover, whenever $w^\fA_{out}(y^k)$ is nonzero, we label $w^\fA_{out}(y^k)$ arbitrary points in $B_k(y)$ as outliers and we call the resulting ste of outliers $X^\fA_{out}$. Finally,  we define all the points in $X_{out} \cup X^\fA_{out}$ as outliers. 
%In the same way as in the proof of \cref{thm:coreset_formal}, one can show that the number of declared outliers is at most $(1 + O(\alpha \eps))z$.
Note that we have $\sum_{y^j \in Y'}w_{out}^\fA (y^j) \le (1+\eps)z^\fA$, so 
\begin{align*}
|X_{out} \cup X_{out}^\fA| 
&\le |X_{out}| + (1+\eps) z^\fA = |X_{out}| + (1+\eps)(z - |X_{out}| + \alpha\eps z) \\
&= (1+ O(\alpha\eps)) z
\end{align*}
which means that in total we label only $(1+O(\alpha\eps))z$ vertices as outliers, as desired. 
We now bound the cost of the set $X \setminus (X_{out} \cup  X^\fA_{out}) $ with respect to the set of centers $f(C^\fA)$. 

To that end, we note that for an arbitrary $x \in B_k(y)$ and any $c\in \fM'$ we have $\phi'(y^k, c) = (2^k + d'(y,c))^2$ and therefore
\begin{align}
\label{eq:black_sabbath}
    \phi(x, f(c)) 
    &\le (d(x, y) + d(y, f(c)))^2 \\
    &\le (2^{k+1} + d(y, f(c)))^2 && x \in B_k(y)\nonumber\\ 
    &\le 4(2^k + d(y,f(c)))^2\nonumber \\
    &\le 4(2^k + d'(y,c))^2 &&\text{Definition of $f$ and $d'$}\nonumber\\
    &= 4\phi'(y^k, c).\nonumber
\end{align}

Hence, we have
\begin{align}
\label{eq:tomasklus}
&\phi(X \setminus (X_{out} \cup X_{out}^\fA), f(C^\fA))\\
&= \sum_{y^k \in Y'} \sum_{x \in B_k(y)\setminus \Xout^\fA} \phi(x, f(C^\fA))\nonumber\\
&\le 4 \sum_{y^k \in Y'} \sum_{x \in B_k(y)\setminus \Xout^\fA} \phi'(y^k, C^\fA)&&\text{\cref{eq:black_sabbath}}\nonumber\\
&= 4 \sum_{y^k \in Y'} w^\fA_{in}(y^k) \phi'(y^k, C^\fA)\nonumber\\
&\le 4\beta \cdot 4\phi'((Y', w^*_{in}), C^*) && \text{$\fA$ is $\beta$ apx}\nonumber\\
&= O(\alpha\beta\OPT). && \text{\cref{eq:imagine_wagons}} \nonumber
\end{align}
\end{proof}

\end{appendices}

Note that assuming $ \frac{\OPT}{2 \eps z} \leq \Theta \leq \frac{\OPT}{\eps z}$, \cref{thm:streaming} implies that $\sum_{j} \tau_\Theta(X_j \cap \Xin^*,Y_j) = O(\OPT)$ with positive constant probability and therefore we can assume that $\alpha = O(1)$. This implies that our final solution has a cost of $O(\alpha\beta \OPT) = O(\beta)\OPT$ and moreover the solution outputs at most $(1 + O(\alpha \eps))z = (1 + O(\eps))z$ outliers, as desired. What remains to be discussed is how to remove the assumption that the algorithm knows $\OPT$. To that end, we again run the algorithm for $O(\log(n\Delta))$ guesses of $\OPT$, namely for all values $2^e$ with $e \in [\log(n\Delta)]$ in parallel as before. 
Each machine will send all $O(\log n)$ respective weighted sets $(Y', w)$ and $|X_{out, j}|$ to the coordinator that outputs the set of candidate centers that minimizes $\phi'((Y', w^\fA_{in}), C^\fA)$ among those where $X_{out}$ satisfies $|X_{out}| \le z$. \cref{eq:tomasklus} certifies that, with positive constant probability, we get an $O(\beta)$-approximation of $\OPT$, as needed. 

\end{document}